\documentclass[preprint]{imsart}
\usepackage{fullpage}
\usepackage[utf8]{inputenc}
\usepackage{amsthm,amsfonts,amssymb,graphicx,bm}
\usepackage[reqno]{amsmath}
\usepackage{algorithm2e}
\usepackage{framed}
\usepackage{color}
\usepackage{natbib}
\usepackage{url} 
\allowdisplaybreaks

\startlocaldefs
\newtheorem{definition}{Definition}
\newtheorem{lemma}{Lemma}
\newtheorem{theorem}{Theorem}
\newtheorem{corollary}{Corollary}
\newtheorem{example}{Example}
\newtheorem{remark}{Remark}
\newtheorem{assumption}{Assumption}

\newcommand{\Ep}{\mathbb{E}}
\renewcommand{\tilde}{\widetilde}
\renewcommand{\hat}{\widehat}
\endlocaldefs

\begin{document}

\begin{frontmatter}
\title{Minimum information dependence modeling}
\runtitle{Minimum information dependence modeling}

\begin{aug}
\author[A]{\fnms{Tomonari}~\snm{Sei}\ead[label=e1]{sei@mist.i.u-tokyo.ac.jp}}
\author[B]{\fnms{Keisuke}~\snm{Yano}\ead[label=e2]{yano@ism.ac.jp}}
%%%%%%%%%%%%%%%%%%%%%%%%%%%%%%%%%%%%%%%%%%%%%%
%% Addresses                                %%
%%%%%%%%%%%%%%%%%%%%%%%%%%%%%%%%%%%%%%%%%%%%%%
\address[A]{The University of Tokyo, 7-3-1 Hongo, Bunkyo-ku, Tokyo, 113-8656, Japan\printead[presep={,\ }]{e1}}

\address[B]{The Institute of Statistical Mathematics, 10-3 Midori cho, Tachikawa City, Tokyo, 190-8562, Japan\printead[presep={,\ }]{e2}}
\end{aug}

\begin{abstract}
We propose a method to construct a joint statistical model
for mixed-domain data to analyze their dependence.
Multivariate Gaussian and log-linear models are particular examples of the proposed model.
It is shown that the functional equation defining the model has a unique solution under fairly weak conditions.
The model is characterized by two orthogonal parameters:
the dependence parameter and the marginal parameter.
To estimate the dependence parameter, a conditional inference together with a sampling procedure is proposed and is shown to provide a consistent estimator.
Illustrative examples of data analyses involving penguins and earthquakes are presented.
\end{abstract}


\end{frontmatter}

\section{Introduction} \label{section:introduction}

In multivariate analysis, there are a lot of statistical models describing dependence such as copula models, regression models, log-linear models and Gaussian graphical models. Such models are quite powerful and frequently used in applications. However, these approaches  depend to varying degrees on the characteristics of the data domain. For example, in copula modeling, the variables are assumed to be real-valued and transformed into $[0,1]$ by monotone transformation. In generalized linear models, the domains of explanatory variables are arbitrary via quantification, but the conditional density of the response variable has to be specified as Gaussian or Poisson, for example. A log-linear model assumes the domain to be discrete.

In this paper, we propose a method to construct dependence models without using domain characteristics. The joint density function that we suppose takes the form
\begin{align}
 p(x_1,x_2;\theta) = e^{\theta^\top h(x_1,x_2)}A_1(x_1;\theta)A_2(x_2;\theta)
 \label{eq:min-info-2dim}
\end{align}
in bivariate cases, and is defined similarly in multivariate cases. Here, the domains of $x_1$ and $x_2$ are arbitrary as long as they have base measures. The function $h(x_1,x_2)$ represents the dependence of the variables. 
The functions $A_1$ and $A_2$ are determined by the marginal distributions of $x_1$ and $x_2$.
It is proved that a density function satisfying the marginal constraints exists and is unique for every value of the parameter $\theta$ under fairly weak conditions (Theorem~\ref{theorem:feasible}).
This makes the model quite flexible since the parametric or nonparametric forms of the marginal distributions can be separately designed. We will call $A_1$ and $A_2$ adjusting functions.
See Section~\ref{section:model} for a precise definition.

We show the usefulness of our model by enumerating various examples.
In particular, models for continuous, discrete, and any other type of variables can be jointly designed up to the same cost as homogeneous data.
There has been a similar attempt to construct a mixed-variable model based on univariate conditional exponential families (\cite{yang2015graphical}), which, however, does not avoid the restriction on the parameter space of joint exponential families.
The conditional Gaussian families for mixed data discussed in \cite{Lauritzen} and \cite{Whittaker1990} are tractable but restrictive.

The adjusting functions in the proposed model (\ref{eq:min-info-2dim}) cannot be written in a closed form except for limited cases.
This means a naive likelihood analysis is intractable. 
However, we can perform conditional inference \citep{CoxHinkley,reid1995roles} given the marginal empirical distribution, as described in Section~\ref{section:conditional}.
It is shown that the conditional likelihood has almost the same information as the full likelihood and is not affected by nuisance parameters that characterize the marginal distribution (Theorem \ref{theorem:negligible}).
These aspects are similar to Fisher's exact test for contingency tables (e.g.,\ \cite{choi2015elucidating,Little1989}), which fixes marginal frequencies for testing of independence.
Given the marginal empirical distribution, the problem of estimation reduces to that of an exponential family on a set of permutations. 
See \cite{Mukherjee2016} for properties of exponential families on permutations.

A closed form for the conditional likelihood remains unknown.
To deal with this, we propose a sampling method for the conditional distribution in a Markov chain Monte Carlo (MCMC) manner (Subsection \ref{subsection:sampling})
and 
a pseudo likelihood method
(Subsection \ref{subsection:Besag})
that make the conditional inference tractable.

The idea behind our method is the minimum information copula model (e.g.,\ \cite{Bedford_et_al2016,BedfordWilson2014,MeeuwissenBedford1997,piantadosi2012copulas}), in which the joint density function of real-valued data is determined by uniform marginals and fixed values of expectations of some statistics.
In a study related to the minimum information copula model, \cite{Geenens2020} constructed a bivariate discrete model written as
$p(x_1,x_2)=c(x_1,x_2)A_1(x_1)A_2(x_2)$,
where the marginal distributions of $c(x_1,x_2)$ are discrete uniform distributions.

The paper is organized as follows.
In Section~\ref{section:model}, we define our model together with practical examples and establish an existence and uniqueness theorem (Theorem \ref{theorem:feasible}). In Section~\ref{section:conditional}, we develop a conditional inference approach for the dependence parameter and prove its validity under mild conditions (Theorems \ref{theorem:negligible}--\ref{theorem:consistencyofPL}).
In Section~\ref{section:numerical}, we present simulation studies for the inference. We conclude the paper in Section~\ref{section:discussion} with some future work. 
In Appendix A, we provide more examples of our model.
In Appendix B,
we summarize useful properties of the model, including information geometry (e.g.,\ \cite{AmariNagaoka2000,Csiszar1975}), and relationship to the optimal transport and  Schr\"{o}dinger problems (e.g.,\ \cite{haasler2021,Leonard2012,PeyreCuturi2019}).
Appendix C gives all proofs of the results.
Appendix D provides illustrative examples of data analyses involving penguins and earthquakes are presented.

\section{Minimum information dependence model}\label{section:model}

In this section, 
we introduce the minimum information dependence model with 
its existence guarantee and present several examples.

\subsection{Definition}

Let $(\mathcal{X}_i,\mathcal{F}(\mathcal{X}_i),{\rm d}x_i)$ for $i=1,\ldots,d$ be a measure space
and denote their product space by $\mathcal{X}=\prod_{i=1}^{d}\mathcal{X}_i$ and ${\rm d}x=\prod_{i=1}^{d}{\rm d}x_i$.
For index $i$, use the notation $-i$ to indicate the removal of the $i$-th coordinate, e.g.,
$x_{-i}=(x_j)_{j\neq i}$, $\mathcal{X}_{-i}=\prod_{j\neq i}\mathcal{X}_j$, and ${\rm d}x_{-i}=\prod_{j\neq i}{\rm d}x_j$.

Let $r_1(x_1;\nu),\ldots,r_d(x_d;\nu)$ be statistical models of marginal densities on $\mathcal{X}_1,\ldots,\mathcal{X}_d$, respectively, where $\nu$ denotes parameters characterizing the marginal densities.
We can assign, if necessary, independent parameters to each $r_i$ as $r_i(x_i;\nu_i)$ by setting $\nu=(\nu_1,\ldots,\nu_d)$.
It is also possible to deal with infinite-dimensional parameters.

We consider a class of probability density functions
\begin{align}
& p(x;\theta,\nu) = \exp\left(\theta^\top h(x)-\sum_{i=1}^d a_i(x_i;\theta,\nu)-\psi(\theta,\nu)\right)
\prod_{i=1}^d r_i(x_i;\nu),
\label{eq:min-info}
\end{align}
where $\theta\in\mathbb{R}^K$ is a $K$-dimensional parameter representing the dependence, and
$h:\mathcal{X}\to\mathbb{R}^K$ is a given function.
The functions $a_i(x_i;\theta,\nu)$ and $\psi(\theta,\nu)$ are simultaneously determined by constraints
\begin{align}
 &\int p(x;\theta,\nu){\rm d}x_{-i} = r_i(x_i;\nu), \quad i=1,\ldots,d,\ \text{and}
 \label{eq:marginal-condition}
 \\
 & \int \sum_{i=1}^d a_i(x_i;\theta,\nu)p(x;\theta,\nu){\rm d}x = 0.
 \label{eq:zero-mean-condition}
\end{align}
The equation (\ref{eq:marginal-condition}), which specifies the marginal distributions, is essential in our modeling.
The equation (\ref{eq:zero-mean-condition}) is assumed just for identifiability of $\psi(\theta,\nu)$,
because
for fixed $a_{i}(x_{i};\theta,\nu)$s and fixed $\psi(\theta,\nu)$ and for any $c\in\mathbb{R}$,
$(a_{i}(x_{i};\theta,\nu)+c)$s and $(\psi(\theta,\nu)-dc)$ yield the same probability density function as that with $a_{i}(x_{i};\theta,\nu)$s and $\psi(\theta,\nu)$.
With these constraints,
 we shall see that the functions $\sum_{i=1}^d a_i(x_i;\theta,\nu)$ and $\psi(\theta,\nu)$ are unique if they exist.
If each term of (\ref{eq:zero-mean-condition}) is integrable, the equation (\ref{eq:zero-mean-condition}) is equivalent to the equation 
$\sum_{i=1}^d\int a_i(x_i;\theta,\nu)r_i(x_i;\nu){\rm d}x_i=0$
that does not involve $p(x;\theta,\nu)$,
under the marginal condition (\ref{eq:marginal-condition}).
Note that the density (\ref{eq:min-info}) is reduced to the independent model $\prod_{i=1}^{d} r_i(x_i;\nu)$ if $\theta=0$.

\begin{definition}
 A statistical model (\ref{eq:min-info}) together with the constraints (\ref{eq:marginal-condition}) and (\ref{eq:zero-mean-condition}) is called a {\em minimum information dependence model}.
The parameter $\theta$ is called the {\em canonical parameter},
$\nu$ is the {\em marginal parameter},
$h(x)$ comprises the {\em canonical statistics},
$a_i(x_i;\theta,\nu)$s are the {\em 
adjusting functions}
and $\psi(\theta,\nu)$ is the {\em potential function}.
\end{definition}

Throughout the paper, we assume that the canonical statistics $h_k(x)$, $k=1,\ldots,K$, are linearly independent modulo additive functions. 
That is, if $\theta$ satisfies
\[\theta^\top h(x)+\sum_{i=1}^{d} A_i (x_i)=0,
\]
with $A_{i}(x_{i})$ not depending on $x_{-i}$ ($i=1,\ldots,d$),
then $\theta=0$.

The canonical statistics are not sufficient statistics for $\theta$ in the full likelihood because the 
adjusting
functions contain $\theta$ and $x$,
but are sufficient statistics in the conditional likelihood; this point will be clarified in Section~\ref{section:conditional}.

We state several useful properties of the minimum information dependence model.
First, the derivative of the potential function with respect to $\theta$ is $\Ep[h(X)]$.
Second, the potential function $\psi(\theta,\nu)$ is shown to be  strictly convex with respect to $\theta$.
Third, the value of 
$\theta$ is directly linked to the total correlation \citep{Watanabe1960}
$\Ep[\log \{p(X;\theta,\nu)/\prod_{j=1}^{d}r_{j}(X_{j};\nu)\}
]$, 
a measure of strength of the association, as
\begin{align*}
\Ep\left[\log \frac{p(X;\theta,\nu)}{\prod_{j=1}^{d}r_{j}(X_{j};\nu)}
\right]
=\theta^{\top}\nabla_{\theta}\psi(\theta,\nu)-\psi(\theta,\nu),
\end{align*}
where this follows from the equation (\ref{eq:zero-mean-condition}) and Lemma S.2 in the appendix. 
Finally, the parameters $\theta$ and $\nu$ are mutually orthogonal with respect to the Fisher information metric;
see Appendices B.1 and B.2 in the appendix for details of these properties.

\begin{figure}[tb]
\centering
\includegraphics[width=14cm]{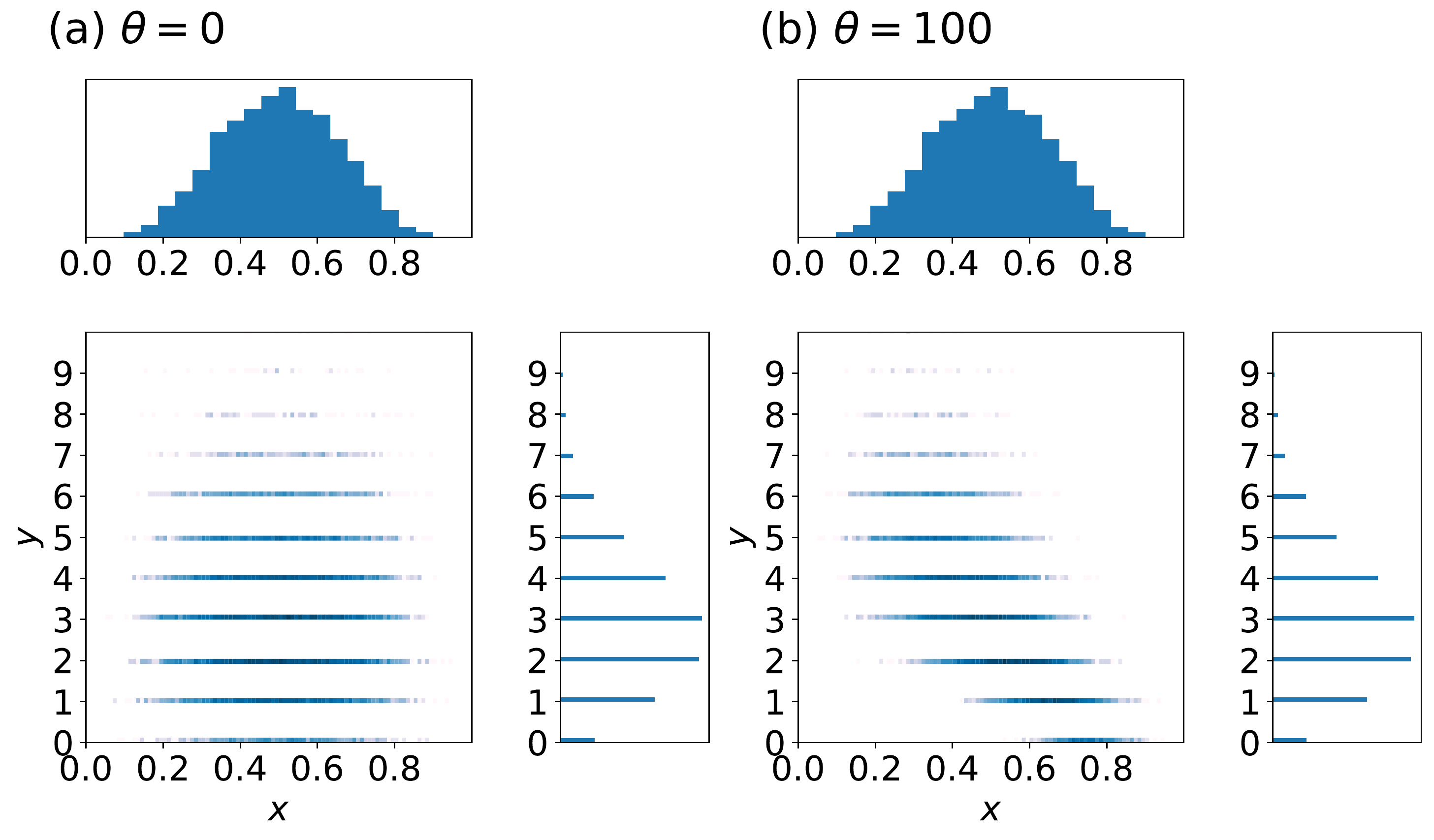}
\caption{Two-dimensional histograms of $10000$ samples from the minimum information dependence model with 
$\mathrm{Beta}(10,10)$ and $\mathrm{Po}(3)$ marginals. The canonical statistic $h(x,y)$ is given by $h(x,y)=x/(y+1)$. The joint histogram and marginal histograms are plotted.
(a) Joint histogram with $\theta=0$. (b) Joint histogram with $\theta=100$.}
\label{fig:poissonbeta}
\end{figure}

As a final remark in this subsection,
our model is a generalization of the minimum information copula model proposed by \cite{BedfordWilson2014},
particular cases of which appear in \cite{MeeuwissenBedford1997} and \cite{piantadosi2012copulas}.
In this copula model, the space $\mathcal{X}_i$ is the interval $[0,1]\subset\mathbb{R}$ and the marginal density functions $r_i(x_i)$ are assumed to be the uniform density on $[0,1]$.
This copula model is derived from the maximum entropy principle (\cite{Jaynes1957}), or equivalently, from the minimum information principle, which is the origin of the name.
We discuss a similar property for our model in Appendix B in the appendix.
Yet, copula models are, in general, intended to be used together with the probability integral transform that makes the variables have uniform marginal distributions. In contrast, our model specifies the marginal model without transforming the variables, which allows arbitrary $\mathcal{X}_i$ rather than $\mathbb{R}$.
Figure \ref{fig:poissonbeta} displays an example of two-dimensional histograms of samples from the minimum information dependence model for mixed variables (discrete and $[0,1]$) with negative correlation.
This example shows that the minimum information dependence model easily expresses arbitrary dependence between variables in arbitrary product spaces.

Also, the minimum information dependence model naturally contains the existing dependence models for specific sample spaces (\cite{Amari2001,HollandWang1987,Jansen1997,Jones_et_al2015}).
Several examples highlighting the connection between the minimum information dependence model and the existing models are presented in Subsections~\ref{subsection:illustrativeexamples} and \ref{subsection:examples}, and Appendix A in the appendix.
Remark \ref{remark:diff between copula and min-info} discusses the difference between the minimum information dependence models and copula models.

\subsection{An illustrative example}
\label{subsection:illustrativeexamples}

We give an elementary example for illustrative purposes. 
More practical examples are provided in Subsection~\ref{subsection:examples}.
Further examples
related to the total positivity (e.g.,\ \cite{HollandWang1987,Kurowicka2015}) and circulas (\cite{Jones_et_al2015})
are presented in Appendix A in the appendix.

\begin{example}[Gaussian model;  \cite{Jansen1997}]
\label{example: Jansen Gaussian}
 The bivariate Gaussian distribution is a minimum information dependence model.
 This fact was pointed out by \cite{Jansen1997} in an argument of the copula theory.
 Suppose that the mean vector is zero for simplicity.
 If the variance and correlation parameters are denoted as $\sigma_1^2,\sigma_2^2$ and $\rho$, respectively, then the density is written as
 \begin{align*}
  p(x_1,x_2;\sigma_1,\sigma_2,\rho)
   &= e^{\theta x_1x_2 - a_1(x_1;\theta,\nu) - a_2(x_2;\theta,\nu) - \psi(\theta,\nu)}\phi(x_1;\sigma_1^2)\phi(x_2;\sigma_2^2),
 \end{align*}
 where $\theta=(\sigma_1\sigma_2)^{-1}\rho(1-\rho^2)^{-1}$ is the canonical parameter, $x_1x_2$ is the canonical statistic, $\nu=(\sigma_1^2,\sigma_2^2)$ is the marginal parameter and $\phi(x_i;\sigma_i^2)$ denotes the univariate normal density. The 
 adjusting
 functions and the potential function are
 \begin{align*}
 a_i(x_i;\theta,\nu) &= \left(\frac{1}{2(1-\rho^2)}-\frac{1}{2}\right)\left(\left(\frac{x_i}{\sigma_i}\right)^2-1\right),\quad i=1,2,
 \\
 \psi(\theta,\nu) &= \frac{1}{2}\log(1-\rho^2) +\frac{1}{1-\rho^2}-1,
 \end{align*}
 respectively, where $\rho=\rho(\theta,\nu)$ is the unique solution of $\theta=(\sigma_1\sigma_2)^{-1}\rho(1-\rho^2)^{-1}$.
 More explicitly,
 $\rho = \rho(\theta,\nu) = 2\theta\sigma_1\sigma_2 / \{1+(1+4\theta^2\sigma_1^2\sigma_2^2)^{1/2}\}.$
 The canonical parameter $\theta$ ranges over $\mathbb{R}$ whereas $\rho$ ranges over $(-1,1)$.
 The one-to-one relationship between $\rho$ and $\theta$ for a given $\nu$ is a consequence of a more general result (see Appendix B.1 in the appendix).
 We will show that higher-dimensional Gaussian models also have a similar structure in Subsection~\ref{subsection:examples}.
  \end{example}

\begin{remark}\label{remark:diff between copula and min-info}
    Here we clarify the difference between minimum information dependence models and copula models.
    Assume that $\mathcal{X}_{i}=\mathbb{R}$ for $i=1,\ldots,d$.
    A copula model is built by using the Sklar theorem and the change of variables as
    \begin{align*}
    p_{\mathrm{copula}}(x;c,\nu)=c(R_{1}(x_{1};\nu),\ldots,R_{d}(x_{d};\nu))
    \prod_{i=1}^{d}r_{i}(x_{i};\nu),
    \end{align*}
    where
    $c(\cdot)$ is a copula density, 
    $R_{i}(\cdot;\nu)$ $(i=1,\ldots,d)$ is the distribution function of the $i$-th variable for $i=1,\ldots,d$,
    and 
    $r_{i}(\cdot;\nu)$ $(i=1,\ldots,d)$ is the density function of the $i$-th variable.
    In contrast,
    the minimum information dependence model is built as
    \begin{align*}    p(x;\theta,\nu)=\exp\left(\theta^{\top}h(x)
    -\sum_{i=1}^{d}a_{i}(x_{i};\theta,\nu)
    -\psi(\theta,\nu)
    \right)
    \prod_{i=1}^{d}r_{i}(x_{i};\nu).
    \end{align*}
    Two models have the form of
    the dependence part and the product of marginals. 
    But, 
    the dependence part in the copula model has the composite form $c \circ R$
    , whereas that of our model has 
    the exponential of the additive form:
    the term $\theta^{\top}h(x)$ independent of marginals plus the terms $-\sum_{i=1}^{d}a_{i}(x_{i};\theta,\nu)-\psi(\theta,\nu)$ depending on the marginals.
    So, the coincidence between copula and minimum information dependence models  depends on the marginal distributions.

    One example of the coincidence is a Gaussian model (Gaussian copula with Gaussian marginals).
    Another example is the bi-variate minimum information dependence copula with uniform marginals. 
    In contrast, 
    one example of the discordance
    is the Farlie--Gumbel--Morgenstern copula with Gaussian marginals that is included in copula models but is excluded in our models.
\end{remark}

In Example \ref{example: Jansen Gaussian}, 
the 
adjusting  
functions 
and potential functions are explicitly represented.
However, such cases are exceptional.
To generate a rich class of distributions, we establish an existence and uniqueness theorem for the functions in the following subsection.

\subsection{Existence and uniqueness} \label{subsection:feasible}

We find tractable conditions for the existence and uniqueness of the 
adjusting 
functions and potential function.
Fix $\theta$ and $\nu$ and drop the dependence on these parameters in this subsection.
Then, the inner product $\theta^\top h(x)$ is replaced with a scalar function $H(x)=\theta^\top h(x)$.
The marginal parameter $\nu$ is also abbreviated as $r_i(x_i)$.

Denote the product density of $r_i(x_i)$ by
$p_0(x)=\prod_{i=1}^d r_i(x_i)$
that corresponds to the density function for $H=0$.
Denote the set of integrable functions with respect to a measure $\mu$ by $L_1(\mu)$.

\begin{definition}
We say that a function $H\in L_1(p_0(x){\rm d}x)$ is {\em feasible} if there exist measurable functions $\{a_i(x_i):i=1,\ldots,d\}$ and a real number $\psi\in\mathbb{R}$ such that the function
$p(x) = e^{H(x)-\sum_{i=1}^d a_i(x_i)-\psi}p_0(x)$
satisfies
$\int p(x){\rm d}x_{-i} = r_i(x_i)$ for each  $i=1,\ldots,d$ and $\int \sum_{i=1}^{d} a_i(x_i)p(x){\rm d}x=0$.
\end{definition}
\begin{definition}
We say that $H$ is {\em strongly feasible} if $H$ is feasible and each $a_i$ belongs to $L_1(r_i(x_i){\rm d}x_i)$.
\end{definition}
\begin{definition}
\label{def: moderately feasible}
We say that $H$ is {\em moderately feasible} if 
there exist $\{b_i\in L_1(r_i(x_i){\rm d}x_i): i=1,\ldots,d\}$ such that
 \begin{align}
  \int e^{H(x)-\sum_{i=1}^{d} b_i(x_i)} p_0(x)
 {\rm d}x < \infty.
 \label{eq:feasible-equivalent}
 \end{align}
\end{definition}

The only difference between feasibility and strong feasibility is the integrability of $a_i(x_i)$. Strong feasibility is convenient for theoretical discussions whereas feasibility is sufficient for the modeling and the inference.
Strong feasibility implies moderate feasibility by Definition \ref{def: moderately feasible}.
If $H$ is feasible and the 
adjusting 
functions are integrable, then $H$ is strongly feasible.

Our first main result is that moderate feasibility is a sufficient condition of feasibility.
The proof, which is based on the results of \cite{Borwein_et_al1994}, is given in Appendix C.1 in the appendix.

\begin{theorem} \label{theorem:feasible}
If a function $H\in L_1(p_0(x){\rm d}x)$ is moderately feasible,
then $H$ is feasible.
 Furthermore, if $H$ is feasible, then $\sum_{i=1}^d a_i(x_i)$ and $\psi$ are unique.
 \end{theorem}

It is usually easy to find integrable functions $b_i$s satisfying the inequality (\ref{eq:feasible-equivalent});
we provide concrete examples in Subsection~\ref{subsection:examples}.

For moderate feasibility, we obtain a couple of corollaries.
First, by applying H\"older's inequality to the integral in (\ref{eq:feasible-equivalent}), we immediately obtain the following corollary.

\begin{corollary}
 The set of moderately feasible functions is convex.
\end{corollary}

Further, we provide the following two useful criteria for moderately feasibility.
The proofs are given in Appendix C.2.

\begin{corollary} \label{corollary:polynomial}
Suppose that 
for each $i=1,\ldots,d$,
the space 
$\mathcal{X}_i$ is equal to $\mathbb{R}$, and the marginal density function $r_i$ has finite moments of any order.
Then, any polynomial function $H(x)$ is moderately feasible.
\end{corollary}

\begin{corollary}
\label{corollary:Lipschitz}
 Suppose that 
 for each $i=1,\ldots,d$,
 $(\mathcal{X}_{i},d_i)$ is a metric space
 and there exists $x^{0}_{i}\in\mathcal{X}_{i}$ for which
 $d_{i}(x_{i},x^{0}_{i})\in L_{1}(r_{i}(x_{i})\mathrm{d}x_{i})$.
 Then, any Lipschitz function $H(x)$ with respect to  $d_{(p)}(x,y):=\|(d_{i}(x_{i},y_{i}))_{i=1,\ldots,d}\|_{p}$ for some $p\in[1,\infty]$ is moderately feasible,
 where $\|\cdot\|_{p}$ is the $p$-norm.
\end{corollary}

\subsection{Practical examples} \label{subsection:examples}

In this subsection, 
we provide various examples of the minimum information dependence model by applying Theorem~\ref{theorem:feasible}.
Denote the canonical part of the model by $H_\theta(x)=\theta^\top h(x)$.

\begin{example}[Gaussian] \label{example:Gaussian}
 For each $i=1,\ldots,d$,
 let $r_i (x_{i};\nu)$ be the Gaussian density $\phi(x_i;\mu_i,\sigma_i^2)$, where $\nu=(\mu_{1},\ldots,\mu_{d},\sigma_1,\ldots,\sigma_{d})$.
 Then, a quadratic function
 $ H_\theta(x)=\sum_{i<j}\theta_{ij}x_ix_j $
 is feasible for any real vector $\theta=(\theta_{ij})\in\mathbb{R}^{d(d-1)/2}$ because the condition of Corollary~\ref{corollary:polynomial} is satisfied.
 The obtained model $p_\theta(x)$ is simply the multivariate normal density.
 Indeed, there exists a unique positive definite matrix $\Sigma=(\sigma_{ij})$ such that $\sigma_{ii}=\sigma_i^2$ and $(\Sigma^{-1})_{ij}=\theta_{ij}$ (\cite{dempster1972covariance}).
 We also point out that the covariance selection model (Gaussian graphical model) is specified by the set of edges $(i,j)$ such that $\theta_{ij}=0$.
 See \cite{Lauritzen} and \cite{Whittaker1990} for details of the covariance selection model.
\end{example}

The following example deals with three-dimensional interaction. 
We emphasize that it is not easy to construct such a model if we use exponential families.

\begin{example}[Three-dimensional interaction]
\label{example: 3-dim interaction}
 Let $d=3$ and $r_i$ be the standard normal density for $i=1,2,3$ and define
 $ H_\theta(x_1,x_2,x_3) = \theta x_1x_2x_3$,
 where $\theta$ is a real parameter. The function is feasible by Corollary~\ref{corollary:polynomial}.
 The obtained distribution can describe Simpson's paradox for continuous variables, that is, the conditional correlation coefficient between the first and second variables given the third variable depends on the value of the third variable; Gaussian distributions do not have this property.
 As $\theta\to\infty$, the distribution tends to a distribution supported on $\{x\in\mathbb{R}^3\mid |x_1|=|x_2|=|x_3|,\ x_1x_2x_3\geq 0\}$ that is the optimal coupling between the three marginal distributions:
 \begin{align*}
 &\mathrm{Minimize} \,\, -\int x_{1}x_{2}x_{3}\, p(x_{1},x_{2},x_{3})\mathrm{d}x
 \\
 &\mathrm{subject \,\, to} \,\, p \,\,\mathrm{with \,\, marginals \,\, equal \,\, to}\,\, r_{i}\text{s}.
 \end{align*}
 We can see this convergence in more detail through the relationship with the entropic optimal transport problem discussed in Appendix B.3 in the appendix.
\end{example}

\begin{example}[count data] \label{example:count}
 Suppose that $(\mathcal{X}_i,{\rm d}x_i)$ for $i=1,2$ is the set of non-negative integers with the counting measure.
 For each $i$, let $r_i(x_i;\nu_i)=(\nu_i^{x_i}/x_i!)e^{-\nu_i}$ be the Poisson distribution with mean $\nu_{i}>0$ and define $H_\theta(x_1,x_2)=\theta x_1x_2$ for $\theta\in\mathbb{R}$.
 Then we can see that $H_\theta$ is feasible for any $\theta\in\mathbb{R}$ from Corollary~\ref{corollary:polynomial}.
 The parameter $\theta$ controls the correlation between the two variables.
 The range of the Pearson correlation coefficient is a proper subset of $(-1,1)$ that depends on $\nu=(\nu_1,\nu_2)$.
 Table~\ref{tab:Poisson} shows numerical evaluation of the lower and upper bounds of the correlation coefficient for several values of $\nu_2$ with $\nu_1$ set to $1$, where the Sinkhorn--Knopp algorithm (\cite{SinkhornKnopp1967}) is applied for this evaluation; see Appendix B.3 in the appendix for details of the computation.
\cite{Geenens2020} discusses another construction of bivariate distributions with Poisson marginals and a negative association.

\begin{table}[tb]
    \caption{\label{tab:Poisson}Ranges of the Pearson correlation coefficient for the Poisson marginal model. Parameter $\nu_1$ is set to 1.}
  \centering
\fbox{%
\begin{tabular}{c|ccccc}
\multicolumn{1}{}{}& \multicolumn{5}{c}{$\nu_2$}\\
        & 0.25& 0.5& 1& 2& 4 \\
        \hline
        upper bound  & 0.82&   0.87& 0.99& 0.94& 0.93\\
        lower bound  & $-0.50$& $-0.67$& $-0.74$& $-0.81$& $-0.87$
\end{tabular}}
\end{table}

 \end{example}

\begin{example}[mixed variables] \label{example:mixed}
In our model,
the domains $\mathcal{X}_1,\ldots,\mathcal{X}_d$ are not necessarily identical.
 Suppose that $d=2$, $\mathcal{X}_1=\mathbb{R}$ and $\mathcal{X}_2=\{0,1\}$ for simplicity.
 (The argument presented herein is easily generalized to other domains and higher-dimensional cases.)
 Then, the minimum information dependence model is
 $p(x_1,x_2;\theta,\nu)
  = \exp(\theta^\top h(x_1,x_2) - \tilde{a}_1(x_1;\theta,\nu) - \tilde{a}_2(x_2;\theta,\nu) - \psi(\theta,\nu))$,
 where $\tilde{a}_i(x_i;\theta,\nu)=a_i(x_i;\theta,\nu)-\log r_i(x_i;\nu)$.
 The interaction between the continuous and discrete variables is described by the canonical statistic $h(x_1,x_2)$.

 The model is interpreted as regression models in two ways as follows.
 First, let $x_1$ be the explanatory variable and $x_2$ be the response variable. The conditional density of $x_2=1$ given $x_1$ is written as
 \begin{align*}
  p(x_2=1|x_1;\theta,\nu)
  &= \frac{1}{1+\exp(- \theta^\top u(x_1)-\alpha(\theta,\nu))},
 \end{align*}
 where $u(x_1)=h(x_1,1)-h(x_1,0)$ and $\alpha(\theta,\nu)=-\tilde{a}_2(1;\theta,\nu)+\tilde{a}_2(0;\theta,\nu)$.
 This is the logistic regression model, except that the intercept term $\alpha(\theta,\nu)$ depends on the other regression coefficient $\theta$. Since $\nu$ is a nuisance parameter, we can treat $\alpha(\theta,\nu)$ as a nuisance parameter. More precisely, it is shown that there is a one-to-one correspondence between $\alpha(\theta,\nu)$ and $r_2(1;\nu)$, given $\theta$, by Theorem~\ref{theorem:feasible}.
 Thus, the proposed model is equivalent to the logistic model in this sense, and
 the difference is which of $\alpha$ and $r_2$ is specified first.

 Next, let $x_1$ be the response variable and $x_2$ be the explanatory variable. The conditional density of $x_1$ given $x_2$ is
 \[
 p(x_1|x_2;\theta,\nu) = \exp(\theta^\top h(x_1,x_2)-\psi_2(\theta,\nu|x_2))m(x_{1};\theta,\nu),
 \]
 where $m(x_{1};\theta,\nu):=e^{-a_{1}(x_{1};\theta,\nu)}r_{1}(x_{1};\nu)$ is the base measure independent from $x_{2}$, and
 $\psi_2(\theta,\nu|x_2)$ is the normalizing constant that makes $\int p(x_1|x_2){\rm d}x_1=1$.
 This is just a generalized linear model except that the base measure depends on the regression coefficient $\theta$. From Theorem~\ref{theorem:feasible}, the base measure has a one-to-one correspondence with the marginal density for a given $\theta$.

 An important consequence here is that the two regression models of opposite direction are derived from a common joint density function. This is in contrast to the traditional regression approach for mixed variables (e.g.,\ Chapter 6 of \cite{Lauritzen} and Chapter 11 of \cite{Whittaker1990}).
\cite{yang2015graphical} proposed a mixed-variable model based on univariate conditional exponential families,
 but the model cannot avoid the restriction on the parameter space due to the integrability.
The difference between our model and the model proposed by \cite{yang2015graphical} is that \cite{yang2015graphical}'s model specifies the base measure a priori instead of the marginal densities, while our model specifies marginal densities instead of the base measure.
\end{example}

\begin{remark}[Choice of $h$]
 The choice of canonical statistics $h$ is one of important issues for practitioners.
 This section and Appendix \ref{sec: more examples} give several examples of $h$. The standard choice of $h$ has the form $h(x)=\prod_{i=1}^{d}h_{i}(x_{i})$ with $h_{i}:\mathcal{X}_{i}\to\mathbb{R}$.
 There are several ways to choose $h_{i}$: The first is to use polynomial functions in a suitable sense. The second is to use monotone functions in a suitable sense. The third is to employ known embedding functions; for example, see Appendix D.2. In any case, the conditional inference in Section \ref{section:conditional} is applicable, and the selection of $h$ using the conditional likelihood is possible. 
\end{remark}

\section{Inference} \label{section:conditional}

In this section, we consider the inference for $\theta$ using the conditional likelihood as well as the sampling algorithm for the minimum information dependence model.

\subsection{Decomposition of the likelihood}

Suppose that $x(1),\ldots,x(n)$ are independent and identically distributed (i.i.d.) according to a density in a minimum information dependence model (\ref{eq:min-info}).
Denote the components of $x(t)$ as $x(t)=(x_i(t))_{i=1}^d$.

We decompose the likelihood function into a marginal part and a dependent part using an order and a rank. By the well-ordering principle,
we can define a total order $\leq_i$
on $\mathcal{X}_i$ for each $i=1,\ldots,d$.
Using the ordering is convenient for the following description and, 
as we shall see in Lemma \ref{lem: decomposition}, the choice does not affect the inference. 
Also, the following remark provides a standard choice of the ordering.

\begin{remark}[Observational order]
Given the $n$ observations $(x(t))_{t=1}^{n}$,
we can use the ``observational order'', $x_i(t) \le_{i} x_i(s)$ 
if $t\le s$, as if they are predetermined.
The observational order makes implementation easier since it does not require sorting.
\end{remark}

For each $i=1,\ldots,d$, define the set of $i$-th marginal values by
\begin{align}
  M_i(1) \leq_i \cdots \leq_i M_i(n),
 \label{eq:marginal-information}
\end{align}
where for each $i$, $M_i=(M_i(1),\ldots,M_i(n))$ are the $n$ observations $(x_i(t))_{t=1}^n$ sorted by the predetermined order $\leq_i$.
We call it the {\em marginal order statistic}.
Define the {\em rank statistic} by a permutation $\pi_i=(\pi_i(t))_{t=1}^n\in\mathbb{S}_n$ such that $x_i(t)=M_i(\pi_i(t))$,
where 
we denote the symmetric group of degree $n$ by $\mathbb{S}_n$.
If there are ties of observations, we choose $\pi$ with equal probability over the set of permutations giving the same observations.
Denote the vector of all statistics as $M=(M_1,\ldots,M_d)\in\prod_{i=1}^d\mathcal{X}_i^n$ and $\pi=(\pi_1,\ldots,\pi_d)\in\mathbb{S}_n^d$.
For each $t=1,\ldots,n$,
the $t$-th observation $x(t)$ is recovered from $M$ and $\pi$, and is written as
\[
x(t) = (M\circ\pi)(t) = (M_i(\pi_i(t)))_{i=1}^d.
\]

\begin{example}
 For illustrative purposes, consider discrete spaces $\mathcal{X}_1=\{{\sf a},{\sf b},{\sf c}\}$ and $\mathcal{X}_2=\{1,2\}$ equipped with the alphabetical and numeric orders, respectively.
 Suppose that we have the $4$ observations
 $x(1) = ({\sf c},2),\ x(2) = ({\sf c},1),\ x(3) = ({\sf b},2),\ x(4) = ({\sf a},1)$.
 Then, the marginal order statistics are $M_1=({\sf a},{\sf b},{\sf c},{\sf c})$ and $M_2=(1,1,2,2)$.
 The rank statistic $\pi_1$ is an element of $\{(3,4,2,1),(4,3,2,1)\}$ with equal probability $1/2$.
 Suppose that we choose $\pi_1=(3,4,2,1)$.
 Similarly, choose $\pi_2=(3,1,4,2)$ from four possible permutations.
 Then, the observations are recovered from $M$ and $\pi$ as
 \begin{align*}
  (M\circ\pi)(1) &= (M_1(\pi_1(1)),M_2(\pi_2(1))) = (M_1(3),M_2(3)) = ({\sf c},2),\\
  (M\circ\pi)(2) &= (M_1(\pi_1(2)),M_2(\pi_2(2))) = (M_1(4),M_2(1)) = ({\sf c},1),\\
  (M\circ\pi)(3) &= (M_1(\pi_1(3)),M_2(\pi_2(3))) = (M_1(2),M_2(4)) = ({\sf b},2),\\
  (M\circ\pi)(4) &= (M_1(\pi_1(4)),M_2(\pi_2(4))) = (M_1(1),M_2(2)) = ({\sf a},1).
 \end{align*}
 \end{example}

The following lemma implies that
the full likelihood is decomposed as the product of 
the conditional likelihood independent of the marginal parameter $\nu$
and the marginal likelihood.

\begin{lemma}\label{lem: decomposition}
The likelihood function is decomposed as
\begin{align*}
 L(M,\pi;\theta,\nu) := \prod_{t=1}^n p((M\circ\pi)(t); \theta,\nu)
 = f(\pi|M;\theta) g(M;\theta,\nu),
\end{align*}
where
\begin{align*}
    f(\pi|M;\theta) 
    = \frac{e^{\sum_{t=1}^{n} \theta^\top h((M\circ\pi)(t))}}
    {\sum_{\tilde\pi\in\mathbb{S}_n^d}e^{\sum_{t=1}^{n} \theta^\top h((M\circ\tilde\pi)(t))}}\ \ \text{and}\ \ 
g(M;\theta,\nu) = \sum_{\tilde\pi\in\mathbb{S}_n^d} L(M,\tilde\pi;\theta,\nu).
\end{align*}
Further, the conditional likelihood does not depend on the choice of the ordering:
\begin{align*}
f(\pi\mid M;\theta)
=\frac{e^{\sum_{t=1}^{n}\theta^{\top}h(x(t))}}{\sum_{\tilde\pi\in\mathbb{S}_n^d}e^{\sum_{t=1}^{n} \theta^\top h(( x_{i}(\tilde\pi_{i}(t)))_{i=1}^{d})}}.
\end{align*}
\end{lemma}

\begin{proof}
For the former assertion,
it is sufficient to show that the conditional distribution of $\pi$ is
 \[
  \frac{L(M,\pi;\theta,\nu)}{\sum_{\tilde\pi\in\mathbb{S}_n^d}L(M,\tilde\pi;\theta,\nu)}
  = \frac{e^{\sum_{t=1}^{n} \theta^\top h((M\circ\pi)(t))}}{\sum_{\tilde\pi\in\mathbb{S}_n^d}e^{\sum_{t=1}^{n} \theta^\top h((M\circ\tilde\pi)(t))}}.
 \]
 Indeed, the adjusting-function part of $\log L(M,\pi;\theta,\nu)$ is
 \[
\sum_{j=1}^d\sum_{t=1}^n a_j(M_j(\pi_j(t));\theta,\nu) = \sum_{j=1}^d\sum_{t=1}^n a_j(M_j(t);\theta,\nu)
 \]
and does not depend on $\pi$, 
which proves the former assertion.
The latter assertion follows immediately from the identity $x(t)=(M\circ \pi)(t)$.
 \end{proof}

\begin{remark}
Let us mention that the choice or the randomization in defining rank statistics does not impact on the inference based on the conditional likelihood.
In fact, the latter part of Lemma \ref{lem: decomposition} shows the conditional likelihood
is independent from the choice of total orders $(\le_{i})_{i=1}^{d}$ and is not affected by the presence of ties, which is different from the inference of discrete copula models \citep{GenestNeslehova2007}.
\end{remark}

\subsection{Conditional maximum likelihood estimation} \label{subsection:conditional-MLE}

In this subsection,
on the basis of the conditional likelihood,
we propose 
a method of estimating $\theta$ and show that this yields a consistent estimator.

We first point out that 
the log conditional likelihood ratio
is almost the same as the log likelihood ratio
and thus utilizing the conditional likelihood for the inference of $\theta$ is reasonable. 
To show this,
we make the following assumptions.

\begin{assumption}
\label{assumption: thm 2}
The following hold:
\begin{enumerate}
\item[(1)] Let $(\mathcal{X}_{i},d_{i},\mathrm{d}x_{i})$ for $i=1,\ldots,d$ be metric measure spaces.
There exist $\kappa>0$ and $\alpha\ge 0$ such that
for each $i=1,\ldots,d$,
the  $\varepsilon$-covering number $\mathcal{N}(\mathcal{X}_{i},d_{i},\varepsilon)$ is bounded as $ \mathcal{N}(\mathcal{X}_{i},d_{i},\varepsilon) \le \kappa\varepsilon^{-\alpha}$ for sufficiently small $\varepsilon>0$.
\item[(2)]
The parameter space $\Theta$ of canonical parameter $\theta$ is bounded.
\item[(3)] The canonical statistic $h(x)$ is Lipschitz continuous with respect to~$x$ with Lipschitz constant $L_{h}$, where $\mathcal{X}$ is endowed with the 2 product metric $d_{\mathcal{X}}(x,y):=(\sum_{i=1}^{d}d^{2}_{i}(x_{i},y_{i}))^{1/2}$.
\end{enumerate}
\end{assumption}

Assumption \ref{assumption: thm 2} (1) states that 
each $\mathcal{X}_{i}$ is a totally bounded space with a finite upper Minkowski--Bouligand dimension (e.g.,\ \cite{Falconer2014}), and is quite general.
Assumption \ref{assumption: thm 2} (2) is a usual assumption.
Assumption \ref{assumption: thm 2} (3) ensures the existence of $p(x;\theta,\nu)$ by Corollary \ref{corollary:Lipschitz} and that of an approximator of $p(x;\theta,\nu)$.

Under these assumptions, we obtain the following theorem implying that the log conditional likelihood ratio (per sample) approaches to the log likelihood ratio (per sample) uniformly in $\theta$ as the sample size gets larger.
\begin{theorem} \label{theorem:negligible}
Let $\theta_{0}$ be the true canonical parameter in the parameter space $\Theta$ of $\theta$, and let $\nu_{0}$ be the true marginal parameter in the parameter space $\mathrm{N}$ of $\nu$.
Suppose that $x(1),\ldots,x(n)$ are i.i.d.~according to $p(x;\theta_{0},\nu_{0})$.
\begin{enumerate}
    \item[(1)] Under Assumption \ref{assumption: thm 2},
there exists a positive constant $C$ not depending on $n$ such that
we have
\begin{align}
    \Ep\left[\sup_{\theta\in\Theta}\frac{1}{n}
    \left|
    \log \frac{\prod_{t=1}^{n}p(x(t);\theta_{0},\nu_{0})}{\prod_{t=1}^{n}p(x(t);\theta,\nu_{0})}
    -\log \frac{f(\pi\mid M;\theta_{0})}{f(\pi\mid M;\theta)}
    \right|
    \right]
    \le C\varepsilon_{n},
    \label{eq: validity of conditional likelihood ratio}
\end{align}
where
\[\varepsilon_{n}:=\max\left\{n^{-\frac{(d-1)\alpha+1}{2(d\alpha+1)}}(\log n)^{-\frac{\alpha}{2(d\alpha+1)}},n^{-\frac{1}{d\alpha+1}}(\log n)^{\frac{1}{d\alpha+1}}\right\}.
\]
\item[(2)] In addition, under Assumption \ref{assumption: thm 2} and under the additional assumption that for $i=1,\ldots,d$, the marginal density $r_{i}(x_{i};\nu)$ is log-Lipschitz continuous with Lipschitz constant $L_{r_{i},\nu}$. where $\sup_{\nu\in \mathrm{N}}L_{r_{i},\nu}<\infty$,
there exists a positive constant $C$ not depending on $n$ such that
for any subset $\mathrm{S}\subset \mathrm{N}$ containing $\nu_{0}$,
we have
\begin{align}
    \Ep\left[\sup_{\theta\in\Theta,\nu\in\mathrm{S}}\frac{1}{n}
    \left|
    \log \frac{\prod_{t=1}^{n}p(x(t);\theta_{0},\nu_{0})}{\prod_{t=1}^{n}p(x(t);\theta,\nu)}
    -\log \frac{f(\pi\mid M;\theta_{0})}{f(\pi\mid M;\theta)}
    \right|
    \right]
    \le C
    \left\{
    \sup_{\nu\in\mathrm{S}}
    \tilde{D}(\nu_{0},\nu)
    +\varepsilon_{n}
    \right\},
    \label{eq: validity of conditional likelihood ratio 2}
\end{align}
where $\tilde{D}(\nu_{0},\nu)
:=\sum_{i=1}^{d}\{\|r_{i}(\cdot;\nu_{0})-r_{i}(\cdot;\nu)\|_{1}
+D(r_{i}(\cdot;\nu_{0}),r_{i}(\cdot;\nu))\}$ with $D(\cdot,\cdot)$ the Kullback--Leibler divergence.
\end{enumerate}
\end{theorem}

The proof of the theorem is given in Appendix C.3 in the appendix. 
The main ingredients of the proof are
the approximation using $\varepsilon$-net,
the Stirling approximation (e.g.,\ \cite{Robbins1955}), 
the $\ell_{1}$ and the Kullback--Leibler deviation inequalities for the multinomial distribution (\cite{Weissman2003} and \cite{Agrawal2020}),
 the recent quantitative stability result for the entropic optimal transport with respect to the marginals (e.g., \cite{EcksteinNutz_arXiv}),
and the Pythagorean theorem for the minimum information dependence model (Theorem S.1 in Appendix B.2 of the appendix).

Theorem \ref{theorem:negligible} implies that the marginal order statistics $M$ are almost ancillary for the dependence parameter $\theta$, that is, they contain little information about $\theta$:
\begin{align*}
\Ep\left[ \sup_{\theta\in \Theta}\left|\frac{1}{n}\log \frac{g(M;\theta_{0},\nu_{0})}{g(M;\theta,\nu_{0})}
\right|
\right]
\to 0.
\end{align*}
In the literature, there are several results related to Theorem \ref{theorem:negligible}.
In the bivariate Gaussian model, 
the sample variances are shown to be almost ancillary for the correlation parameter; see Example 2.30 of \cite{CoxHinkley}.
In contingency tables,
\cite{choi2015elucidating} and \cite{Little1989} have also shown that the marginal frequency is an almost ancillary statistic for the dependence parameter.

Henceforth, fix the marginal order statistics $M$ and denote
$ h_*(\pi) = \sum_{t=1}^n h((M\circ\pi)(t)) \in\mathbb{R}^K$ 
for simplicity.
The conditional likelihood is then expressed as
\begin{align}
 f(\pi|M;\theta) = \frac{e^{\theta^\top h_*(\pi)}}{\sum_{\tilde\pi\in\mathbb{S}_n^d}e^{\theta^\top h_*(\tilde\pi)}}
  \label{eq:conditional-likelihood}
\end{align}
and forms an exponential family with a canonical parameter $\theta$ and a sufficient statistic $h_*(\pi)$.
Exponential families on permutations without conditioning and their limiting behavior were investigated by \cite{Mukherjee2016}, where the limit of the model was shown to be a minimum information copula model under suitable conditions.

\begin{definition}
The conditional maximum likelihood estimate (CLE) $\hat\theta$ is a maximizer of the conditional likelihood (\ref{eq:conditional-likelihood}).
\end{definition}

Using Theorem \ref{theorem:negligible}, we obtain the consistency of CLE.
To show this, we make an additional assumption.

\begin{assumption}
\label{assumption: cor 4}
There exists a $p(\cdot;\theta_{0},\nu_{0})$-square-integrable function $A(\cdot)$ such that
\[\left|\sum_{i=1}^{d}a_{i}(x_{i};\theta,\nu_{0})-\sum_{i=1}^{d} a_{i}(x_{i};\theta',\nu_{0})\right|<A(x)\|\theta-\theta'\|\quad\text{for} \quad \theta\ne \theta' \in \Theta.\]
\end{assumption}

This assumption ensures that the log-likelihood of the minimum information dependence model forms Glivenko--Cantelli class (e.g., \cite{vanderVaartandWellner}),
and is expected by the quantitative stability result for the entropic optimal transport with respect to the cost (e.g., \cite{EcksteinNutz_arXiv}).

Then, we obtain the following consistency result of CLE $\hat{\theta}$.
\begin{corollary}
\label{cor:clemle-negligible}
 Under Assumptions \ref{assumption: thm 2} and \ref{assumption: cor 4},
 we have $\hat{\theta}\to \theta_{0}$ in probability.
\end{corollary}

We then consider the asymptotic variance of $\hat{\theta}$.
Let $\Psi(\theta)$ be the potential function of the conditional likelihood (\ref{eq:conditional-likelihood}) as an exponential family, that is,
$\Psi(\theta)
 = \Psi(\theta|M)
 = \log\left(\sum_{\tilde\pi\in\mathbb{S}_n^d} e^{\theta^\top h_*(\tilde\pi)}\right)$.
Then the conditional likelihood (\ref{eq:conditional-likelihood}) can be written as $f(\pi|M)=\exp(\theta^\top h_*(\pi)-\Psi(\theta))$.
Denote the derivative with respect to the parameter as $\partial_j=\partial/\partial\theta_j$.
The expectation parameter and Fisher information matrix are
\begin{align}
     \mu_j(\theta) &= \partial_j\Psi(\theta) = E[h_{*j}(\pi)|M] \ \text{and}
     \label{eq:conditional-mu}\\
    G_{jk}(\theta) &= \partial_j\partial_k\Psi(\theta) = E[\{h_{*j}(\pi)-\mu_j(\theta)\}\{h_{*k}(\pi)-\mu_k(\theta)\}|M], \label{eq:conditional-Fisher}
 \end{align}
 respectively, where $h_{*j}(\pi)$ is the $j$-th element of $h_*(\pi)$ and the expectation is taken with respect to $f(\pi|M;\theta)$.

 For 2 by 2 contingency tables, CLE was shown to be asymptotically normal and efficient by \cite{harkness1965properties}, based on the asymptotic form of the non-central hypergeometric distribution (\cite{hannan1963normal}); see also \cite{kou1996asymptotics}. For general contingency tables, the following theorem holds.

 \begin{theorem}[\cite{haberman1977analysis}, Theorem 4.1] \label{theorem:asymptotic-normality}
   Suppose that $\mathcal{X}_1,\ldots,\mathcal{X}_d$ are finite.
   Then, the asymptotic distribution of $\sqrt{n}(\hat\theta-\theta)$ is ${\rm N}(0,g^{jk}(\theta))$, where $g^{jk}(\theta)$ is the inverse of the Fisher information matrix $g_{jk}(\theta)$ with respect to $\theta$ of the model (\ref{eq:min-info}).
   The conditional Fisher information matrix in (\ref{eq:conditional-Fisher}) satisfies $G_{jk}(\theta)/n\to g_{jk}(\theta)$ in probability.
 \end{theorem} 
 
  In Appendix B.1, we show that the Fisher information matrix $g_{jk}$ is the Hessian matrix of the potential function $\psi$ and actually coincides with that given by \cite{haberman1977analysis}.
  We also provide a useful expression for $g_{jk}$ using the back-fitting algorithm (\cite{buja1989linear}).
  Furthermore, in Appendix B.2, the canonical parameter $\theta$ and the marginal parameter $\nu$ are shown to be orthogonal to each other.
  
 We expect that the same property as Theorem~\ref{theorem:asymptotic-normality} holds for infinite sample spaces.
 A numerical study supporting this claim is provided in Section~\ref{section:numerical}.
 Note that even the $\sqrt{n}$-consistency of the maximum likelihood estimator (MLE) for the exponential families on permutations has not been shown in the literature. \cite{Mukherjee2016} showed the MLE to be consistent, but obtained a $\sqrt{n}$-consistency result only for the pseudo-likelihood estimator of Besag.
 %Note that even the $\sqrt{n}$-consistency of the maximum likelihood estimator (MLE) for the exponential families on permutations was not known in the literature (e.g.,\,\cite{Mukherjee2016}).
 Although the asymptotic normality in general cases is not proved, it is reasonable to use the inverse of the conditional information matrix $G_{jk}(\hat\theta)/n$ as an estimate of the asymptotic covariance of $\hat\theta$. We also suggest using the likelihood ratio (or score/Wald) test and Akaike's information criterion assuming asymptotic normality.

From the theory of discrete exponential families (e.g.,\, \cite{Rinaldo2009EJS}), we obtain the following lemma.

\begin{lemma} \label{lemma:cmle-existence}
Let $P={\rm conv}(\{h_*(\tilde\pi)\mid\tilde\pi\in \mathbb{S}_n^d\})\subset\mathbb{R}^K$ be the convex hull of the range of the sufficient statistic.
Suppose that the interior of $P$ is nonempty.
Then, CLE $\hat\theta$ exists if and only if $h_*(\pi)\in {\rm int}(P)$.
The estimator is unique and satisfies
$\mu_j(\hat\theta) = h_{*j}(\pi)$,
whenever it exists.
\end{lemma}

The condition ${\rm int}(P)\neq\emptyset$ in the lemma is difficult to check because it requires the computation of all possible values of $h_*$. We provide a tractable sufficient condition in Subsection~\ref{subsection:Besag}.

Since any exponential family is log-concave with respect to the canonical parameter,
we can use, in principle, any convex programming solver to obtain CLE. However, a critical issue here is to calculate the normalizing constant in the denominator of (\ref{eq:conditional-likelihood}).
To overcome this difficulty, we propose a sampling approach and a pseudo likelihood approach in the following subsections.

\subsection{Estimation via Monte Carlo} \label{subsection:sampling}

To perform the conditional inference,
we need to compute the expectations of several statistics under the conditional distribution (\ref{eq:conditional-likelihood}).
We propose a sampling method in the Metropolis--Hastings manner described in Table \ref{tab: exchange algorithm}, which is quite easy to implement.
We call the method the exchange algorithm.
Note that the method is essentially the same as those for contingency tables (e.g.,\ \cite{Diaconis1998}).

\begin{table}[tb]
 \caption{\label{tab: exchange algorithm}Exchange algorithm.}
 \centering
 \fbox{%
\begin{tabular}{l}
\textbf{Input}: An initial permutation $\pi^{(0)}\in\mathbb{S}_n^d$ and the number of samples $L$.\\
\textbf{Output}: $L$ samples $(\pi^{(l)}\in \mathbb{S}_n^d )_{l=1}^{L}$ of permutations.\\
Step 1: Initialize $\pi\leftarrow\pi^{(0)}$ and $l=1$.\\
Step 2: Select $1\leq i\leq d$ and $1\leq s<t\leq n$ uniformly at random.\\
Let $\tau_{st}^i\in\mathbb{S}_n^d$ be the transposition between $s$ and $t$ with respect to the $i$-th variable.\\
Compute the conditional likelihood ratio $\rho = \frac{f(\pi\circ\tau_{st}^i|M;\theta)}{f(\pi|M;\theta)}
 = \frac{e^{H_s(\pi\circ\tau_{st}^i)+H_t(\pi\circ\tau_{st}^i)}}{e^{H_s(\pi)+H_t(\pi)}}$,
\\
where $H_t(\pi)=\theta^\top h((M\circ\pi)(t))$.\\
Step 3: Generate a random number $u$ uniformly distributed on $[0,1]$.\\
If $u\leq\min(1,\rho)$, then update $\pi$ to $\pi\circ\tau_{st}^i$.\\
Step 4: Let $\pi^{(l)}\leftarrow\pi$ and $l\leftarrow l+1$.\\
Go to Step 2 if $l\leq L$, and output $(\pi^{(l)})_{l=1}^L$ otherwise.
\end{tabular}}
\end{table}

Note that the marginal order statistics $M$ are preserved during the procedure.
The state space of the Markov chain is $\mathbb{S}_n^d$.
The following lemma is obtained immediately from the construction, and the proof is omitted.

\begin{lemma}
 The chain $(\pi^{(l)})_{l=1}^\infty$ from the exchange algorithm
 is ergodic with its stationary distribution $f(\pi|M;\theta)$ in (\ref{eq:conditional-likelihood}).
\end{lemma}

\begin{remark}\label{remark:sampling algorithm for the min-info}
 It is natural to consider a sampling method for $p(x;\theta,\nu)$.
 A naive method is just to generate $n$ random elements $(x_i(t))_{t=1}^n$ according to the marginal distribution $r_i(x_i;\nu)$ for each $i=1,\ldots,d$, independently, and then to employ the exchange algorithm.

 However, the method is not exact. Indeed, if $n=1$, then the procedure generates a sample from the independent model $\prod_{i=1}^d r_i(x_i;\nu)$, not from the correct distribution.
 It is expected that the distribution of the sample generated in this way converges to the correct distribution as $n\to\infty$.
 This observation is supported by Theorem~\ref{theorem:negligible} because the target marginal density $g(M;\theta,\nu)$ and the independent counterpart $g(M;0,\nu)$ are asymptotically equivalent as $n\to\infty$.
\end{remark}

We can compute CLE $\hat\theta$ via MCMC in line with \cite{geyer1992constrained}.
More specifically, let $\theta$ be a current estimate of $\hat\theta$ and let $\{\pi^{(l)}\}_{l=1}^L$ denote a sample from the conditional likelihood in (\ref{eq:conditional-likelihood}) obtained by the exchange algorithm.
The quantities $\mu_j(\theta)$ and $G_{jk}(\theta)$ are approximated by
$\check\mu_j
= \sum_{l=1}^L h_{*j}(\pi^{(l)})/L$
and
$\check{G}_{jk} = \sum_{l=1}^L (h_{*j}(\pi^{(l)})-\check{\mu}_j)(h_{*k}(\pi^{(l)})-\check{\mu}_k)/L$.
Then, the estimate is updated by Fisher's scoring method:
$\theta_j\leftarrow \theta_j + \sum_{k=1}^K\check{G}^{jk}\{h_{*k}(\pi)-\check{\mu}_k\}$,
where $\check{G}^{jk}$ is the inverse matrix of $\check{G}_{jk}$.
The procedure is repeated until convergence.
As noted by \cite{geyer1992constrained}, the MCMC samples can be reused at every step by importance sampling.

The standard error of $\hat\theta_j$ is estimated by $(\check{G}^{jj}/n)^{1/2}$.
Hypothesis testing and model selection based on the likelihood ratio statistic (or score/Wald statistic) are also available, where the statistic is computed from the MCMC samples.

\subsection{Besag's pseudo likelihood} \label{subsection:Besag}

An alternative to CLE is Besag's pseudo likelihood estimator for permutations (see \cite{Mukherjee2016}), which does not require the normalizing constant in (\ref{eq:conditional-likelihood}).
The {\em pseudo likelihood estimator} (PLE) is defined as a maximizer of
\[
 \prod_{i=1}^d\prod_{1\leq s<t\leq n}
 f(\pi_i(s),\pi_i(t)\mid (\pi_j(u))_{(u,j)\neq (s,i),(t,i)}, M; \theta),
 \]
 where $\pi=(\pi_i(t))\in\mathbb{S}_n^d$ is the observed rank statistic and
 \begin{align*}
 f(\pi_i(s),\pi_i(t)\mid (\pi_j(u))_{(u,j)\neq (s,i),(t,i)},M;\theta)
 &= \frac{1}{1+e^{\theta^\top(h_*(\pi\circ\tau_{st}^i)-h_*(\pi))}}.
\end{align*}
Recall that $\tau_{st}^i$ denotes the transposition between $s$ and $t$ with respect to the $i$-th variable (see Subsection~\ref{subsection:sampling}).
A similar form of PLE for copula models was proposed by \cite{chen_sei2022} in the framework of scoring rules.

 The expression of the pseudo likelihood coincides with the likelihood of a logistic regression model, where the explanatory variable is $u_{st}^i=h_*(\pi)-h_*(\pi\circ\tau_{st}^i)\in\mathbb{R}^K$ and the response variable is $y_{st}^i=1$ for all $i$ and $(s,t)$.
 Hence, any software package for logistic regression can be utilized.

 As a consequence of regarding PLE as MLE of a logistic regression model, we obtain a condition for the existence of PLE.

 \begin{lemma}[\cite{AlbertAnderson1984}]
  Let $Q={\rm conv}(\{h_*(\pi\circ\tau_{st}^i)\mid 1\leq i\leq d,1\leq s<t\leq n\})$.
  Suppose that ${\rm int}(Q)\neq\emptyset$.
  Then, PLE exists if and only if $h_*(\pi)\in{\rm int}(Q)$.
 \end{lemma}
 
  The condition $h_*(\pi)\in{\rm int}(Q)$ is a sufficient condition for $h_*(\pi)\in{\rm int}(P)$ in Lemma~\ref{lemma:cmle-existence} because $Q\subset P$.
  In other words, we obtain the following theorem that is practically useful since we first try to find PLE, and if that succeeds, we can then proceed to computing CLE.
  
  \begin{theorem} \label{theorem:existence-sufficient}
   If PLE exists, then CLE exists.
  \end{theorem}

  Further, we obtain the consistency of PLE.
  For $i=1,\ldots,d$, and for $(s,t)\in\{1,\ldots,n\}^{2}$ with $s\ne t$, let $u_{st}^i=h_*(\pi)-h_*(\pi\circ\tau_{st}^i)\in\mathbb{R}^K.$
  \begin{theorem}
  \label{theorem:consistencyofPL}
In addition to Assumptions \ref{assumption: thm 2} (2)--(3), we assume that 
\begin{itemize}
\item for any $v\in \mathbb{R}^{K}$ and $\theta\in\Theta$,
$\Ep \left[
\sum_{i=1}^{d}\{v^{\top}u^{i}_{12}\}^{2}
\{\cosh(\theta^{\top}u^{i}_{12} /2)\}^{-2}
\right]>0$, and
\item for each $i=1,\ldots,d$, there exists $x^{0}_{i}$ such that $\Ep d_{i}^{2}(X_{i}(1),x^{0}_{i})<\infty$.
\end{itemize}
Then,
PLE converges to $\theta_{0}$ in probability.
  \end{theorem}

The proof is given in Appendix C.5 and employs the theory of $U$-statistics (in particular, the theory of $M_{m=2}$-estimators); see \citealp{PenaGine,BoseChatterjee_2018}).
The assumptions in the theorem are quite mild. The first additional assumption ensures the global uniqueness of the expected pseudo likelihood; \cite{Hyvarinen2006} discusses a similar assumption for PLE in Boltzman machines. The second additional assumption ensures the moments of marginals.

Note that in Boltzman machines, PLE can be regarded as the contrastive divergence learning \citep{Hinton2002} that is a surrogate of MLE via MCMC; see \cite{Hyvarinen2006}. The same argument is applicable to our case, which, together with Theorems \ref{theorem:existence-sufficient}--\ref{theorem:consistencyofPL}, suggests that PLE is adopted as the initial value of Fisher's scoring method.

\begin{remark}[Computation time]
Generally, as the computation of CLE includes MCMC iteration in each optimization step, the computation of PLE is faster than that of CLE. This aspect is confirmed in the subsequent simulation studies.
\end{remark}

Finally, by applying the theory of $M_{m=2}$-estimators (Theorem 2.3 of \cite{BoseChatterjee_2018}; see also Appendix C.6), we can estimate the asymptotic variance of PLE $\hat{\theta}_{\mathrm{PLE}}$ by the sandwich estimator $(4/n)\hat{J}^{-1}_{\mathrm{PLE}}\hat{I}_{\mathrm{PLE}}\hat{J}^{-1}_{\mathrm{PLE}}$, where 
\begin{align*}
\hat{I}_{\mathrm{PLE}}&=
\frac{1}{n}
\sum_{s=1}^{n}
\left(
\frac{1}{n}\sum_{t\ne s, t=1}^{n}
\sum_{i=1}^{d}\frac{u^{i}_{st}}{1+e^{\hat{\theta}_{\mathrm{PLE}}^{\top}u^{i}_{st} }}
\right)
\left(
\frac{1}{n}\sum_{t\ne s, t=1}^{n}
\sum_{i=1}^{d}\frac{(u^{i}_{st})^{\top}}{1+e^{\hat{\theta}_{\mathrm{PLE}}^{\top}u^{i}_{st} }}
\right)\,\,\text{and}
\\
\hat{J}_{\mathrm{PLE}}&=
\frac{2}{n(n-1)}
\sum_{1\le s < t\le n}
\sum_{i=1}^{d}\frac{u^{i}_{st}(u^{i}_{st})^{\top}}{\{1+e^{\hat{\theta}_{\mathrm{PLE}}^{\top}u^{i}_{st}}\}^{2}}.
\end{align*}

\section{Simulation studies} \label{section:numerical}

In this section, we provide several numerical studies for the inference based on CLE and PLE.

\subsection{Gaussian cases}

We first examine the performance of CLE developed in Section~\ref{section:conditional} by applying it to the Gaussian model.
We generated a random sample $\{x(t)\}_{t=1}^n$ of size $n=50$ from the $4$-dimensional centered Gaussian distribution with the covariance matrix $\sigma_{ij}=(1/2)^{|i-j|}$.
For estimation, we assumed the full Gaussian model, which is a minimum information dependence model (Example~\ref{example:Gaussian}).
The parameter of interest is
$\theta=(\theta_k)_{k=1}^6=(-\sigma^{12},-\sigma^{13},-\sigma^{14},-\sigma^{23},-\sigma^{24},-\sigma^{34})$,
where $\sigma^{ij}$ denotes the inverse of $\sigma_{ij}$.
The true value of $\theta$ is set to $\theta=(2/3,0,0,2/3,0,2/3)$.
The tolerance for solving the conditional likelihood equation 
was set to $10^{-2}$ and the MCMC length $L$ was adaptively increased from an initial value $L=150n=7500$.
We used PLE as an initial value for the Fisher scoring method on CLE.
Furthermore, the scoring method is restarted if a component of $\theta$ becomes a huge value at some step due to variability of MCMC.
These rules are practically effective and adopted in subsequent examples as well.

\begin{table}[htb]
 \caption{\label{tab:Gaussian}Root mean square errors of three estimators.
 The model is Gaussian. The true parameter value is $\theta=(2/3,0,0,2/3,0,2/3)$.
 }
 \centering
 \fbox{%
 \begin{tabular}{c|cccccc}
 & $\theta_1$& $\theta_2$& $\theta_3$& $\theta_4$& $\theta_5$& $\theta_6$\\
 \hline
 CLE& 0.283& 0.253& 0.225& 0.314& 0.251& 0.285\\ 
MLE& 0.285& 0.254& 0.226& 0.313& 0.251& 0.286\\ 
PLE& 0.300& 0.262& 0.235& 0.337& 0.260& 0.306 
\end{tabular}}
\end{table}

We repeated the same experiment 1,000 times. Table~\ref{tab:Gaussian} shows the root mean square errors of  CLE and MLE together with those of PLE defined in Section~\ref{subsection:Besag}.
CLE has almost the same performance as MLE. PLE is also competitive but slightly worse.
Note that the values on $\theta_5$ and $\theta_6$ are close to those on $\theta_2$ and $\theta_1$, respectively, by symmetry of the covariance structure.

\begin{table}[htb]
 \caption{\label{tab:Gaussian-more}
 The root mean squares of the norm $\|\hat\theta-\theta\|$ for the three estimators.
 The model is Gaussian. The covariance structures are $\sigma_{ij}=\rho^{|i-j|}$ for the auto-regressive model and $\sigma_{ij}=(1-\rho)\delta_{ij}+\rho$ for the exchangeable model.
 }
 \centering
 \fbox{%
 \begin{tabular}{c|cccc|ccc}
  & \multicolumn{4}{|c|}{Auto-regressive model} & \multicolumn{3}{|c}{Exchangeable model}\\
  & $\rho=0$& $\rho=1/4$& $\rho=1/2$& $\rho=3/4$& $\rho=1/4$& $\rho=1/2$& $\rho=3/4$ \\
  \hline
  CLE & 0.415 & 0.507 & 0.647 & 1.410 & 0.477 & 0.688 & 1.392 \\ 
  MLE & 0.416 & 0.509 & 0.653 & 1.422 & 0.479 & 0.693 & 1.398 \\ 
  PLE & 0.433 & 0.529 & 0.694 & 1.494 & 0.507 & 0.734 & 1.516 \\ 
\end{tabular}}
\end{table}

We also examined other covariance matrices of the form $\sigma_{ij}=\rho^{|i-j|}$ for $\rho\in\{0,1/4,1/2,3/4\}$
and $\sigma_{ij}=(1-\rho)\delta_{ij}+\rho$ for $\rho\in\{1/4,1/2,3/4\}$. The root mean squares of the norm $\|\hat\theta-\theta\|$ are summarized in Table~\ref{tab:Gaussian-more}, where the number of experiments is 200 in each case.

The mean computational time (resp.\ standard deviation) for CLE and PLE per each experiment was 6.2 (3.3) and 2.0 (0.2) in seconds, respectively. We find that PLE is faster and (numerically) more stable.

\subsection{Three-dimensional interaction}

We next consider a three-dimensional interaction model
\[ h(x_1,x_2,x_3) = (x_1x_2, x_1x_3, x_2x_3,  x_1x_2x_3)^\top
\]
on $\mathcal{X}=\prod_{i=1}^3\mathcal{X}_i=\mathbb{R}^3$.
The true parameter value of $\theta$ is set to $\theta=(1,0,0,-1)$ and the true marginal densities are set to the standard normal distribution.
For the simulation, we first generate a ``population'' of size $N=10^3$ by the sampling algorithm in Remark~\ref{remark:sampling algorithm for the min-info} with the number of iterations $L=150N$. Then a sample of size $n=100$ is randomly selected from the population without replacement.
The tolerance for CLE and the step length of MCMC are set to the same values as the Gaussian case.
We repeated the same experiment 1,000 times.

\begin{table}[htb]
    \caption{
    \label{tab:3-dim}Root mean square error of two estimators and  coverage probabilities of the 95\% confidence intervals based on CLE.
    The model is the three-dimensional interaction model and the true parameter value is $\theta=(1,0,0,-1)$.
    }
    \centering
    \fbox{
    \begin{tabular}{c|cccc}
         & $\theta_1$& $\theta_2$& $\theta_3$& $\theta_4$  \\\hline
        error of CLE & 0.214& 0.159& 0.153& 0.168 \\
        error of PLE & 0.260& 0.190& 0.189& 0.250 \\\hline
        coverage & 0.953& 0.952& 0.965& 0.968\\
    \end{tabular}}
\end{table}

Table~\ref{tab:3-dim} shows the root mean square errors of CLE and PLE together with the coverage probability of the 95\% confidence intervals constructed from CLE.
As expected, CLE has smaller error than PLE.
The confidence intervals are almost exact or slightly conservative.

\begin{table}[htb]
 \caption{\label{tab:3dim-more}
 The root mean squares of the norm $\|\hat\theta-\theta\|$ for CLE and PLE.
 The model is the three-dimensional interaction model. The true parameter values are $\theta=(a,0,0,-a)$, where $a\in\{0,1,2\}$.
 }
 \centering
 \fbox{%
 \begin{tabular}{c|ccc}
  & $a=0$& $a=1$& $a=2$\\
  \hline
CLE & 0.211 & 0.353 & 0.556 \\ 
  PLE & 0.230 & 0.459 & 0.738 \\ 
%%% The following is the previous version of the table.
%%% The sample size was incorrectly set to n=50.
%CLE & 0.316 & 0.551 & 0.822 \\ 
%  PLE & 0.356 & 0.711 & 1.083 \\
\end{tabular}}
\end{table}

We also examined other parameter values $\theta=(a,0,0,-a)$ for $a\in \{0,1,2\}$.
The root mean squares of the norm $\|\hat\theta-\theta\|$ are summarized in Table~\ref{tab:3dim-more}.

\subsection{Mixed variables}

Finally, we study a case with continuous and discrete variables.
As in Example~\ref{example:mixed}, the marginal distributions are set to Beta(10, 10) and Poisson(3), respectively. The canonical statistic is
$h(x_1,x_2) = x_1/(1+x_2)$ and the true parameter values are set to  $\theta\in\{0,10,100\}$. The sampling method is the same as the preceding subsection except that the sample size is $n=50$.
The tolerance for solving CLE was set to $10^{-5}$.

Table~\ref{tab:mixed} shows the root mean square errors, biases and standard deviations of CLE and PLE. CLE has better performance than PLE.

\begin{table}[htb]
    \caption{
    \label{tab:mixed}
    Root mean square errors, biases and standard deviations of the two estimators.
    The model is $h(x_1,x_2)=x_1/(1+x_2)$ with the beta and Poisson marginals. The true parameter value is set to $\theta\in\{0,10,100\}$.
    }
    \centering
    \fbox{
    \begin{tabular}{r|ccc|ccc|ccc}
         & \multicolumn{3}{c}{$\theta=0$} &\multicolumn{3}{|c}{$\theta=10$}& \multicolumn{3}{|c}{$\theta=100$}\\
         & RMSE& bias& SD
         & RMSE& bias& SD
         & RMSE& bias& SD\\
         \hline
CLE & 8.19 & 0.08 & 8.18 & \hphantom{0}9.40 & 1.14 & \hphantom{0}9.33& 29.26 & 4.96 & 28.84 \\ 
PLE& 9.03 & 0.12 & 9.03 & 10.34 & 1.89 & 10.17 & 32.34 & 7.31 & 31.51 \\ 
%%% The following is the previous version of the table.
%  CLE & 8.294 & \hphantom{0}9.416 & 29.850 \\ 
%  PLE & 9.479 & 10.291 & 33.450 \\ 
    \end{tabular}}
\end{table}

\section{Future directions} \label{section:discussion}

We here address potential future directions of our work.

The interpretability of the canonical parameter $\theta$ have to be more clarified. 
In Section \ref{section:model}, we have clarified the connection between the total correlation and $\theta$.
In the application to Earthquake data in Appendix D.2,
we have demonstrated that the estimation of $\theta$ not only provided the existence of the dependence but also identified the pattern of the dependence. For clearer interpretation of $\theta$, the connection to the partial correlation has to be more investigated.

The asymptotic properties of CLE have to be clarified. The asymptotic normality and efficiency for contingency tables are described in Theorem~\ref{theorem:asymptotic-normality}. We expect that the same properties are valid even for infinite sample spaces.
The limiting behavior of the approximate sampling algorithm in Remark~\ref{remark:sampling algorithm for the min-info} as $n\to\infty$ is also under investigation.

In Theorem~\ref{theorem:existence-sufficient}, we found that CLE of $\theta$ exists under suitable conditions.
However, if the dimension of $\theta$ is high, the estimator may not exist and some regularization will be necessary.
A simple method of regularization is to assume a (conditional) conjugate prior density
\[
p(\theta|M) = \exp(\lambda_0\mu_0^\top\theta-\lambda_0\Psi(\theta))
\]
for the exponential family (\ref{eq:conditional-likelihood}), where 
$\Psi(\theta)$ is the potential function of the conditional likelihood in (\ref{eq:conditional-likelihood}),
and
$\mu_0\in\mathbb{R}^K$ and $\lambda_0>0$ are hyper-parameters.
The maximum a posteriori estimator
exists if $\mu_0$ is an interior point of $P$ defined in Lemma~\ref{lemma:cmle-existence}.
In practice, we can select $\mu_0$ as the sample mean of randomly generated vectors in $h_*(\mathbb{S}_n^d)$.
The hyper-parameter $\lambda_0$ may be set to $1/n$ as a rule of thumb.
Investigating the PLE with regularization would be another important issue.

In our estimation procedure, we need to perform MCMC sampling as stated in Subsection~\ref{subsection:sampling}. The proposed algorithm was to exchange two elements of permutations. This is just one particular move. A mover producing a shorter mixing time should be found. Parallel algorithms will also be valuable.

We focused on estimation of  the canonical parameter $\theta$ representing the dependence.
However, estimation of the marginal parameters $\nu$ is also important in some cases. For example, if our goal is to predict future observations, then an estimator of the marginal parameters is necessary. A simple procedure to estimate $\nu$ is to perform the maximum likelihood estimation assuming the independent model $\prod_{j=1}^d r_j(x_j;\nu)$. The obtained estimator $\hat\nu$ is consistent by Theorem~\ref{theorem:negligible}. However, $\hat\nu$ is not asymptotically efficient in general, as observed in the bivariate Gaussian model with a common variance.

In our analysis, we assumed that the data were completely observed. In practice, handling data missing is a common challenge. To deal with missing data, we can extend the domain $\mathcal{X}_j$ to $\mathcal{X}_j\cup\{{\sf NA}\}$ for each $j=1,\ldots,d$, where {\sf NA} indicates missing. Then, the proposed method is applicable whenever we specify the canonical statistic $h(x)$ including the missing indicator. For example, $h(x_1,x_2)=x_1I_{\{x_2={\sf NA}\}}$ represents a missing effect of the second variable $x_{2}$ on the first variable $x_1$. A missing data analysis is important future work.

Finally, there are two mathematical open problems.
First, the properties of the potential function $\psi(\theta)$ in the model (\ref{eq:min-info}) are  unknown except for convexity. We conjecture that $\psi(\theta)$ is analytic at every interior point, by analogy with the standard theory of exponential families.  
However, $\psi(\theta)$ is defined by functional equations, which makes the problem complicated.
The second open problem is the equivalence of the three feasibility conditions discussed in Subsection~\ref{subsection:feasible}.
The problem is related to the closedness of sum spaces of $L_1$-functions; see \cite{Ruschendorf1993}.

\section*{Acknowledgements} \label{section:acknowledgements}

We thank the AE and three referees for their constructive comments that have improved the quality of this paper.
We thank Yici Chen, Hironori Fujisawa, Masayuki Kano, Hisahiko Kubo,
Michiko Okudo, and Akifumi Okuno
for helpful discussions.
We used the earthquake catalog provided by the Japan Meteorological Agency (\cite{JMA})
and used GMT software package (\cite{WesselSmith1998}) to create the maps.
This work is supported by JST CREST (JPMJCR1763), JSPS KAKENHI (19K20222, 21H05205, 21K11781, 21K12067), and MEXT (JPJ010217).
The {\sf{R}} and Python codes are available at \url{https://doi.org/10.5281/zenodo.8012980}.

\appendix
\renewcommand{\theequation}{S.\arabic{equation}}
\renewcommand{\thetheorem}{S.\arabic{theorem}}
\renewcommand{\thelemma}{S.\arabic{lemma}}
\renewcommand{\thecorollary}{S.\arabic{corollary}}
\renewcommand{\theexample}{S.\arabic{example}}
\renewcommand{\thetable}{S.\arabic{table}}
\renewcommand{\thefigure}{S.\arabic{figure}}
  \bigskip
  \bigskip
  \bigskip
  \begin{center}
    {\LARGE\bf Supplementary material for ``Minimum information dependence modeling''}
\end{center}
  \medskip

The supplementary material provide 
(A) more examples of the minimum information dependence model,
(B) useful properties of the model, including information geometry, and relationship to the optimal transport and Schr\"{o}dinger problems,
(C) proofs of the main results,
and
(D) application to real data.
Hereafter, the numbering for the main manuscript should be like Theorem 1, and the numbering for this supplement should be like Theorem S.1, respectively.

\section{More examples of the minimum information dependence model}\label{sec: more examples}

Here, we provide more examples of the minimum information dependence model.

\begin{example}[total positivity]
 Let $d=2$ and $\mathcal{X}_1=\mathcal{X}_2=\mathbb{R}$.
 The marginal density functions $r_i$ are arbitrary.
 Suppose that $f_i(x_i)$ is a strictly increasing function and $\int f_i(x_i)^2r_i(x_i){\rm d}x_i<\infty$ for each $i$.
 Then $H_\theta(x_1,x_2)=\theta f_1(x_1)f_2(x_2)$ is feasible since the condition of Theorem~\ref{theorem:feasible} is satisfied by $b_i(x_i)=|\theta|f_i(x_i)^2/2$.
 If $\theta>0$, then the density function has total dependence, that is,
 $ p(x_1,x_2)p(y_1,y_2) - p(x_1,y_2)p(y_1,x_2)>0 $
 for all pairs $(x_1,x_2)$ and $(y_1,y_2)$ with $x_1<y_1$ and $x_2<y_2$.
 See \cite{HollandWang1987} and \cite{Kurowicka2015} for totally dependent continuous distributions.
\end{example}

\begin{example}[circular model]
 Let ${\rm d}x_i$ be the Lebesgue measure on the interval $[0,2\pi)$, which represents a unit circle. Let $r_i(x_i)=1/(2\pi)$ be the uniform density and define $H_\theta(x_1,x_2)=\theta \cos(x_1-x_2)$.
 Then $H_\theta$ is feasible for any $\theta\in\mathbb{R}$.
 Indeed, we can find the exact solution $a_i(x_i)=0$ and $\psi(\theta)=\log Z(\theta)$, where $Z(\theta)=\int_0^{2\pi}e^{\theta\cos u}{\rm d}u/(2\pi)$.
 More generally, if $H_\theta(x_1,x_2)=\theta f(x_1-x_2)$ for a given function $f$, then the exact solution is $a_i(x_i)=0$ and $\psi (\theta) = \log\{\int_0^{2\pi} e^{\theta f(u)}{\rm d}u/(2\pi)\}$.
 The model is a particular class of circulas as discussed by \cite{Jones_et_al2015}.
\end{example}

If all domains are finite, the model reduces to the log-linear model as follows.

\begin{example}[log-linear model] \label{example:log-linear}
 Consider a discrete space $\mathcal{X}=\prod_{i=1}^d \mathcal{X}_i$, where $\mathcal{X}_i=\{0,\ldots,I_i-1\}$ ($I_i\geq 2$) is equipped with the counting measure ${\rm d}x_i$.
 We denote $x_s=(x_j)_{j\in s}$ for each subset $s\subset[d]=\{1,\ldots,d\}$.
 Consider the minimum information dependence model
 \[
  p(x;\theta) = \exp\left(
  \sum_{s\subset[d],|s|\geq 2}\theta_s^\top h_s(x_s) - \sum_{i=1}^d a_i(x_i;\theta,\nu)-\psi(\theta,\nu)
  \right)\prod_{i=1}^d r_i(x_i;\nu),
 \]
 where $h_s=(\delta_{y_s})_{y_s\in\mathcal{Y}_s}$, $\mathcal{Y}_s=\prod_{j\in s}\{1,\ldots,I_j-1\}$ and $\delta_{y_s}(x_s)$ is the Kronecker delta.
 The model is log-linear except on the main effects $a_i(x_i;\theta,\nu)$s that have a one-to-one correspondence with the marginal distributions $r_i(x_i;\nu)$ for a given $\theta$ by Theorem~\ref{theorem:feasible}. The model reduces to the log-linear model in this sense.
 The parameterization is explained in terms of information geometry by \cite{Amari2001}, where the pair of parameters $\theta$ and $(r_i(\cdot;\nu))_{i=1}^d$ are called the mixed coordinate system.
\end{example}

\section{Properties of the model}
\label{section: Supplementary properties}

This supplement presents properties of the minimum information dependence model:
expressions for
the potential function
and
Fisher information matrix (Appendix \ref{subsection:potential}),
the information-geometrical structure (Appendix \ref{subsection:dual-foliation}),
and the relationship with the entropic transport and Schr\"{o}dinger problems (Appendix  \ref{subsection:transport}).

\subsection{Properties of the potential function} \label{subsection:potential}

Consider the minimum information dependence model $p(x;\theta,\nu)$ determined by (\ref{eq:min-info}), (\ref{eq:marginal-condition}), and (\ref{eq:zero-mean-condition}).
In this subsection, we fix the marginal parameters $\nu$ and omit them from expressions as
\begin{align}
 p_\theta(x) = e^{\theta^\top h(x)-\sum_i a_i(x_i;\theta) -\psi(\theta)}p_0(x),
 \label{eq:min-info-without-nuisance}
\end{align}
where $p_0(x)=\prod_{i=1}^d r_i(x_i)$.
Denote the Kullback--Leibler divergence by
\[
D(p,q)=\int p(x)\log\{p(x)/q(x)\}{\rm d}x.
\]
Let $\mathcal{M}$ be the set of density functions $p$ with the given marginals $r_i$.

\begin{lemma} \label{lemma:potential}
The potential function $\psi(\theta)$ can be written as
\begin{align}
 \psi(\theta) = \sup_{p\in\mathcal{M}}\left\{\theta^\top \int h(x)p(x){\rm d}x - D(p,p_0)\right\},
 \label{eq:psi}
\end{align}
where the supremum is attained by (\ref{eq:min-info-without-nuisance}).
In particular, $\psi(\theta)$ is convex.
\end{lemma}

\begin{proof}
 The objective function in (\ref{eq:psi}) is equal to
 \[
  - \int p(x)\log\frac{p(x)}{p_0(x)e^{\theta^\top h(x)}}{\rm d}x
  = -\int u(x)k(x)\log u(x){\rm d}x,
 \]
 where $k(x)=p_0(x)e^{\theta^\top h(x)}$ and $u(x)=p(x)/\{p_0(x)e^{\theta^\top h(x)}\}$.
 This is the objective function of Lemma~\ref{lemma:Borwein} in the proof of Theorem~\ref{theorem:feasible} except for the sign.
 From the proof of Theorem~\ref{theorem:feasible}, the unique optimizer $p_\theta$ of the right hand side of (\ref{eq:psi}) is given by (\ref{eq:min-info-without-nuisance}).
 The maximum value is
 \begin{align*}
  -\int p_\theta(x)\log\frac{p_\theta(x)}{p_0(x)e^{\theta^\top h(x)}}{\rm d}x
  &= \int p_\theta(x)\left(\sum_i a_i(x_i;\theta)+\psi(\theta)\right){\rm d}x\\
  &= \psi(\theta),
 \end{align*}
 which proves the equality in (\ref{eq:psi}).
 Since $\psi$ can be expressed as the supremum of affine functions, it is convex.
\end{proof}

Denote the domain of the potential function by
\[
 \mathcal{N} = \{\theta\in\mathbb{R}^K\mid \psi(\theta)<\infty\}.
\]
Suppose that ${\rm int}\mathcal{N}\neq\emptyset$.
Hereafter, we assume smoothness of $\psi$ and $a_i(x_i;\theta)$ as well as interchangeability between integration and differentiation.
These assumptions are satisfied if the sample space $\mathcal{X}$ is finite.
We use subscripts $\alpha,\beta,\ldots$ for $\theta$ and abbreviate derivatives using $\partial_\alpha=\partial/\partial\theta_\alpha$.
We also write
$c_\theta(x) = \sum_{i=1}^d a_i(x_i;\theta) + \psi(\theta)$.
The following lemma is easily proved by the identities $\int \partial_\alpha p_\theta(x){\rm d}x=0$ and $\int \partial_\alpha\partial_\beta p_\theta(x){\rm d}x=0$.

\begin{lemma} \label{lemma:psi-derivatives}
The first two derivatives of $\psi$ are
\begin{align*}
    \partial_\alpha\psi(\theta)
    &= \int h_\alpha(x)p_\theta(x){\rm d}x,
    \\
    \partial_\alpha\partial_\beta\psi(\theta)
    &= -\int p_\theta(x)\partial_\alpha\partial_\beta\log p_\theta(x){\rm d}x
    \\
    &= \int (h_\alpha(x)-\partial_\alpha c_\theta(x)) (h_\beta(x)-\partial_\beta c_\theta(x))
    p_\theta(x){\rm d}x
\end{align*}
In particular, $\psi(\theta)$ is strictly convex whenever $\{h_\alpha(x)\}_{\alpha=1}^K$ are linearly independent modulo additive functions.
\end{lemma}

From the lemma, the mean of $h(x)$ is the same as the gradient of $\psi(\theta)$ and the Fisher information matrix $g_{\alpha\beta}(\theta)$ is equal to the Hessian of $\psi(\theta)$.
 We call $\eta_\alpha=\partial_\alpha\psi$ the expectation parameter.
 Since $\psi(\theta)$ is strictly convex, we have a  one-to-one correspondence between the coordinates $\theta_\alpha$ and $\eta_\alpha$.

 Note that $g_{\alpha\beta}(\theta)$ is different from the covariance of $h_\alpha$
 because $\partial_\alpha c_\theta(x)$ is not the expectation of $h_\alpha$.
 The following lemma provides their difference.
  
 \begin{lemma}
 Let $\mathcal{A}$ be the subspace of $L_2(p_\theta(x){\rm d}x)$ spanned by additive functions \[  \sum_{j=1}^d b_j (x_j) \] and denote the orthogonal projection to $\mathcal{A}$ by $P_{\mathcal{A}}$.
 Then we have
 \[
 \partial_\alpha c_\theta=P_{\mathcal{A}}h_\alpha.
 \]
 In particular,
  \begin{align}
  g_{\alpha\beta}
  &= {\rm Cov}_\theta[(I-P_{\mathcal{A}})h_\alpha,(I-P_{\mathcal{A}})h_\beta]
  \label{eq:Fisher-info}
  \\
  &= {\rm Cov}_\theta[h_\alpha,h_\beta]
  - {\rm Cov}_\theta\left[P_{\mathcal{A}}h_\alpha,P_{\mathcal{A}}h_\beta
  \right],
  \label{eq:variance-decomposition}
 \end{align}
 where $I$ denotes the identity operator. 
 \end{lemma}
 
\begin{proof}
 Consider the minimization problem of minimizing $\int \{h_\alpha(x)-B(x)\}^2p_\theta(x){\rm d}x$ with respect to $B\in\mathcal{A}$.
 By using the variational method, we obtain the stationary condition
 \begin{align*}
  \int \{h_\alpha(x) - B(x)\} p_\theta(x)dx_{-i} &= 0
 \end{align*}
 for all $i$ and $x_i$.
 The solution is $B(x)=\partial_\alpha c_\theta(x)$
 because the derivative of the identity $\int p_\theta(x) {\rm d}x_{-i} = r_i(x_i)$ with respect to $\theta_\alpha$ is
 \[
  \int \{h_\alpha(x) - \partial_\alpha c_\theta(x)\} p_\theta(x){\rm d}x_{-i} = 0.
 \]
 Hence we have $\partial_\alpha c_\theta=P_{\mathcal{A}}h_\alpha$.
 The identities (\ref{eq:Fisher-info}) and  (\ref{eq:variance-decomposition}) follow from Lemma~\ref{lemma:psi-derivatives} and the variance decomposition of $h_\alpha=(h_\alpha-P_{\mathcal{A}}h_\alpha) + P_{\mathcal{A}}h_\alpha$, respectively.
 \end{proof}
 
 The projection in the lemma is the same as those for additive models (\cite{buja1989linear}). One recursive algorithm for finding $P_{\mathcal{A}}f$ for a given $f$ is called the back-fitting algorithm.

 We finally check that the Fisher information matrix coincides with the limit of the conditional Fisher information $G_{\alpha\beta}/n$ in (\ref{eq:conditional-Fisher}) if the sample space $\mathcal{X}$ is finite.
 Theorem~1.1 of \cite{haberman1977analysis} states that, in our terminology, the limit of $G_{\alpha\beta}/n$ given the marginal statistic $M$ as $n\to\infty$ is
 \begin{align}
 \bm{h}_\alpha^\top (I_{\mathcal{X}}-P_{\mathcal{A}}^*)^\top D(p_\theta)(I_{\mathcal{X}}-P_{\mathcal{A}}^*)\bm{h}_\beta,
 \label{eq:Haberman-Fisher}
 \end{align}
 where $\bm{h}_\alpha=(h_\alpha(x))_{x\in\mathcal{X}}$ is considered a column vector, $I_{\mathcal{X}}$ is the identity matrix, $D(p_\theta)$ is diagonal matrix with the diagonal part $(p_\theta(x))_{x\in\mathcal{X}}$, $\mathcal{A}$ is a set of vectors of the form $(\sum_{j=1}^d b_j(x_j))_{x\in\mathcal{X}}$ and $P_{\mathcal{A}}^*$ is the orthogonal projection to $\mathcal{A}$ with respect to the metric $D(p_\theta)$. It is straightforward to see that the expression (\ref{eq:Haberman-Fisher}) is the same as $g_{\alpha\beta}$ in (\ref{eq:Fisher-info}).

\subsection{Information geometry of the minimum information dependence model} \label{subsection:dual-foliation}

Here we investigate the information-geometrical structure of the minimum information dependence model: the Pythagorean theorem and the dual foliation structure that implies orthogonality between $\theta$ and $\nu$.
The case of finite spaces was discussed in \cite{Amari2001}.
For simplicity, we assume interchangeability of integrals, sums, and derivatives.

 Let $\mathcal{M}(r_{1},\ldots,r_{d})$ denote the set of density functions with given marginals $r_i$. For a given $H:\mathcal{X}\to\mathbb{R}$, let $\mathcal{E}(H)$ denote the set of probability density functions of the form $e^{H(x)-\sum_{i=1}^{d} a_i(x_i)}$ with $a_i(x_i)$ measurable.
 The following Pythagorean theorem holds.
 
\begin{theorem}[Pythagorean theorem; e.g.,\ \cite{Csiszar1975}] \label{theorem:Pythagorean}
 Let \[\mathcal{M}=\mathcal{M}(r_1,\ldots,r_d) \ \text{and}\ \mathcal{E}=\mathcal{E}(H).\]
 Then, for any $p\in\mathcal{M}$, $q\in \mathcal{M}\cap\mathcal{E}$, and $s\in\mathcal{E}$, we have
 \[
  D(p,s) = D(p,q) + D(q,s).
 \]
\end{theorem}

\begin{proof}
 Let 
 \[q(x)=\exp\left(H(x)-\sum_{i=1}^{d} a_i^q(x_i)\right)p_0(x)\]
 and 
 \[s(x)=\exp\left(H(x)-\sum_{i=1}^{d} a_i^s(x_i)\right)p_0(x).
 \]
 Then we have
 \begin{align*}
     D(p,s) - D(p,q) - D(q,s)
     &= \int \{p(x)-q(x)\}\log\frac{q(x)}{s(x)}{\rm d}x
     \\
     &= \int \{p(x)-q(x)\}\left\{-\sum_{i=1}^{d} a_i^q(x_i)+\sum_{i=1}^{d} a_i^s(x_i)\right\}{\rm d}x
     \\
     &= 0
 \end{align*}
 since $p$ and $q$ have the same marginals.
\end{proof}

In terms of information geometry, $\mathcal{M}(r_1,\ldots,r_d)$ and $\mathcal{E}(H)$ provide a dual foliation (e.g.,~\cite{AmariNagaoka2000}). More specifically, the space $\mathcal{P}$ of probability densities on $\mathcal{X}$ can be decomposed as
\[
 \mathcal{P} = \bigsqcup_{r_1,\ldots,r_d}\mathcal{M}(r_1,\ldots,r_d)
\]
and
\[
 \mathcal{P} = \bigsqcup_{[H]\in\mathcal{H}/\sim} \mathcal{E}(H),
\]
where $\mathcal{H}$ is the set of measurable functions and the equivalence relation $H_1\sim H_2$ is defined by $H_1(x)-H_2(x)=\sum_{i=1}^{d} a_i(x_i)$ for some $a_i$.
Each pair of manifolds $\mathcal{M}(r_1,\ldots,r_d)$ and $\mathcal{E}(H)$ intersects at most at one point by Theorem~\ref{theorem:Pythagorean}.

We obtain the following corollary.

\begin{corollary}[Orthogonality]
 Consider the minimum information dependence model (\ref{eq:min-info}).
 The parameters $\theta$ and $\nu$ are mutually orthogonal with respect to the Fisher information metric.
\end{corollary}

\subsection{Relationship with the entropic optimal transport and the Schr\"{o}dinger bridge problems} \label{subsection:transport}

Here,
we argue that
the minimum information dependence model (\ref{eq:min-info-without-nuisance}) 
with the marginal parameters fixed 
has a close relationship with the entropic optimal transport problem (e.g.,~\cite{PeyreCuturi2019}) and the Schr\"{o}dinger bridge problem (e.g.,~\cite{Leonard2012}). 
Fix $\varepsilon>0$, function $H:\mathcal{X}\to\mathbb{R}$,
and marginal densities $\{r_{i} \,\text{on}\, \mathcal{X}_{i}:i=1,\ldots,d\}$. 
The following convex optimization problem is called the (multi-marginal) entropic optimal transport problem $\mathrm{OT}_{\varepsilon}(H,\{r_{i}:i=1,\ldots,d\})$:
\begin{align}
\begin{split}
    & {\rm Minimize}\ \ 
    \mathcal{F}_{H,\varepsilon}(p):=
    -\int
    H(x)p(x){\rm d}x
    +
    \varepsilon\int p(x)\log\frac{p(x)}{\prod_{i=1}^{d}r_{i}(x_{i})}{\rm d}x,
    \\
    & {\rm subject\ to}\ \ p\in 
    \mathcal{M}(r_1,\ldots,r_d ).
\end{split}
        \label{eq:OT-entropic}
\end{align}
Given $s\in\mathcal{E}(H/\varepsilon)$,
the following is called the  Schr\"{o}dinger bridge problem $\mathrm{SB}_{\varepsilon}(H,\{r_{i}:i=1,\ldots,d\})$:
\begin{align}
\begin{split}
    & {\rm Minimize}\ \ 
     D( p, s),
    \\
    & {\rm subject\ to}\ \ p\in 
    \mathcal{M}(r_1,\ldots,r_d ).
\end{split}
        \label{eq:Schrodinger}
\end{align}
\begin{lemma}
\label{lem: relation to EOT}
A density $p_{\theta}$ in the minimum information dependence model (\ref{eq:min-info-without-nuisance})
with the canonical parameter $\theta$, the canonical statistics $h(x)$, and marginals $r_{i}$s is characterized by
the solution of both $\mathrm{OT}_{1} ( \theta^{\top}h, \{r_{i} : i=1 , \ldots , d \})$ and $\mathrm{SB}_{1}(\theta^{\top}h, \{r_{i}:i=1,\ldots,d\})$.
\end{lemma}
\begin{proof}
From Lemma~\ref{lemma:potential} in Subsection~\ref{subsection:potential}, we conclude that the density is given by
$\mathrm{OT}_{1} ( \theta^{\top} h , \{r_{i} : i=1 ,  \ldots , d \} )$.
From Theorem \ref{theorem:Pythagorean}, 
we conclude that
the density is given by 
$\mathrm{SB}_{1}(\theta^{\top}h, \{r_{i}:i=1,\ldots,d\})$.
\end{proof}

In the past several years, 
the theory and applications of these problems have gathered attention in many fields; for details, see \cite{PeyreCuturi2019}.
A computationally efficient iteration algorithm called the Sinkhorn--Knopp algorithm (\cite{SinkhornKnopp1967}) is used to find the optimizer of these problems. 
The number of marginal densities $r_i$ is typically $d=2$,
and
the case $d\geq 3$ is referred to as the multi-marginal optimal transport (e.g.\ \cite{Ruschendorf2002} and \cite{kim2014general}) or multi-marginal Schr\"{o}dinger problem.
Recently, the relationship between the  multi-marginal optimal transport and graphical models was pointed out by \cite{haasler2021}.

The optimizer of $\mathrm{OT}_{\varepsilon}(H,\{r_{i}:i=1,\ldots,d\})$ converges to that of the optimal transport problem as $\varepsilon\to 0$ under certain conditions (e.g.,~\cite{PeyreCuturi2019}):
\begin{align}
\begin{split}
    & {\rm Maximize}\ \ \int H(x)p(x){\rm d}x,
    \\
    & {\rm subject\ to}\ \ p\in \mathcal{M}(r_1,\ldots,r_d).
    \label{eq:OT-Monge}
\end{split}
\end{align}
From the convergence result, we get the range of expectation parameter $\eta_\theta=\int h(x)p_\theta(x){\rm d}x$ of the model (\ref{eq:min-info}).
Consider this in a one-parameter case $H(x)=\theta h(x)$.
By dividing the objective function in (\ref{eq:OT-Monge}) by $\theta$,
the convergence result implies that
$\eta_{\theta}$ is bounded by the optimal value of (\ref{eq:OT-Monge})
and can be numerically obtained by the entropic optimal transport $\mathrm{OT}_{1/\theta}(h, \{r_{i}:i=1,\ldots,d\})$
for sufficiently large $|\theta|$.
Such an argument for obtaining the range of the expectation for copulas was discussed in \cite{BedfordWilson2014}.
To compute Table~\ref{tab:Poisson} in Example~\ref{example:count}, the domain $\mathcal{X}_i$ was truncated to $\{0,\ldots,20\}$ and the Sinkhorn--Knopp algorithm was applied to find $p(x_1,x_2;\theta,\nu)$ for $\theta=\pm 10$.
In fact, the result in Table~\ref{tab:Poisson} can also be obtained by the exact solution of the optimal transport problem that maximizes (resp.\ minimizes) the correlation coefficient. Here, the exact solution has a monotone (resp.\ anti-monotone) support, which means $p(x)p(\tilde{x})>0$ only if $(\tilde{x}_1-x_1)(\tilde{x}_2-x_2)\geq 0$ (resp.\ $(\tilde{x}_1-x_1)(\tilde{x}_2-x_2)\leq 0$).

\section{Proof of the results}
\label{section:proofs}

This supplement provides the proofs of the theoretical results.

\subsection{Proof of Theorem~\ref{theorem:feasible}} \label{subsection:feasible-proof}

We rely on the following lemma. Recall that the space $\mathcal{X}=\prod_{i=1}^d\mathcal{X}_i$ is equipped with a measure ${\rm d}x=\prod_{i=1}^d {\rm d}x_i$.
For the proof, see Corollary 3.2 of \cite{Csiszar1975}, and Corollary 3.6 and Theorem 4.2 of \cite{Borwein_et_al1994}.

\begin{lemma} \label{lemma:Borwein}
Consider a function $k\in L_1({\rm d}x)$ that is positive almost everywhere with respect to ${\rm d}x$ (a.e.\ ${\rm d}x$).
Consider the following optimization problem:
\begin{align*}
\mbox{Minimize} &\ \ \int_{\mathcal{X}} u(x)\{\log u(x)-1\} k(x){\rm d}x
\\
\mbox{subject\ to} &\ \ \int_{\mathcal{X}_{-i}} u(x)k(x){\rm d}x_{-i} = r_i(x_i)\ \mbox{a.e.}\ {\rm d}x_i\ \mbox{for\ each}\ i,
\\
&\ \ u\in L_1(k(x){\rm d}x),\ \ u(x)\geq 0\ \mbox{a.e.}\ {\rm d}x,
\end{align*}
with the convention $0\log 0=0$.
If there exists a function $u(x)>0$ (a.e.\ $k(x)dx$) that satisfies the constraints and has a finite objective function value, then the problem has a unique optimal solution $u_*$ that is written as
\begin{align}
 u_*(x) = \prod_{i=1}^d f_i(x_i)
 \label{eq:Borwein-optimal}
\end{align}
for some positive measurable functions $f_1,\ldots,f_d$.
\end{lemma}

\begin{proof}[Proof of Theorem \ref{theorem:feasible}]
Suppose that $H(x)$ is integrable with respect to $p_0(x){\rm d}x$ and satisfies (\ref{eq:feasible-equivalent}).
Define $k(x)$ and $u(x)$ by $k(x)=p_0(x)e^{H(x)-\sum_{i=1}^{d}b_i(x_i)}$
and $u(x)=e^{-H(x)+\sum_{i=1}^{d}b_i(x_i)}$.
Then $k(x)$ is integrable due to (\ref{eq:feasible-equivalent}) and $u(x)$ satisfies all the requirements in Lemma~\ref{lemma:Borwein}.
Indeed, $u(x)k(x)$ has the marginal densities $r_i(x_i)$, $u(x)$ is positive almost everywhere, and the objective function
\[
 \int u(x)\{\log u(x)-1\}k(x){\rm d}x = \int p_0(x)\left(-H(x) + \sum_{i=1}^{d} b_i(x_i) - 1\right){\rm d}x
\]
is finite by assumption on $H$ and $b_i$.
Using the optimal solution $u_*(x)$ in (\ref{eq:Borwein-optimal}), define
\begin{align*}
 p_*(x) := k(x)u_*(x) 
 = p_0(x)
 \exp\left(
 H(x)-\sum_{i=1}^{d} b_i(x_i)
 +\sum_{i=1}^{d} \log f_i(x_i)
 \right).
\end{align*}
This is of the form (\ref{eq:min-info}).
Finally, the function $\sum_{i=1}^{d}\log f_i(x_i)$ belongs to $L_1(p_*(x){\rm d}x_i)$ because
\begin{align*}
\int u_*(x)\{\log u_*(x)-1\}k(x){\rm d}x
= \int p_*(x)\sum_{i=1}^{d}\log f_i(x_i){\rm d}x - 1
\end{align*}
is finite. Setting $\psi=\int p_*(x)\sum_{i=1}^{d}\{b_i(x_i)-\log f_i(x_i)\}{\rm d}x$ and $a_i(x_i)=b_i(x_i)-\log f_i(x_i)-\psi/d$, we deduce that $H(x)$ is feasible.

The uniqueness of both $\sum_{i=1}^{d} a_i(x_i)$ and $\psi$ follows from the uniqueness of $p_*(x)$. Indeed, $\psi$ is the expectation of
$-\log(p_*(x)/p_0(x))+H(x)$ with respect to $p_*(x)$
and $\sum_{i=1}^{d} a_i(x_i)$ is the residual.
\end{proof}

\subsection{Proofs of Corollaries \ref{corollary:polynomial}--\ref{corollary:Lipschitz}}
\label{subsection:feasible-examples-proof}

\begin{proof}[Proof of Corollary \ref{corollary:polynomial}]
 Since the set of moderately feasible functions is convex, it is enough to prove that a monomial $H(x)=c x_1^{\alpha_1}\cdots x_d^{\alpha_d}$ is moderately feasible, where the $\alpha_i$'s are non-negative integers and $c$ is a real number.
 By the arithmetic-geometric mean inequality,
 %\[
 % |x_1^{\alpha_1}\cdots x_d^{\alpha_d}|
 % \leq \frac{1}{d}(|x_1|^{\alpha_1 d}+\cdots+|x_d|^{\alpha_d d}),
 %\]
  we deduce that
 \[
  \int e^{H(x)-|c|(|x_1|^{\alpha_1 d}+\cdots|x_d|^{\alpha_d d})/d}p_0(x){\rm d}x \leq \int p_0(x){\rm d}x < \infty.\]
 Since $|x_i|^{\alpha_i d}$ is integrable with respect to $r_i(x_i){\rm d}x_i$ by assumption,
 $H$ is moderately feasible.
\end{proof}

\begin{proof}[Proof of Corollary \ref{corollary:Lipschitz}]
From the Lipschitz continuity of $H(x)$, we have
\[
H(x)\le H(x^{0}) 
    +Ld_{(p)}(x,x^{0})
    \le H(x^{0}) 
    +L\sum_{i=1}^{d}d_{i}(x_{i},x^{0}_{i}),
\]
which implies 
\[\int e^{
    H(x)-
    H(x^{0})-L\sum_{i=1}^{d}d_{i}(x_{i},x^{0}_{i})
    }p_0(x){\rm d}x \leq \int p_0(x){\rm d}x < \infty.
\]
By the integrability assumption on $d_{i}(x_{i},x^{0}_{i})$,
we obtain the conclusion.
\end{proof}

\subsection{Proof of Theorem~\ref{theorem:negligible}}
\label{subsection:negligible-proof}

Here, we give a proof of Theorem \ref{theorem:negligible}.
%In the proof,
%we drop the dependence on $\nu_{0}$ 
%in $p(x;\theta,\nu_{0})$,
%$\{a_{j}(x_{j};\theta,\nu_{0}):j=1,\ldots,d\}$,
%$\psi(\theta,\nu_{0})$,
%$g(M;\theta,\nu_{0})$,
%$\{r_{j}(x_{j};\nu_{0}):j=1,\ldots,d\}$,
%and 
%$p_{0}(x;\nu_{0})$
%because $\nu_{0}$ is fixed.

\begin{proof}[Proof of Theorem \ref{theorem:negligible}]

It suffices to show
for any subset $\mathrm{S}\subset \mathrm{N}$ containing $\nu_{0}$,
we have
\begin{align*}
    \Ep\left[\sup_{\theta\in\Theta,\nu\in\mathrm{S}}\frac{1}{n}
    \left|
    \log \frac{\prod_{t=1}^{n}p(x(t);\theta_{0},\nu_{0})}{\prod_{t=1}^{n}p(x(t);\theta,\nu)}
    -\log \frac{f(\pi\mid M;\theta_{0})}{f(\pi\mid M;\theta)}
    \right|
    \right]
    \le C
    \left\{
    \sup_{\nu\in\mathrm{S}}
    \tilde{D}(\nu_{0},\nu)
    +\varepsilon_{n}
    \right\}
\end{align*}
with some $C>0$ not depending on $n$.

The proof follows four steps.
We first construct an approximation of $g(M;\theta,\nu)$ using an $\varepsilon$-net,
second approximate the marginal density by the Stirling approximation,
third evaluate the approximated marginal density by the approximated Pythagorean relation,
and lastly evaluate the residuals.

\textbf{Preparation: lemmas.}
Before proving the theorem, 
we introduce lemmas for the proof:
\begin{enumerate}
\item[(1)] The first lemma shows that 
the Lipschitz continuity of $h(x)$ implies the Lipschitz continuity of the $a_{j}$'s;
\item[(2)] The second lemma shows that the boundedness of $h(x)$ implies the boundedness of the $a_{j}$'s;
\item[(3)] The third lemma states the finite sample evaluation of the Stirling approximation;
\item[(4)] The fourth lemma states the inequalities related to the log-sum-exp function;
\item[(5)] The fifth lemma states the deviation inequality in the $\ell_{1}$ distance;
\item[(6)] The sixth lemma presents the sensitivity of the entropic optimal transport problems with respect to marginals in the Kullback--Leibler divergence.
\end{enumerate}

\begin{lemma}
\label{lem: Lipschitz}
The Lipschitz continuity of $h(x)$ implies the Lipschitz continuity of the $a_{j}(x_{j};\theta,\nu)$'s with Lipschitz constant $\xi L_{h}$, where $\xi:=\sup_{\theta\in\Theta}\|\theta\|_{2}$.
\end{lemma}
\begin{proof}
Fix $j\in\{1,\ldots,d\}$.
From the constraint (\ref{eq:marginal-condition}),
we have
\begin{align*}
    e^{a_{j}(x_{j};\theta,\nu)}
    = \int e^{\theta^{\top} h(x)-\sum_{k\neq j}a_{k}(x_{k};\theta,\nu)-\psi(\theta,\nu)}
    \prod_{k\neq j}r_{k}(x_{k};\nu)\mathrm{d}x_{-j}.
\end{align*}
For $\tilde{x}_{j}\ne x_{j}$,
let $\tilde{x}:=(x_{1},\cdots,\tilde{x}_{j},\cdots,x_{d})$.
Let $d_{j}(\cdot,\cdot)$ be the distance in $\mathcal{X}_{j}$.
Together with the Lipschitz continuity of $h(x)$,
this implies
\begin{align*}
    e^{a_{j}(x_{j};\theta,\nu)}
    &\le 
     \int 
    e^{\xi L_{h}d_{j}(x_{j},\tilde{x}_{j})+\theta^{\top} h(\tilde{x})-\sum_{k\neq j}a_{k}(x_{k};\theta,\nu)-\psi(\theta,\nu)}
    \prod_{k\neq j}r_{k}(x_{k};\nu)\mathrm{d}x_{-j}
    \nonumber\\
    &=
    e^{\xi L_{h}d_{j}(x_{j},\tilde{x}_{j})}
    e^{a_{j}(\tilde{x}_{j};\theta,\nu)},
\end{align*}
which completes the proof.
\end{proof}

\begin{lemma}
\label{lem: boundedness of a}
The 
adjusting 
functions 
and the potential function are bounded as
\[
\left\| \sum_{j=1}^{d} a_{j} ( \cdot ; \theta,\nu ) \right\|_{\infty}\le 2d
\| \theta^{\top}h(\cdot) \|_{\infty}
\, and\, |\psi(\theta,\nu)|\le (2d+1)\| \theta^{\top}h(\cdot) \|_{\infty}.
\]

\end{lemma}
\begin{proof}
Let $H_{\infty}:=\|\theta^{\top}h(\cdot)\|_{\infty}$.
Fix $j\in\{1,\ldots,d\}$.
From the constraint  (\ref{eq:marginal-condition}),
we have
\begin{align*}
    e^{a_{j}(x_{j};\theta,\nu)}
    &= \int e^{\theta^{\top} h(x)-\sum_{k\neq j}a_{k}(x_{k};\theta,\nu)-\psi(\theta,\nu)}
    \prod_{k\neq j}r_{k}(x_{k};\nu)\mathrm{d}x_{-j}
    \\
    &\le 
    e^{H_{\infty}+\lambda_{j}}
    \ \text{with}\ 
    e^{\lambda_{j}}:=
    \int e^{-\sum_{k\neq j}a_{k}(x_{k};\theta,\nu)-\psi(\theta,\nu)}
    \prod_{k\neq j}r_{k}(x_{k};\nu )\mathrm{d}x_{-j}.
\end{align*}
Likewise, we have
$e^{a_{j}(x_{j};\theta,\nu)}\ge e^{-H_{\infty}+\lambda_{j}}$.
Then, we obtain \[-H_{\infty}+\lambda_{j}\le a_{j}(x_{j};\theta,\nu) \le
H_{\infty}+\lambda_{j}.\]
Together with the constraint (\ref{eq:zero-mean-condition}),
taking the sum of this over $j$ and the expectation with respect to $p(x;\theta,\nu)$ yields
$-dH_{\infty}\le \sum_{j=1}^{d}\lambda_{j}\le dH_{\infty}$, which implies
\begin{align*}
-2dH_{\infty}
\le
- dH_{\infty}- \sum_{j=1}^{d}\lambda_{j}
\le 
 \sum_{j=1}^{d}a_{j}(x_{j};\theta,\nu) \le dH_{\infty}+\sum_{j=1}^{d}\lambda_{j}
 \le 2dH_{\infty}.
\end{align*}
This together with the definition of $\psi(\theta,\nu)$ completes the proof.
\end{proof}

\begin{lemma}[e.g., \cite{Robbins1955}]
\label{lem: Stirling}
For $l\in\mathbb{N}$, we have
\begin{align*}
    l^{l}
    e^{-l}(2\pi l)^{1/2}\left(1+\frac{1}{12l+1}\right)\le l! \le 
    l^{l}
    e^{-l}(2\pi l)^{1/2}\left(1+\frac{1}{12l}\right).
\end{align*}
\end{lemma}

\begin{lemma}
\label{lem: log sum exp}
For $l\in\mathbb{N}$, $\{x_{j}\in\mathbb{R}:j=1,\ldots,l\}$, and $\beta>0$, we have
\begin{align*}
    \max_{j=1,\ldots,l}\{x_{j}\}
    \le 
    \frac{1}{\beta}\log \sum_{j=1}^{l}\exp(\beta x_{j})
    \le
    \max_{j=1,\ldots,l}\{x_{j}\}+\frac{1}{\beta}\log l.
\end{align*}
\end{lemma}
\begin{proof}
The left-hand side relation follows 
from $\exp (\max_{j=1,\ldots,l}\{x_{j}\})\le 
\sum_{j=1}^{l}\exp(x_{j})$.
The right-hand side relation follows
from $\sum_{j=1}^{l}\exp(\beta x_{j})\le l \exp(\beta\max_{j=1,\ldots,l}\{x_{j}\})$.
\end{proof}

\begin{lemma}
\label{lem: L1 and KL deviation inequality}
Let $l$ be a positive integer greater than 2 and  $Y=(Y_{j})_{j=1}^{l}$ be a random vector following the multinomial distribution with $n$ trials and event probabilities $p=(p_{j})_{j=1}^{l}$. 
Let $D(p^{1},p^{2})$ be the Kullback--Leibler divergence from $p^{1}$ to $p^{2}$, where $p^{1}$ and $p^{2}$ are probability densities on $[l]=\{1,\ldots,l\}$.
Then, 
for $\delta>(6l/n)^{1/2}$, we have
\begin{align*}
    \Pr\left(\|Y/n-p\|_{1}\ge \delta\right) &\le 2\exp(-n\delta^{2}/(2l)) \ \text{and}\\
    \Pr\left(D(Y/n,p)\ge \delta^{2}\right) 
    &\le 2\exp(-n\delta^{2}/(2l)).
\end{align*}
\end{lemma}
\begin{proof}
From Theorem 2.1 in \cite{Weissman2003}, we have
\begin{align*}
    \Pr(\|Y/n-p\|_{1}\ge \delta)
    \le 2^{l}e^{-n\delta^{2}\varphi(\pi_{p})/4},
\end{align*}
where $\varphi(a):=\{1/(1-2a)\}\log\{(1-a)/a\}$ for $a\in[0,1/2)$ and $\varphi(1/2)=2$,
and $\pi_{p}$ is some point in $[0,1/2]$.
Since $\varphi(a)\ge 2$ for $a\in [0,1/2]$, we have
\begin{align*}
    \Pr(\|Y/n-p\|_{1}\ge \delta)
    \le 2^{l}e^{-n\delta^{2}/2}.
\end{align*}
Observe
that for $\delta\ge (4l/n)^{1/2}$, we have
\begin{align*}
    (l-1)\log 2 \le l\le \frac{1}{4}n\delta^{2} 
    \le \frac{1}{2}(1-1/l)n\delta^{2}.
\end{align*}
This implies 
\begin{align*}
    \Pr(\|Y/n-p\|_{1}\ge \delta)
    \le 2^{l}e^{-n\delta^{2}/2}
    \le 2e^{-n\delta^2 / (2l)}.
\end{align*}

From Theorem I.2 in \cite{Agrawal2020},
we have
\begin{align*}
    \Pr(D(Y/n,p)\ge \delta^{2})\le 
    e^{-n\delta^{2}+(l-1)\log \{(en\delta^{2})/(l-1)\}}.
\end{align*}
Observe that for $\delta \ge (6l/n)^{1/2}$, we have
\begin{align*}
    n\delta^{2} - (l-1)\log\frac{en\delta^{2}}{l-1}
    >\frac{n\delta^{2}}{2} -\log 2
    >\frac{1}{2}\frac{n\delta^{2}}{l-1} -\log 2
    >
    \frac{1}{2}\frac{n\delta^{2}}{l} -\log 2
\end{align*}
from the inequality
$x > 2+2\log x \ \text{for}\ x>6$.
So, we get
\begin{align*}
    \Pr(D(Y/n,p)\ge \delta^{2})\le 
    2e^{-n\delta^{2}/(2l)},
\end{align*}
which completes the proof.
\end{proof}

\begin{lemma}
\label{lem: Stability}
Let $I=\prod_{j=1}^{d}I_{j}$ be a finite set
and let $H$ be a function on $I$.
Let $D(\cdot,\cdot)$ be the Kullback--Leibler divergence between probability densities on $I$.
For a probability density $\pi$ on $I$,
probability densities $\{\pi_{j}\}$ on $I_{j}$s,
and 
a function H on $I$,
define
\begin{align*}
    \mathcal{S}(\pi;\{\pi_{j}\},H)
    :=\sum_{i\in I}H(i)\pi(i)-\sum_{i\in I}\pi(i)\log \frac{\pi(i)}{\prod_{j=1}^{d}\pi_{j}(i_{j})}.
\end{align*}
Let $\{p_{j}:j=1,\ldots,d\}$ and $\{q_{j}:j=1,\ldots,d\}$
be probability densities on $I_{j}$s, respectively.
Let $p$ be a density maximizing
$\mathcal{S}(\pi;\{p_{j}\},H)$
with respect to $\pi$ of which
the $j$-th marginal density is $p_{j}$ for each $j$.
Let $q$ be a density maximizing
$\mathcal{S}(\pi;\{q_{j}\},H)$
with respect to $\pi$ of which the $j$-th marginal density is $q_{j}$ for each $j$.
Then, we have
\begin{align*}
    D(p,q)
\le
(4d+2)
\max_{i\in I}|H(i)|\sum_{j=1}^{d}\|p_{j}-q_{j}\|_{1}
+
    \sum_{j=1}^{d}D(p_{j},q_{j}).
\end{align*}
\end{lemma}
\begin{proof}
Denote by $a^{p}_{j}(i_{j})$ the 
adjusting 
functions of $p$,
and denote by $\psi^{p}$ the potential function of $p$, respectively.
Denote by $a^{q}_{j}(i_{j})$ the 
adjusting
functions of $q$,
and denote by $\psi^{q}$ the potential function of $q$, respectively.

Observe that we have
\begin{align*}
    D(p,q)=\Ep_{p}
    \left[
    \left\{\sum_{j=1}^{d} a^{q}_{j}(i_{j})+\psi^{q}
    \right\}
    -
    \left\{\sum_{j=1}^{d} a^{p}_{j}(i_{j})+\psi^{p}
    \right\}
    \right]
    +\sum_{j=1}^{d}D(p_{j},q_{j}),
\end{align*}
where $\Ep_{p}$ denotes the expectation with respect to $p$.
From Lemma \ref{lem: boundedness of a}, the boundedness of $H$ implies
\begin{align*}
    \sum_{j=1}^{d}a^{q}_{j}(i_{j})+\psi^{q}\le
    A:=(4d+1)\max_{i\in I}|H(i)|.
\end{align*}
Then, we get
\begin{align*}
        &\Ep_{p}\left[
    \left\{\sum_{j=1}^{d} a^{q}_{j}+\psi^{q}
    \right\}
    -
    \left\{\sum_{j=1}^{d} a^{p}_{j}+\psi^{p}
    \right\}
    \right]\\
    &\le A\sum_{j=1}^{d}\|p_{j}-q_{j}\|_{1}
    +\Ep_{q}\left[
    \left\{\sum_{j=1}^{d} a^{q}_{j}+\psi^{q}
    \right\}\right]
    -\Ep_{p}\left[
    \left\{\sum_{j=1}^{d} a^{p}_{j}+\psi^{p}
    \right\}\right]\\
    &=
    A\sum_{j=1}^{d}\|p_{j}-q_{j}\|_{1}
    +
    \mathcal{S}(q;\{q_{j}\},H)-
    \mathcal{S}(p;\{p_{j}\},H),
\end{align*}
where the last equation follows since 
\begin{align*}
    \mathcal{S}(p;\{p_{j}\},H)
    =\Ep_{p}\left[\sum_{j=1}^{d}a^{p}_{j}(i_{j})+\psi^{p}\right].
\end{align*}
From Theorem 3.6 (ii) of \cite{EcksteinNutz_arXiv} stating
\begin{align*}
    |\mathcal{S}(q;\{q_{j}\},H)-
    \mathcal{S}(p;\{p_{j}\},H)|
    \le \max_{i\in I}|H(i)|
    \sum_{j=1}^{d}\|p_{j}-q_{j}\|_{1},
\end{align*}
we obtain
\begin{align*}
    \Ep_{p}\left[
    \left\{\sum_{j=1}^{d} a^{q}_{j}+\psi^{q}
    \right\}
    -
    \left\{\sum_{j=1}^{d} a^{p}_{j}+\psi^{p}
    \right\}
    \right]
    \le \left(A+\max_{i\in I}|H(i)|\right)\sum_{j=1}^{d}\|p_{j}-q_{j}\|_{1},
\end{align*}
which gives
\begin{align*}
D(p,q)
\le
\left(A+\max_{i\in I}|H(i)|\right)\sum_{j=1}^{d}\|p_{j}-q_{j}\|_{1}
+
    \sum_{j=1}^{d}D(p_{j},q_{j})
\end{align*}
and completes the proof.
\end{proof}

\textbf{First step: Approximation of the marginal likelihood of $M$}.
We begin with providing an approximation of the marginal likelihood $g(M;\theta,\nu)$ of $M$.
To do so, we prepare a partition of the sample space as follows:
Fix $\varepsilon>0$. 
\begin{enumerate}
\item For $j=1,\ldots,d$,
let $I_{j}(\varepsilon)=\mathcal{N}(\mathcal{X}_{j},d_{i},\varepsilon)$;
\item Let $I(\varepsilon):=[I_{1}(\varepsilon)]\times \cdots \times [I_{d}(\varepsilon)]
=
\{1,\ldots,I_{1}(\varepsilon)\}\times \cdots \times \{1,\ldots,I_{d}(\varepsilon)\}$;
\item
For $j=1,\ldots,d$,
let $\mathcal{X}_{j}^{\varepsilon}:=\{y_{j,k}^{\varepsilon}:k\in[I_{j}(\varepsilon)]\}$ be an $\varepsilon$-net with cardinality $I_{j}(\varepsilon)$;
\item Let $\mathcal{X}^{\varepsilon}:=\mathcal{X}_{1}^{\varepsilon}\times \cdots \times \mathcal{X}_{d}^{\varepsilon}$;
\item For $j=1,\ldots,d$,
we denote the ball in $\mathcal{X}_{j}$ with center $x_{j}\in\mathcal{X}_{j}$ and radius $r>0$ by $B_{j}(x_{j},r)$;
\item For $j=1,\ldots,d$, 
create a partition $\mathcal{P}_{j}:=\sqcup_{k=1}^{I_{j}(\varepsilon)} D_{j,k}^{\varepsilon}$ of $\mathcal{X}_{j}$ as follows:
Set $D_{j,1}^{\varepsilon}=B_{j}(y_{j,1}^{\varepsilon},\varepsilon)$.
For $k>2$, set $D_{j,k}^{\varepsilon}=B_{j}(y_{j,k}^{\varepsilon},\varepsilon)\setminus (\sqcup_{k'=1,\ldots,k-1} D_{j,k'}^{\varepsilon})$ sequentially;
\item Let $\mathcal{P}:=\sqcup_{i\in I(\varepsilon)}D_{i}^{\varepsilon}$ be a partition of $\mathcal{X}$,
where
$D_{i}^{\varepsilon}:=\prod_{j=1}^{d}D_{j,i_{j}}^{\varepsilon}$
for $i=(i_{1},\ldots,i_{d})\in I(\varepsilon)$.
\end{enumerate}

Using this partition,
we approximate the density $p(x; \theta,\nu)$ 
and the marginal density $g(M;\theta)$
as follows.
Let $H_{\theta}(x)=\theta^{\top}h(x)$ and let
\begin{align*}
    H^{\varepsilon}_{\theta}(x)&:=\sum_{i\in I(\varepsilon)}H_{\theta}(y_{i}^{\varepsilon})1_{D_{i}^{\varepsilon}}(x)
    \ \text{and}
    \\
    a^{\varepsilon}_{j}(x_{j};\theta,\nu)&:=
    \sum_{i_{j}\in [I_{j}(\varepsilon)]}a_{j}(y_{j,i_{j}}^{\varepsilon};\theta,\nu)1_{D_{j,i_{j}}^{\varepsilon}}(x_{j})
    \quad\text{for}\quad j=1,\ldots,d.
\end{align*}
Let $\psi^{\varepsilon}(\theta,\nu)$ and $p^{\varepsilon}(x;\theta,\nu)$ be given by
\begin{align*}
    \psi^{\varepsilon}(\theta,\nu)
    &=\log \int \exp\left(H^{\varepsilon}_{\theta}(x)-\sum_{j=1}^{d}a^{\varepsilon}_{j}(x_{j};\theta,\nu)\right)
    p_{0}(x)\mathrm{d}x \quad\text{and}
    \\
    p^{\varepsilon}(x;\theta,\nu)
    &=\exp\left(H^{\varepsilon}_{\theta}(x)-\sum_{j=1}^{d}a^{\varepsilon}_{j}(x_{j};\theta,\nu)-\psi^{\varepsilon}(\theta,\nu)
    \right)
    p_{0}(x),
\end{align*}
respectively.
Note that $p^{\varepsilon}(x;\theta,\nu)$ 
is a density in a minimum information dependence model but its marginal densities $\{r^{\varepsilon}_{j}(x_{j};\theta,\nu): j=1,\ldots,d\}$ are not necessarily the same as those of $p(x;\theta,\nu)$.
We denote
the corresponding marginal distribution of $M$ by $g^{\varepsilon}(M;\theta,\nu)$.

Here we show $\log g^{\varepsilon}(M;\theta,\nu)$
is a ($2L\varepsilon$)-approximation of $\log g(M;\theta,\nu)$
with $L:=(d+1)\xi L_{h}$.
Fix $j=1,\ldots,d$.
For any $x_{j}\in\mathcal{X}_{j}$, there exists $i_{j}\in [I_{j}(\varepsilon)]$ such that $d_{j}(x_{j},y_{j,i_{j}}^{\varepsilon})\le \varepsilon$.
Therefore, we have
\begin{align*}
    |a_{j}(x_{j};\theta,\nu)-a^{\varepsilon}_{j}(x_{j};\theta,\nu)|
    =|a_{j}(x_{j};\theta,\nu)-a_{j}(y_{j,i_{j}}^{\varepsilon};\theta,\nu)|
    \le \xi L_{h} d_{j}(x_{j},y_{j,i_{j}}^{\varepsilon})
    \le \xi L_{h} \varepsilon,
\end{align*}
where the first inequality follows from Lemma \ref{lem: Lipschitz}.
By a similar reasoning to the above, we have
\begin{align*}
    |H_{\theta}(x)-H^{\varepsilon}_{\theta}(x)|
    &\le 
    \sup_{\theta\in\Theta}\|\theta\|_{2}
    \|h(x)-h(y_{i}^{\varepsilon})\|_{2}
    \le 
    \xi L_{h}\varepsilon.
\end{align*}
Together with the identity
\begin{align*}
    \psi(\theta,\nu)-\psi^{\varepsilon}(\theta,\nu)
    &=\log \int e^{H_{\theta}(x)-H^{\varepsilon}_{\theta}(x)+H^{\varepsilon}_{\theta}(x)-\sum_{j=1}^{d}a_{j}+\sum_{j=1}^{d}a^{\varepsilon}_{j}-\sum_{j=1}^{d}a^{\varepsilon}_{j}}
    p_{0}\mathrm{d}x
    \\
    &
    \quad-
    \log \int e^{H^{\varepsilon}_{\theta}(x)-\sum_{j=1}^{d}a^{\varepsilon}_{j}(x_{j};\theta,\nu)}p_{0}\mathrm{d}x,
\end{align*}
these inequalities imply 
\begin{align*}
    |\psi(\theta,\nu)-\psi^{\varepsilon}(\theta,\nu)|
    \le L\varepsilon.
\end{align*}
Consequently, we have
\begin{align}
    &|\log p(x;\theta,\nu)- \log p^{\varepsilon}(x;\theta,\nu) |\nonumber\\
    &=\left|H_{\theta}(x)-H^{\varepsilon}_{\theta}(x)-\sum_{j=1}^{d}a_{j}(x_{j};\theta,\nu)+\sum_{j=1}^{d}a^{\varepsilon}_{j}(x_{j};\theta,\nu)-\psi(\theta,\nu)+\psi^{\varepsilon}(\theta,\nu)\right|
    \nonumber\\
    &\le 2L\varepsilon.
    \label{eq: difference between p and ptilde}
\end{align}
Thus, we obtain
\begin{align}
    \left|\frac{1}{n}\log \frac{g(M;\theta,\nu)}{g^{\varepsilon}(M;\theta,\nu)}\right|
    &=
    \left|\frac{1}{n}\log \frac{\sum_{\pi\in \mathbb{S}_n^d}
    \prod_{t=1}^{n}p(x(t);\theta,\nu) }{\sum_{\pi\in \mathbb{S}_n^d}
    \prod_{t=1}^{n}p^{\varepsilon}(x(t);\theta,\nu) }\right|
    \nonumber\\
    &\le \left| \frac{1}{n}\log \exp(2L\varepsilon n)\right|
    =2L\varepsilon.
    \label{eq: approximation of g}
\end{align}

\textbf{Second step: asymptotic expansion of the marginal likelihood ratio.}
We next give an asymptotic expansion of $\log g^{\varepsilon}(M;\theta_{0},\nu)/ g^{\varepsilon}(M;\theta,\nu)$ by using the following contingency table:
\begin{enumerate}
    \item For $i=(i_{j})_{j=1}^{d}\in I(\varepsilon)$, let $N(i):=|\{t\in [n]:x(t)\in\prod_{j=1}^{d}D^{\varepsilon}_{j,i_{j}}\}|$.
    \item Then, $N=(N(i))_{i\in I(\varepsilon)}$ is an $I_{1}(\varepsilon)\times I_{2}(\varepsilon)\times \cdots \times I_{d}(\varepsilon)$-frequency table.
    \item For $i=(i_{j})_{j=1}^{d}\in I(\varepsilon)$, the expected frequency of $N(i)$ is
    \[
    p^{*}(i):=\int_{D_{i}^{\varepsilon}}p(x;\theta_0, \nu_{0})\mathrm{d}x.
    \]
    \item For $j=1,\ldots,d$, we denote the marginal frequency of $i_{j}\in [I_{j}(\varepsilon)]$ by $N_{j}(i_{j})$.
    \item 
    For $j=1,\ldots,d$ and $i_{j}\in[I_{j}(\varepsilon)]$,
    the expected frequency of $N_{j}(i_{j})$ is 
    \[
    p^{*}_{j}(i_{j})
    :=\int_{D_{j,i_{j}}^{\varepsilon}} r_{j}(x_{j};\nu_{0})\mathrm{d}x_{j}
    .
    \]
    \item We denote by $\mathcal{T}(\varepsilon)$
the set of $I_{1}(\varepsilon)\times I_{2}(\varepsilon)\times \cdots \times I_{d}(\varepsilon)$-frequency tables that have the same marginal frequencies as $N$.
\end{enumerate}
Define the probability density
$p^{\varepsilon}_{\theta,\nu}$ on $I(\varepsilon)$ as
\begin{align*}
    p^{\varepsilon}_{\theta,\nu}(i):=\int_{D^{\varepsilon}_{i}}
    p^{\varepsilon}(x;\theta,\nu)
    \mathrm{d}x 
    \ \ \text{for}\ \  
    i\in I(\varepsilon).
\end{align*}
Note that we have
\begin{align*}
    p^{\varepsilon}_{\theta,\nu}(i)
    &=
    \exp\left\{H_{\theta}(y_{i}^{\varepsilon})-\sum_{j=1}^{d}
    a_{j}\left(y_{j,i_{j}}^{\varepsilon};\theta,\nu\right)-\psi^{\varepsilon}(\theta,\nu)\right\}
    \prod_{j=1}^{d}\int_{D^{\varepsilon}_{j,i_{j}}}r_{j}(x_{j};\nu)\mathrm{d}x_{j}
\end{align*}
and this forms a minimum information dependence model with the canonical statistics $h^{\varepsilon}:i\mapsto h(y_{i}^{\varepsilon})$. 
Then, observe that 
for any $x\in \mathcal{X}$,
there exists $i(x)\in I(\varepsilon)$ such that 
$x\in D_{i(x)}^{\varepsilon}$,
and
we have
\begin{align*}
    p^{\varepsilon}(x;\theta,\nu)
    &=\exp\left\{H_{\theta}(y_{i(x)}^{\varepsilon})-\sum_{j=1}^{d}
    a_{j}\left(y_{j,i_{j}(x)}^{\varepsilon};\theta,\nu\right)-\psi^{\varepsilon}(\theta,\nu)\right\}
    \prod_{j=1}^{d}r_{j}(x_{j};\nu)\\
    &=p^{\varepsilon}_{\theta,\nu}(i(x))
    \frac{\prod_{j=1}^{d}r_{j}(x_{j};\nu)}{\prod_{j=1}^{d}\int_{D^{\varepsilon}_{j,i_{j}(x)}}r_{j}(x_{j};\nu)\mathrm{d}x_{j}}.
\end{align*}
This implies
\begin{align*}
\prod_{t=1}^{n}p^{\varepsilon}(x(t);\theta,\nu)
&=\mathcal{R}_{0}(M)
\prod_{i\in I(\varepsilon)}p^{\varepsilon}_{\theta,\nu}(i)^{N(i)},
\end{align*} 
where $\mathcal{R}_{0}(M)$ is defined by
\begin{align*}
    \mathcal{R}_{0}(M;\nu)
    :=
    \prod_{t=1}^{n}
    \frac{\prod_{j=1}^{d}r_{j}(x_{j}(t);\nu)}{\prod_{j=1}^{d}\int_{D^{\varepsilon}_{j,i_{j}(x(t))}}r_{j}(x_{j};\nu)\mathrm{d}x_{j}}
\end{align*}
and depends only on $M$, $\nu$, and $\varepsilon$.
This implies 
\begin{align*}
    g^{\varepsilon}(M;\theta,\nu)
    &=\mathcal{R}_{0}(M;\nu)
    \sum_{\pi\in\mathbb{S}^{d}_{n}}\prod_{i\in I(\varepsilon)}
p^{\varepsilon}_{\theta,\nu}(i)^{N(i)}
    \\
    &=\mathcal{R}_{0}(M;\nu)\sum_{\tilde{N}\in\mathcal{T}(\varepsilon)}C(\tilde{N})\prod_{i\in I(\varepsilon) }p^{\varepsilon}_{\theta,\nu}(i)^{\tilde{N}(i)},
\end{align*}
where $C(\tilde{N})$ is the number of permutations such that the frequency table is $\tilde{N}$.
Note that we have \[C(\tilde{N})=n!\prod_{j=1}^{d}\prod_{i_{j}\in [I_{j}(\varepsilon)]}\tilde{N}_{j}(i_{j})!\Bigg/\prod_{i\in I(\varepsilon)}\tilde{N}(i)!.\]
This yields
\begin{align}
    g^{\varepsilon}(M;\theta,\nu)
    =
    \mathcal{R}_{0}(M;\nu)\sum_{\tilde{N}\in\mathcal{T}(\varepsilon)}\frac{n!\prod_{j=1}^{d}\prod_{i_{j}\in [I_{j}(\varepsilon)]}\tilde{N}_{j}(i_{j})!}{\prod_{i\in I(\varepsilon)}\tilde{N}(i)!}\prod_{i\in I(\varepsilon)}p^{\varepsilon}_{\theta,\nu}(i)^{\tilde{N}(i)}.
    \label{eq: hypergeom and tildeg}
\end{align}

We then expand $g^{\varepsilon}(M;\theta,\nu)$ by using the finite sample evaluation of the Stirling approximation (Lemma \ref{lem: Stirling}).
For $\tilde{N}\in\mathcal{T}(\varepsilon)$ with $\tilde{N}(i)\ne 0$ for all $i\in I$,
Lemma \ref{lem: Stirling} gives 
\begin{align}
    &\prod_{j,i_{j}}\tilde{N}_{j}(i_{j})!
    \le e^{\sum_{j,i_{j}}\tilde{N}_{j}(i_{j})\log \tilde{N}_{j}(i_{j})-nd}\prod_{j,i_{j}}\left(2\pi\tilde{N}_{j}(i_{j})\right)^{\frac{1}{2}}\left(1+\frac{1}{12\tilde{N}_{j}(i_{j})}\right),
    \label{eq: Ntildejij}
\end{align}
where the index $j$ ranges from $1$ to $d$ and the index $i_{j}$ ranges from 1 to $I_{j}(\varepsilon)$. 
Also, we have
\begin{align}
    \prod_{i\in I(\varepsilon)}\tilde{N}(i)!
    \ge
    e^{\sum_{i\in I(\varepsilon)}\tilde{N}(i)\log \tilde{N}(i)-n}\prod_{i\in I(\varepsilon)}\left(2\pi\tilde{N}(i)\right)^{\frac{1}{2}}\left(1+\frac{1}{12\tilde{N}(i)+1}\right).
    \label{eq: Ntildei}
\end{align}
Let $D(p,q)$ be the Kullback--Leibler divergence from a probability density $p$ on $I(\varepsilon)$ to a probability density $q$ on $I(\varepsilon)$.
Then, we have
\begin{align}
    \prod_{i\in I(\varepsilon)}p^{\varepsilon}_{\theta,\nu}(i)^{\tilde{N}(i)}
    =
    e^{-nD(\hat{p},p^{\varepsilon}_{\theta,\nu})-n\log n+\sum_{i\in I(\varepsilon)}\tilde{N}(i)\log\tilde{N}(i)
    }.
    \label{eq: expression of pep}
\end{align}
These three inequalities (\ref{eq: Ntildejij})--(\ref{eq: expression of pep}), together with the trivial inequality $0^{0}\mathrm{e}^{-0}\le 0! \le 0^{0}\mathrm{e}^{-0}$, imply
\begin{align*}
    g^{\varepsilon}(M;\theta,\nu)\le 
    n!
    e^{nd\log n-nd+n}
    \mathcal{R}_{0}(M;\nu)
    \sum_{\tilde{N}\in\mathcal{T}(\varepsilon)}
    e^{nQ(\hat{p},p^{\varepsilon}_{\theta,\nu})}f_{1}(\tilde{N}),
\end{align*}
where $Q(p,q)$ for two probability densities $p$ and $q$ on $I(\varepsilon)$, with $p$ having marginals $r_{j}$, is defined by
\begin{align*}
    Q(p,q):=-D(p,q)+\sum_{j=1}^{d}\sum_{i_{j}\in I_{j}(\varepsilon)}r_{j}(i_{j})\log r_{j}(i_{j})
\end{align*}
taking $0\log 0 = 0$,
and $f_{1}(\tilde{N})$ is defined by
\begin{align*}
    f_{1}(\tilde{N}):=\frac{\prod_{j\in [d]\,,\,i_{j}\in[I_{j}(\varepsilon)]:\tilde{N}_{j}(i_{j})>0}\left(2\pi\tilde{N}_{j}(i_{j})\right)^{1/2}\left(1+\frac{1}{12\tilde{N}_{j}(i_{j})}\right)}
    {\prod_{i\in I(\varepsilon):\tilde{N}(i)>0}\left(2\pi\tilde{N}(i)\right)^{1/2}\left(1+\frac{1}{12\tilde{N}(i)+1}\right)}.
\end{align*}
Since there exists an absolute constant $C_{1}>0$ such that
\begin{align*}
f_{1}(\tilde{N})
\le (C_{1}n)^{\sum_{j=1}^{d}I_{j}(\varepsilon)/2},
\end{align*}
we have
\begin{align}
    g^{\varepsilon}(M;\theta,\nu)\le
    \mathcal{R}(M;\nu)
    (C_{1}n)^{
    \sum_{j=1}^{d}I_{j}(\varepsilon)/2}
    \sum_{\tilde{N}\in\mathcal{T}(\varepsilon)}
    e^{nQ(\hat{p},p^{\varepsilon}_{\theta,\nu})},
    \label{eq: upper bound of gtilde}
\end{align}
where $\mathcal{R}(M;\nu)$ is defined by
\[
\mathcal{R}(M;\nu):=
n! e^{nd\log n-nd+n}
    \mathcal{R}_{0}(M;\nu).
\]
Likewise, we have, for an absolute constant $C_{2}>0$ such that
\begin{align}
    g^{\varepsilon}(M;\theta,\nu)\ge
    \mathcal{R}(M;\nu)
    (C_{2}n)^{-\sum_{j=1}^{d}I_{j}(\varepsilon)/2}\sum_{\tilde{N}\in\mathcal{T}(\varepsilon)}
    e^{nQ(\hat{p},p^{\varepsilon}_{\theta,\nu})}.
    \label{eq: lower bound of gtilde}
\end{align}

We lastly evaluate 
$(1/n)\log g^{\varepsilon}(M;\theta_{0},\nu_{0})/g^{\varepsilon}(M;\theta,\nu)$ by applying the inequality of the log-sum-exp function (Lemma \ref{lem: log sum exp}).
Lemma \ref{lem: log sum exp} gives
\begin{align*}
    \max_{\tilde{N}\in\mathcal{T}(\varepsilon)}Q(\hat{p},p^{\varepsilon}_{\theta,\nu})
    \le
    \frac{1}{n}\log \sum_{\tilde{N}\in\mathcal{T}(\varepsilon)}
    e^{nQ(\hat{p},p^{\varepsilon}_{\theta,\nu})}
    \le
    \max_{\tilde{N}\in\mathcal{T}(\varepsilon)}Q(\hat{p},p^{\varepsilon}_{\theta,\nu})
    +\frac{\log |\mathcal{T}(\varepsilon)|}{n}.
\end{align*}
Together with (\ref{eq: upper bound of gtilde}) and (\ref{eq: lower bound of gtilde}), this implies 
\begin{align}
    &\frac{1}{n}\log \frac{g^{\varepsilon}(M;\theta_{0},\nu_{0})}{g^{\varepsilon}(M;\theta,\nu)}
    \nonumber\\
    &\le
    \max_{\tilde{N}\in\mathcal{T}(\varepsilon)}Q(\hat{p},p^{\varepsilon}_{\theta_{0},\nu_{0}})
    -
    \max_{\tilde{N}\in\mathcal{T}(\varepsilon)}Q(\hat{p},p^{\varepsilon}_{\theta,\nu})
    +\frac{1}{n}\log\frac{\mathcal{R}_{0}(M;\nu_{0})}{\mathcal{R}_{0}(M;\nu)}
    +R_{1}(n,\varepsilon)
    \nonumber\\
    &=\inf_{p\in\mathcal{M}(\{\hat{p}_{j}\})}D(p,p^{\varepsilon}_{\theta,\nu})
    -\inf_{p\in\mathcal{M}(\{\hat{p}_{j}\})}D(p,p^{\varepsilon}_{\theta_{0},\nu_{0}})
    +\frac{1}{n}\log\frac{\mathcal{R}_{0}(M;\nu_{0})}{\mathcal{R}_{0}(M;\nu)}
    +R_{1}(n,\varepsilon),
    \label{eq: upper bound of marginal likelihood ratio}
\end{align}
where $\mathcal{M}(\{\hat{p}_{j}\})$ is the set of probability densities on $I(\varepsilon)$ with marginal densities $(\hat{p}_{j}=N_{j}/n)_{j=1}^{d}$
and $R_{1}(n,\varepsilon)$ is defined by
\begin{align*}
    R_{1}(n,\varepsilon)
    :=
    \frac{\log |\mathcal{T}(\varepsilon)|}{n}
    +\frac{\sum_{j=1}^{d}I_{j}(\varepsilon)}{2}\frac{\log (C_{1}C_{2}n^{2})}{n}
    .
\end{align*}
Also, we get
\begin{equation}
    \begin{split}
    \frac{1}{n}\log \frac{g^{\varepsilon}(M;\theta_{0},\nu_{0})}{g^{\varepsilon}(M;\theta,\nu)}
    &\ge
    \inf_{p\in\mathcal{M}(\{\hat{p}_{j}\})}D(p,p^{\varepsilon}_{\theta,\nu})
    -\inf_{p\in\mathcal{M}(\{\hat{p}_{j}\})}D(p,p^{\varepsilon}_{\theta_{0},\nu_{0}})\\
    &\quad\quad
    +\frac{1}{n}\log\frac{\mathcal{R}_{0}(M;\nu_{0})}{\mathcal{R}_{0}(M;\nu)}
    -R_{1}(n,\varepsilon).
    \end{split}
\label{eq: lower bound of marginal likelihood ratio}
\end{equation}

\textbf{Third step: evaluation of $\inf_{p\in\mathcal{M}(\{\hat{p}_{j}\})}D(p,p^{\varepsilon}_{\theta,\nu})$.}
We begin by measuring the $\ell_{1}$ distances between the following marginal densities on $[I_{j}(\varepsilon)]$s:
\begin{enumerate}
    \item $\hat{p}_{j}$: the marginal density of $\hat{p}$ that only depends on $N_{j}$;
    \item $p^{*}_{j}$: the marginal density of $p^{*}$;
    \item $p^{\varepsilon}_{j,\theta,\nu}$: the marginal density of $p^{\varepsilon}_{\theta,\nu}$ given by
    \[
    p^{\varepsilon}_{j,\theta,\nu}(i_{j}):=
    \sum_{
    k\neq j,
    i_{k}\in[I_{k}(\varepsilon)]
    }
    p^{\varepsilon}_{\theta,\nu}(i);
    \]
    \item $p^{\varepsilon}_{j,\theta_{0},\nu_{0}}$: the marginal density of $p^{\varepsilon}_{\theta_{0},\nu_{0}}$.
\end{enumerate}
Let $\delta>0$.
From the union bound and Lemma \ref{lem: L1 and KL deviation inequality},
we have
\begin{align}
    &\Pr\left(
    \max_{j=1,\ldots,d}\|\hat{p}_{j}-p^{*}_{j}\|_{1}\le \delta
    \ ,\ 
    \max_{j=1,\ldots,d}D(\hat{p}_{j},p^{*}_{j}) \le \delta^{2}
    \right)
    \nonumber\\
    &\ge 
    1-\sum_{j=1}^{d}\Pr\left(
    \|\hat{p}_{j}-p^{*}_{j}\|_{1}> \delta
    \right)
    -\sum_{j=1}^{d}\Pr\left(
    D(\hat{p}_{j},p^{*}_{j})> \delta^{2}
    \right)
    \nonumber\\
    &\ge 1-4\sum_{j=1}^{d}\exp\left\{-\frac{n\delta^{2}}{2I_{j}(\varepsilon)}\right\}.
    \label{eq: bound on l1 between phat and pstar}
\end{align}
Note that we have
\begin{align*}
    p^{*}_{j}(i_{j})
    =\int_{\mathcal{X}_{1}\times \cdots\times D_{j,i_{j}}^{\varepsilon}\times\cdots\times \mathcal{X}_{d}}p(x;\theta_{0},\nu_{0})\mathrm{d}x
    =
    \sum_{
    k\neq j,
    i_{k}\in[I_{k}(\varepsilon)]
    }\int_{D_{i}^{\varepsilon}}p(x;\theta_{0},\nu_{0})\mathrm{d}x.
\end{align*}
This gives
\begin{align}
\max_{j=1,\ldots,d}
    \left|\log \frac{p^{*}_{j}(i_{j})}{p^{\varepsilon}_{j,\theta_{0},\nu_{0}}(i_{j})} \right|
    &\le 2L\varepsilon
    \label{eq: bound on log between pstar and peptheta}
\end{align}
and
\begin{align*}
    |p^{*}_{j}(i_{j})-p^{\varepsilon}_{j,\theta_{0},\nu_{0}}(i_{j})|
    &\le
    p^{*}_{j}(i_{j})\max\left\{1-\mathrm{e}^{-2L\varepsilon},\mathrm{e}^{2L\varepsilon}-1\right\}.
\end{align*}
The latter inequality, together with the inequalities $\exp(x)-1<2x$ for $0<x<1$,
and $1-\exp(-x)<x$ for $0<x<1$,
gives
\begin{align}
\max_{j=1,\ldots,d}\|p^{*}_{j}-p^{\varepsilon}_{j,\theta_{0},\nu_{0}}\|_{1}
    &\le 4L\varepsilon,
    \label{eq: bound on l1 between pstar and peptheta0}
\end{align}
provided that $4L\varepsilon<1$.
Further,
letting $p^{*}_{j,\nu}(i_{j}):=\int_{D^{\varepsilon}_{j,i_{j}}}r_{j}(x_{j};\nu)\mathrm{d}x_{j}$,
a similar argument gives
\begin{align}
\max_{j=1,\ldots,d}\|p^{*}_{j}-p^{\varepsilon}_{j,\theta,\nu}\|_{1}
    &\le 4L\varepsilon+
    \max_{j=1,\ldots,d}
    \sum_{i_{j}\in I_{j}(\varepsilon)}
    \left|\int_{ D_{j,i_{j}}^{\varepsilon}}
    \{r_{j}(x_{j};\nu_{0})-r_{j}(x_{j};\nu)\}
    \mathrm{d}x_{j}
    \right|
    \nonumber\\
    &
    \le 4L\varepsilon+
    \max_{j=1,\ldots,d}
    \sum_{i_{j}\in I_{j}(\varepsilon)}
    \int_{ D_{j,i_{j}}^{\varepsilon}}
    |r_{j}(x_{j};\nu_{0})-r_{j}(x_{j};\nu)|
    \mathrm{d}x_{j}
    \nonumber\\
    &= 4L\varepsilon+
    \max_{j=1,\ldots,d}
    \|r_{j}(\cdot;\nu_{0})-r_{j}(\cdot;\nu)\|_{1},
    \label{eq: bound on l1 between pstar and peptheta}
\end{align}
provided that $4L\varepsilon<1$.

We next employ Lemma \ref{lem: Stability} to evaluate $\inf_{p\in\mathcal{M}(\{\hat{p}_{j}\})}D(p,p^{\varepsilon}_{\theta,\nu})$.
First, consider 
\[\inf_{p\in\mathcal{M}(\{\hat{p}_{j}\})}D(p,p^{\varepsilon}_{\theta_{0},\nu_{0}}).\]
Let $\check{p}_{\theta_{0}}$ be the density on $I(\varepsilon)$ maximizing 
\begin{align*}
    \sum_{i\in I(\varepsilon)}
    H^{\varepsilon}_{\theta_{0}}(y^{\varepsilon}_{i})p(i)-\sum_{i=(i_{1},\ldots,i_{d})\in I(\varepsilon)}p(i)\log \frac{p(i)}{\prod_{j=1}^{d}\hat{p}_{j}(i_{j})}
\end{align*}
such that $p\in \mathcal{M}(\{\hat{p}_{j}\})$.
From Lemma \ref{lem: relation to EOT}, we have
\begin{align}
    \inf_{p\in\mathcal{M}(\{\hat{p}_{j}\})}D(p,p^{\varepsilon}_{\theta_{0},\nu_{0}})
    &=D(\check{p}_{\theta_{0}},p^{\varepsilon}_{\theta_{0},\nu_{0}}).
    \label{eq: inf D}
\end{align}
Applying Lemma \ref{lem: Stability} to equation (\ref{eq: inf D}) gives 
\begin{align*}
    D(\check{p}_{\theta_{0}},p^{\varepsilon}_{\theta_{0},\nu_{0}})
    &\le 
    (4d+2)\max_{i\in I(\varepsilon)}
    |H^{\varepsilon}_{\theta_{0}}(y_{i}^{\varepsilon})|
    \sum_{j=1}^{d}\|\hat{p}_{j}-p^{\varepsilon}_{j,\theta_{0},\nu_{0}}\|_{1}
    +\sum_{j=1}^{d}D(\hat{p}_{j},p^{\varepsilon}_{j,\theta_{0},\nu_{0}})
    \nonumber\\
    &\le
    (4d+2)d\max_{i\in I(\varepsilon)}
    |H^{\varepsilon}_{\theta_{0}}(i)|(\delta + 4L\varepsilon)
    +d(\delta^{2}+2L\varepsilon)
\end{align*}
with probability at least 
$1-4\sum_{j=1}^{d}\exp\{-n\delta^{2}/(2I_{j}(\varepsilon))\}$,
provided that $4L\varepsilon<1$.
Since for the diameter $\mathrm{diam}(\mathcal{X})$ of $\mathcal{X}$, we have
\begin{align*}
    \max_{i\in I(\varepsilon)}
    |H^{\varepsilon}_{\theta_{0}}(y_{i}^{\varepsilon})|
    \le  \xi\{L\varepsilon + L_{h} \mathrm{diam}(\mathcal{X})\},
\end{align*}
we obtain
\begin{align}
    \Pr\left(\inf_{p\in\mathcal{M}(\{\hat{p}_{j}\})}D(p,p^{\varepsilon}_{\theta_{0},\nu_{0}})
    \le
    U
    (\delta + \delta^{2}+6L\varepsilon)
    \right)
    \ge 
    1-4\sum_{j=1}^{d}\exp\left\{-\frac{n\delta^{2}}{2I_{j}(\varepsilon)}\right\}
    \label{eq: upper bound of inf D}
\end{align}
with 
\[U:=(4d+2)d\xi
    (1/4+L_{h}\mathrm{diam}(\mathcal{X}))+d,
\]
provided that $4L\varepsilon<1$.
Second, consider $\inf_{p\in\mathcal{M}(\{\hat{p}_{j}\})}D(p,p^{\varepsilon}_{\theta,\nu})$. A similar argument gives 
\begin{align*}
    D(\check{p}_{\theta},p^{\varepsilon}_{\theta,\nu})
    &\le 
    (4d+2)\xi\{L\varepsilon + L_{h} \mathrm{diam}(\mathcal{X})\}
    \sum_{j=1}^{d}\|\hat{p}_{j}-p^{\varepsilon}_{j,\theta,\nu}\|_{1}
    +\sum_{j=1}^{d}D(\hat{p}_{j},p^{\varepsilon}_{j,\theta,\nu}).
\end{align*}
Here we have
\begin{align*}
    \sum_{j=1}^{d}\|\hat{p}_{j}-p^{\varepsilon}_{j,\theta,\nu}\|_{1}
    \le 4Ld\varepsilon+
    d\delta
    +\sum_{j=1}^{d}\|r_{j}(\cdot;\nu_{0})-r_{j}(\cdot;\nu)\|_{1}
\end{align*}
and we have
\begin{align*}
    \sum_{j=1}^{d}D(\hat{p}_{j},p^{\varepsilon}_{j,\theta,\nu})
    &\le
    2Ld\varepsilon+\sum_{j=1}^{d}
    D(\hat{p}_{j},p_{j,\nu})
    \\
    &\le 2Ld\varepsilon+\sum_{j=1}^{d}D(p^{*}_{j},p_{j,\nu})
    +d\delta^{2}
\end{align*}
with probability at least $1-4\sum_{j=1}^{d}\exp\{-n\delta^{2}/(2I_{j}(\varepsilon))\}$,
where 
\[p_{j,\nu}(i_{j}):=\int_{\mathcal{X}_{1}\times\cdots\times D^{\varepsilon}_{j,i_{j}}\times \cdots \times \mathcal{X}_{d}} p(x;\theta,\nu)\mathrm{d}x\]
and Lemma \ref{lem: L1 and KL deviation inequality} is used.
Since the data processing inequality gives
\begin{align*}
    D(p^{*}_{j},p_{j,\nu})\le D(r_{j}(\cdot;\nu_{0}),r_{j}(\cdot;\nu)),
\end{align*}
we obtain
\begin{align}
\begin{aligned}
    &\Pr\left(\inf_{p\in\mathcal{M}(\{\hat{p}_{j}\})}D(p,p^{\varepsilon}_{\theta,\nu})
    \le
    U
    \tilde{D}(\nu_{0},\nu)
    +
    U
    (\delta + \delta^{2}+6L\varepsilon)
    \right)\\
    &\ge 
    1-4\sum_{j=1}^{d}\exp\left\{-\frac{n\delta^{2}}{2I_{j}(\varepsilon)}\right\},
    \label{eq: upper bound of inf D theta nu}
\end{aligned}
\end{align}
provided that $4L\varepsilon<1$.

\textbf{Final step: evaluation of remaining terms.}
We finally combine inequalities 
(\ref{eq: upper bound of marginal likelihood ratio}), 
(\ref{eq: lower bound of marginal likelihood ratio}), and
(\ref{eq: upper bound of inf D})
to complete the proof.
Assume that $4L\varepsilon<1$.
Inequalities (\ref{eq: upper bound of marginal likelihood ratio}), 
(\ref{eq: lower bound of marginal likelihood ratio}), and
(\ref{eq: upper bound of inf D}) yield
\begin{align*}
    \left|\frac{1}{n}\log \frac{g^{\varepsilon}(M;\theta_{0},\nu_{0})}{g^{\varepsilon}(M;\theta,\nu)}\right|
    \le 
    \frac{1}{n}\log \frac{\mathcal{R}_{0}(M;\nu_{0})}{\mathcal{R}_{0}(M;\nu)}
    +
    R_{1}(n,\varepsilon) + U
    \left\{
    \tilde{D}(\nu_{0},\nu)+
    \delta + \delta^{2}+6L\varepsilon
    \right\}
\end{align*}
with probability at least 
$1-4\sum_{j=1}^{d}\exp\{-n\delta^{2}/(2I_{j}(\varepsilon))\}$.

Consider an upper bound of $(1/n)\log \mathcal{R}_{0}(M;\nu_{0})/\mathcal{R}_{0}(M;\nu)$.
Observe that, for $t=1,\ldots,n$, 
\begin{align*}
    \frac{\int_{D^{\varepsilon}_{j,i_{j}(x(t))}}
    r_{j}(\tilde{x}_{j};\nu)\mathrm{d}\tilde{x}_{j}
    }
    {\int_{D^{\varepsilon}_{j,i_{j}(x(t))}}
    r_{j}(\tilde{x}_{j};\nu_{0})\mathrm{d}\tilde{x}_{j}
    }
    &=
    \frac{\int_{D^{\varepsilon}_{j,i_{j}(x(t))}}
    r_{j}(x_{j}(t);\nu)
    \frac{r_{j}(\tilde{x}(t);\nu)}{r_{j}(x_{j}(t);\nu)}
    \mathrm{d}\tilde{x}_{j}
    }
    {\int_{D^{\varepsilon}_{j,i_{j}(x(t))}}
    r_{j}(x_{j}(t);\nu_{0})
    \frac{r_{j}(\tilde{x}(t);\nu_{0})}{r_{j}(x_{j}(t);\nu_{0})}
    \mathrm{d}\tilde{x}_{j}
    }\\
    &\le e^{2L_{r_{j}}\varepsilon}
    \frac{
    r_{j}(x_{j}(t);\nu)
    }
    {
    r_{j}(x_{j}(t);\nu_{0})
    }.
\end{align*}
Under the assumption that for $i=1,\ldots,d$, the marginal density $r_{i}(x_{i};\nu)$ is log-Lipschitz continuous with Lipschitz constant $L_{r_{i}}$, where $\sup_{\nu\in \mathrm{N}}L_{r_{i},\nu}<\infty$,
this gives
\begin{align*}
    \frac{1}{n}\log \frac{\mathcal{R}_{0}(M;\nu_{0})}{\mathcal{R}_{0}(M;\nu)}
    \le 2\sum_{j=1}^{d}L_{r_{j}}\varepsilon.
\end{align*}

Consider an upper bound of $R_{1}(n,\varepsilon)$.
The inequality $I_{j}(\varepsilon)<\kappa\varepsilon^{-\alpha}$ gives
\begin{align*}
    \frac{\log |\mathcal{T}(\varepsilon)|}{n}
    &\le \kappa^{d}\varepsilon^{-d\alpha}
    \frac{\log (n+1)}{n}
    \ \text{and}\\
    \frac{\sum_{j=1}^{d}I_{j}(\varepsilon)}{2}\frac{\log (C_{1}C_{2}n^{2})}{n}
    &\le
    C_{3}\varepsilon^{-\alpha}\frac{\log n}{n},
\end{align*}
where 
$C_{3}$ is a positive constant independent of $n$ and $\varepsilon$.
Then, 
provided that $\varepsilon^{-1}$ is upper-bounded by a polynomial of $n$,
we have
\begin{align*}
    \left|\frac{1}{n}\log \frac{g^{\varepsilon}(M;\theta_{0},\nu_{0})}{g^{\varepsilon}(M;\theta,\nu)}\right|
    \le
    C_{4}\varepsilon^{-d\alpha}\frac{\log n}{n}
    +
    U\tilde{D}(\nu_{0},\nu)
    +
    \left(6UL+2\sum_{j=1}^{d}L_{r_{j}}\right)\varepsilon+
    U(\delta+\delta^{2})
\end{align*}
with probability at least $1-4d\exp\{-n\varepsilon^{\alpha}\delta^{2}/(2\kappa)\}$,
where $C_{4}$ is a positive constant independent of $n$ and $\varepsilon$.
Setting $\varepsilon=(\log n / n)^{1/(d\alpha+1)}$
gives
\begin{align*}
    \Ep\left[\left|\frac{1}{n}\log \frac{g^{\varepsilon}(M;\theta_{0},\nu_{0})}{g^{\varepsilon}(M;\theta,\nu)}\right|\right]
    &=
    \int_{0}^{\infty} \Pr\left(
    \left|\frac{1}{n}\log \frac{g^{\varepsilon}(M;\theta_{0},\nu_{0})}{g^{\varepsilon}(M;\theta,\nu)}\right|> \tilde{\delta}
    \right)\mathrm{d}\tilde{\delta}\\
    &\le
    C_5
    \left[\tilde{D}(\nu_{0},\nu)
    + \varepsilon_{n}
    \right],
\end{align*}
where $C_{5}$ is a positive constant independent of $n$.
This completes the proof.

\end{proof}

\subsection{Proof of Corollary~\ref{cor:clemle-negligible}}
\label{subsection:clemle-proof}

\begin{proof}

We employ Corollary 3.2.3 (i) of \cite{vanderVaartandWellner} after checking that all conditions therein are satisfied.
Let $\hat{\theta}^{(m)}$ be MLE of $\theta$ given $\nu=\nu_{0}$.
Let $\mathbb{M}_{n}(\theta):=\sum_{t=1}^{n}\log p(x(t);\theta,\nu_{0})/n$
and
let $\mathbb{M}(\theta):=\Ep[\mathbb{M}_{n}(\theta)]$, where $\Ep$ is the expectation with respect to $p(x;\theta_{0},\nu_{0})$.
We denote by $\Ep_{n}$ the expectation with respect to the empirical distribution.

We first show 
\begin{align*}
    \mathbb{M}_{n}(\hat{\theta})\ge \sup_{\theta\in\Theta}\mathbb{M}_{n}(\theta)-o_{P}(1).
\end{align*}
Let $r_{n}(\theta)$ be defined as
\begin{align*}
r_{n}(\theta):=
 \frac{1}{n}
    \left\{
    \log \frac{\prod_{t=1}^{n}p(x(t);\theta,\nu_{0})}{\prod_{t=1}^{n}p(x(t);\theta_{0},\nu_{0})}
    -\log \frac{f(\pi\mid M;\theta,\nu_{0})}{f(\pi\mid M;\theta_{0},\nu_{0})}
    \right\}.
 \end{align*}
From the definition of MLE,
we have
\begin{align*}
    \frac{1}{n}
    \log \frac{\prod_{t=1}^{n}p(x(t);\hat{\theta}^{(m)},\nu_{0})}{\prod_{t=1}^{n}p(x(t);\theta_{0},\nu_{0})}
    \ge 
    \frac{1}{n}
    \log \frac{\prod_{t=1}^{n}p(x(t);\hat{\theta},\nu_{0})}{\prod_{t=1}^{n}p(x(t);\theta_{0},\nu_{0})}.
\end{align*}
From the definition of CLE,
we have
\begin{align*}
    \frac{1}{n}
    \log \frac{\prod_{t=1}^{n}p(x(t);\hat{\theta}^{(m)},
    \nu_{0})}{\prod_{t=1}^{n}p(x(t);\theta_{0},\nu_{0})}
    &=
    \frac{1}{n}\log \frac{f(\pi\mid M; \hat{\theta}^{(m)})}{f(\pi\mid M; \theta_{0})}
    +r_{n}(\hat{\theta}^{(m)})
    \\
    &\le 
    \frac{1}{n}\log \frac{f(\pi\mid M; \hat{\theta})}{f(\pi\mid M; \theta_{0})}
    +r_{n}(\hat{\theta}^{(m)})
    \\
    &=
    \frac{1}{n}
    \log \frac{\prod_{t=1}^{n}p(x(t);\hat{\theta},\nu_{0})}{\prod_{t=1}^{n}p(x(t);\theta_{0},\nu_{0})}
    -r_{n}(\hat{\theta})
    +r_{n}(\hat{\theta}^{(m)}).
\end{align*}
Combining these yields
\begin{align*}
    \left| 
    \frac{1}{n}
    \log \frac{\prod_{t=1}^{n}p(x(t);\hat{\theta},\nu_{0})}{\prod_{t=1}^{n}p(x(t);\theta_{0},\nu_{0})}
    -\frac{1}{n}
    \log \frac{\prod_{t=1}^{n}p(x(t);\hat{\theta}^{(m)},\nu_{0})}{\prod_{t=1}^{n}p(x(t);\theta_{0},\nu_{0})}
    \right|
    \le 2\sup_{\theta\in\Theta}r_{n}(\theta).
\end{align*}
This, together with Theorem \ref{theorem:negligible},
gives
\begin{align*}
    \left|\frac{1}{n}\log \frac{\prod_{t=1}^{n} p(x(t);\hat{\theta}^{(m)},\nu_{0})}
    {\prod_{t=1}^{n} p(x(t);\hat{\theta},\nu_{0})}
    \right|\to 0 \,\,\text{in probability},
\end{align*}
which implies
\begin{align*}
    \mathbb{M}_{n}(\hat{\theta})\ge \sup_{\theta\in\Theta}\mathbb{M}_{n}(\theta)-o_{P}(1).
\end{align*}
%which, together with Theorem \ref{theorem:negligible}, completes the proof.

Second, we show that $\sup_{\theta\in\Theta}|\mathbb{M}_{n}(\theta)-\mathbb{M}(\theta)|\to 0$ in probability.
Observe that
\begin{align*}
    &\mathbb{M}_{n}(\theta)-\mathbb{M}(\theta)\\
    &=
    \underbrace{
    \left\{\Ep_{n}\theta^{\top}h(X)-\Ep[\theta^{\top}h(X)]\right\}}_{=:T_{1}}
    -
    \underbrace{
    \left\{\Ep_{n}\sum_{t=1}^{d}a_{i}(X_{i};\theta,\nu_{0})
    -
    \Ep \sum_{t=1}^{d}a_{i}(X_{i};\theta,\nu_{0})\right\}
    }_{=:T_{2}}.
\end{align*}
Since $h$ is bounded, $\sup_{\theta\in\Theta}|T_{1}|\to 0$ in probability.
From Assumption \ref{assumption: cor 4} (1), we get $\sup_{\theta\in\Theta}|T_{2}|\to 0$ in probability.
Thus, we obtain
$\sup_{\theta\in\Theta}|\mathbb{M}_{n}(\theta)-\mathbb{M}(\theta)|\to 0$ in probability.

Finally, 
we shall show
\begin{align}
\Ep[\log p(X;\theta_{0},\nu_{0})]>\Ep[\log p(X;\theta,\nu_{0})]\quad\text{ for }\quad
\theta\ne\theta_{0}.
\label{eq: unique maximizer}
\end{align}
Consider 
$\theta^{*}\in\Theta$ satisfying 
$\{\Ep[\log p(X;\theta_{0},\nu_{0})]-\Ep[\log p(X;\theta^{*},\nu_{0})]\}=0$.
The Jensen inequality implies 
that for such  $\theta^{*}$, 
$\{\log p(x;\theta_{0},\nu_{0})-\log p(x;\theta^{*},\nu_{0})\}=0$ , $x$-a.e..
However,
since $h(x)$ is linearly independent modulo additive functions
and since we have
\begin{align*}
&\log p(x;\theta_{0},\nu_{0})-\log p(x;\theta,\nu_{0})\\
&=(\theta_{0}-\theta)^{\top}h(x)
-(\psi(\theta_{0},\nu_{0})-\psi(\theta,\nu_{0}))
-\sum_{i=1}^{d}\{a_{i}(x_{i};\theta_{0},\nu_{0})-a_{i}(x_{i};\theta,\nu_{0})\},
\end{align*}
such $\theta^{*}$ is equal to $\theta_{0}$, which shows (\ref{eq: unique maximizer}).
So, the maximizer of $\mathbb{M}(\theta)$ is uniquely determined.
Then, applying Corollary 3.2.3 (i) of \cite{vanderVaartandWellner}, we obtain $\hat{\theta}\to\theta_{0}$ in probability, which completes the proof.
\end{proof}

\subsection{Proof of Theorem \ref{theorem:consistencyofPL}}
\label{subsection: proof of theorem:consistencyofPL}

\begin{proof}

The proof consists of three steps:
re-writing the pseudo likelihood as a $U$-statistics (in particular, re-writing its maximizer as an $M_{m=2}$-estimator; \cite{BoseChatterjee_2018}),
deriving the Fisher consistency of the pseudo likelihood,
showing the consistency of the PLE.

\textbf{First step: the pseudo likelihood as a $U$-statistics.}
Observe that 
the log pseudo likelihood normalized by $n(n-1)/2$ is written in a form of $U$-statistics as
\begin{align}
\mathrm{PL}(\theta)&:=
\frac{2}{n(n-1)}\sum_{1\le s<t\le n}\left(
\sum_{i=1}^{d}\log 
\frac{1}
{
1+e^{\theta^{\top}(h_{*}(\pi\circ \tau^{i}_{st})-h_{*}(\pi))}
}
\right)
\nonumber\\
&=-\frac{1}{n(n-1)}\sum_{1\le s\ne t\le n}
\left(\sum_{i=1}^{d}
\log \left\{1+e^{\theta^{\top} u^{i}(X(s),X(t))}\right\}
\right),
\end{align}
where, for $x,\tilde{x}\in \mathcal{X}$,
\[u^{i}(x,\tilde{x}):=
h(x^{(i)})+h(\tilde{x}^{(i)})-h(x)-h(\tilde{x})
\] with $x^{(i)}=(x_{1},\ldots,\tilde{x}_{i},\ldots,x_{d})$,
and
$\tilde{x}^{(i)}=(\tilde{x}_{1},\ldots,x_{i},\ldots,\tilde{x}_{d})$
and we used the symmetry of $u^{i}(x,\tilde{x})$.
From this, we can employ properties of $U$-statistics to prove the consistency of the pseudo likelihood estimator.

\textbf{Second step: the Fisher consistency of the pseudo likelihood.}
Next, we shall see that the pseudo likelihood estimator has the Fisher consistency:
\begin{align*}
\Ep^{2}[\mathrm{PL}(\theta)] \le \Ep^{2}[\mathrm{PL}(\theta_{0})] \quad \text{for}\quad \theta\in\Theta,
\end{align*}
where $\Ep^{2}$ denotes the expectation with respect to $p(x;\theta_{0},\nu_{0})p(\tilde{x};\theta_{0},\nu_{0})$.
To see this, we follow the arguments in the proof for Theorem 4 of \cite{chen_sei2022}.
Denote $p(x):=p(x;\theta_{0},\nu_{0})$
and
$p_{\theta}(x):=p(x;\theta,\nu_{0})$.
Then we have
\begin{align*}
\Ep^{2}[\mathrm{PL}(\theta)]
&=-\sum_{i=1}^{d}\iint 
\log \left\{1+e^{\theta^{\top} u^{i}(x,\tilde{x})}\right\}
p(x)
p(\tilde{x})
\mathrm{d}x\mathrm{d}\tilde{x}
\\
&=-\frac{1}{2}\sum_{i=1}^{d} \iint 
\log \left\{1+e^{\theta u^{i}(x,\tilde{x})}\right\}
p(x)p(\tilde{x})
\mathrm{d}x\mathrm{d}\tilde{x}
\\
&\quad -\frac{1}{2}\sum_{i=1}^{d} \iint 
\log \left\{1+e^{\theta u^{i}(x^{(i)},\tilde{x}^{(i)})}\right\}
p(x^{(i)})p(\tilde{x}^{(i)})
\mathrm{d}x^{(i)}\mathrm{d}\tilde{x}^{(i)}.
\end{align*}
Letting 
\begin{align*}
\sigma^{i}(x,\tilde{x};\theta):=
\frac{p_{\theta}(x)p_{\theta}(\tilde{x})}{p_{\theta}(x)p_{\theta}(\tilde{x})+p_{\theta}(x^{(i)})p_{\theta}(\tilde{x}^{(i)})}
&=\frac{1}{1+e^{\theta^{\top} u^{i}(x,\tilde{x}) }},
\end{align*}
we have
\begin{align*}
\Ep^{2}[\mathrm{PL}(\theta)]
&=
\frac{1}{2}
\sum_{i=1}^{d}\iint 
\log \left\{\sigma^{i}(x,\tilde{x};\theta)
\right\}
p(x)p(\tilde{x})
\mathrm{d}x\mathrm{d}\tilde{x}
\\
&\quad 
+
\frac{1}{2}
\sum_{i=1}^{d}\iint 
\log 
\left\{
\sigma^{i}(x^{(i)},\tilde{x}^{(i)};\theta)
\right\}
p(x^{(i)})p(\tilde{x}^{(i)})
\mathrm{d}x^{(i)}\mathrm{d}\tilde{x}^{(i)}
\end{align*}
and then we get
\begin{align*}
\Ep^{2}[\mathrm{PL}(\theta_{0})]-\Ep^{2}[\mathrm{PL}(\theta)]
&=
\frac{1}{2}\sum_{i=1}^{d}\iint \log 
\left\{\frac{\sigma^{i}(x,\tilde{x};\theta_{0})}{\sigma^{i}(x,\tilde{x};\theta)}\right\} 
p(x)p(\tilde{x})
\mathrm{d}x\mathrm{d}\tilde{x}\\
&\quad +\frac{1}{2}\sum_{i=1}^{d}\iint \log 
\left\{\frac{1-\sigma^{i}(x,\tilde{x};\theta)}{1-\sigma^{i}(x,\tilde{x};\theta_{0})}
\right\}
p(x^{(i)})p(\tilde{x}^{(i)})
\mathrm{d}x^{(i)}\mathrm{d}\tilde{x}^{(i)}.
\end{align*}
Here 
observe that 
\begin{align*}
p(x)p(\tilde{x})&=\sigma^{i}(x,\tilde{x};\theta_{0})\{p(x)p(\tilde{x})+p(x^{(i)})p(\tilde{x}^{(i)})\}\quad\text{and}\\
p(x^{(i)})p(\tilde{x}^{(i)})
&=\{1-\sigma^{i}(x,\tilde{x};\theta_{0})\}\{p(x)p(\tilde{x})+p(x^{(i)})p(\tilde{x}^{(i)})\}.
\end{align*}
Together with the fact that
the integration with respect to $\mathrm{d}x\mathrm{d}\tilde{x}$ is the same as that with respect to $\mathrm{d}x^{(i)}\mathrm{d}\tilde{x}^{(i)}$,
this yields
\begin{align*}
&\Ep^{2}[\mathrm{PL}(\theta_{0})]-\Ep^{2}[\mathrm{PL}(\theta)]
\\
&=
\sum_{i=1}^{d}\iint 
\underbrace{\left\{\sigma^{i}(x,\tilde{x};\theta_{0})
\log \frac{\sigma^{i}(x,\tilde{x};\theta_{0})}{\sigma^{i}(x,\tilde{x};\theta)}
+
(1-\sigma^{i}(x,\tilde{x};\theta_{0}))
\log 
\frac{1-\sigma^{i}(x,\tilde{x};\theta_{0})}{1-\sigma^{i}(x,\tilde{x};\theta)}
\right\}}_{\ge 0} \mu_{p}(\mathrm{d}x\mathrm{d}\tilde{x})\ge 0,
\end{align*}
with $\mu_{p}(\mathrm{d}x\mathrm{d}\tilde{x}):=
\{p(x)p(\tilde{x})+p(x^{(i)})p(\tilde{x}^{(i)})\}\mathrm{d}x\mathrm{d}\tilde{x}/2
$, which proves the Fisher consistency.

\textbf{Final step: the consistency of PLE.}
We shall see that
the PLE converges to $\theta_{0}$ in probability.
Together with the Fisher consistency of the pseudo likelihood
and the equation
\begin{align*}
\nabla^{2}_{\theta}\Ep^{2}[\mathrm{PL}(\theta)]
=\frac{1}{4}\Ep^{2}\left[
\sum_{i=1}^{d}
\frac{u^{i}(X,\tilde{X})(u^{i}(X,\tilde{X}))^{\top}}
{\{\cosh(\theta^{\top} u^{i}(X,\tilde{X})/2)
\}^{2}}
\right],
\end{align*}
the first additional assumption implies 
\begin{align*}
\sup_{\theta\notin G}\Ep^{2}\mathrm{PL}(\theta)<\Ep^{2}\mathrm{PL}(\theta_{0}) 
\,\text{for any open neighbor $G$ of $\theta_{0}$}.
\end{align*}
So, from the same argument as in Corollary 3.2.3 (i) of \cite{vanderVaartandWellner}, it
suffices to show that 
\begin{align}
\sup_{\theta\in\Theta}|\Ep^{2}_{n}m_{\theta}(X,\tilde{X})-\Ep^{2}m_{\theta}(X,\tilde{X})|\to 0\quad \text{in probability},
\label{eq: GC for U}
\end{align}
where $\Ep_{n}^{2}$ denotes the expectation with respect to $\{1/(n(n-1))\}\sum_{1\le s\ne t\le n}\delta_{x(s),x(t)}$,
and
\begin{align*}
m_{\theta}(x,\tilde{x}):=
\sum_{i=1}^{d}\log \frac{1}{1+e^{\theta^{\top}u^{i}(x,\tilde{x})}}.
\end{align*}
From the Glivenko--Cantelli theorem for $U$-statistics (Corollary 5.2.5 of \cite{PenaGine}),
the sufficient condition for (\ref{eq: GC for U}) is that the bracketing number $\mathcal{N}_{[]}^{(1)}(\mathcal{H},\varepsilon)$
is finite for any $\varepsilon>0$,
where 
$\mathcal{H}:=\{m_{\theta}(\cdot,\cdot):\theta\in \Theta\}$
and
the bracketing number $\mathcal{N}_{[]}^{(1)}(\mathcal{H},\varepsilon)$ is defined as the minimum number $N$ for which there exist $f_{1}(\cdot,\cdot),\ldots,f_{N}(\cdot,\cdot)$
and $\triangle_{1}(\cdot,\cdot)\ge 0,\ldots,\triangle_{N}(\cdot,\cdot)\ge 0$
with $\Ep^{2}|f_{i}|<\infty$ and
$\Ep^{2}\triangle_{i}<\varepsilon$ for all $i\in\{1,\ldots,N\}$ such that for all $m\in \mathcal{H}$ there exists $i\in\{1,\ldots,N\}$ with $|f_{i}-m|<\triangle_{i}$.
Here, observe that 
\begin{align*}
|m_{\theta}(x,\tilde{x})-m_{\theta'}(x,\tilde{x})|
&\le 
\sup_{\theta\in\Theta}
\left\{\left\|\nabla_{\theta}m_{\theta}(x,\tilde{x})\right\|\right\} \|\theta-\theta'\|
\\
&=
\sup_{\theta\in\Theta}\left|\sum_{i=1}^{d}\frac{
\|u^{i}(x,\tilde{x})\|
}{1+e^{\theta^{\top}u^{i}(x,\tilde{x})}}\right|
\|\theta-\theta'\|
\\
&\le
\sum_{i=1}^{d}\|u^{i}(x,\tilde{x})\|
\|\theta-\theta'\|
\\
&\le 2L_{h}\sum_{i=1}^{d}d_{i}(x_{i},\tilde{x}_{i})\|\theta-\theta'\|
\\
&\le 2L_{h}\sum_{i=1}^{d}\{d_{i}(x_{i},x^{0}_{i})+d_{i}(\tilde{x}_{i},x^{0}_{i})\}\|\theta-\theta'\|
\end{align*}
where $L_{h}$ is the Lipschitz constant of $h$ as in Assumption 1 (3).
Then, 
together with Assumption 1(2) (boundedness of the parameter space),
using the second additional assumption  and Theorem 2.7.11 of \cite{vanderVaartandWellner} (bounding the bracketing number of  classes that are Lipschitz in parameter),
we have Equation (\ref{eq: GC for U}), which concludes the proof of the consistency of the pseudo likelihood estimator.
\end{proof}

\subsection{The asymptotic variance of the pseudo Likelihood estimators}
\label{appendix: PLE variance}

To calculate the asymptotic variance of PLE, 
we employ the following result on the asymptotic normality of $M_{m=2}$-estimators.
\begin{theorem}[Theorem 2.3 of \cite{BoseChatterjee_2018}]
Let $f_{\theta}(x,\tilde{x})$ be a measurable real valued function which is symmetric in $x$ and $\tilde{x}$.
Suppose the following hold:
\begin{itemize}
\item[M1] A function $f_{\theta}(x,\tilde{x})$ is convex in $\theta$;
\item[M2] The function $\Ep^{2}f_{\theta}(X,\tilde{X})$ is finite for all $\theta\in\Theta$;
\item[M3] The maximizer $\theta_{f}$ of $\Ep^{2}f_{\theta}(X,\tilde{X})$ exists uniquely;
\item[M4] For the gradient, we have $\Ep^{2}\|\nabla_{\theta}f_{\theta}(X,\tilde{X})\|^{2}<\infty$ for all $\theta\in\Theta$;
\item[M5] The Hessian $J:=\Ep^{2}[\nabla^{2}_{\theta=\theta_{f}}f_{\theta}(X,\tilde{X})]$ exists and is positive definite. 
\end{itemize}
Then, for the maximizer $\hat{\theta}_{f}$ of $\Ep^{2}_{n}[f_{\theta}(X,\tilde{X})]$, the asymptotic distribution of $n^{1/2}(\hat{\theta}_{f}-\theta_{f})$ is 
\[
N(0,4 J^{-1} K J^{-1})
\quad\text{with}\quad
K:= \Ep_{\tilde{X}}[ \Ep_{X}[ \nabla_{\theta=\theta_{f}} f_{\theta}(X,\tilde{X})\mid \tilde{X}] 
\{\Ep_{X}[ \nabla_{\theta=\theta_{f}} f_{\theta}(X,\tilde{X})\mid \tilde{X}]\}^{\top}
],
\]
where let $\Ep_{X}[\,\cdot\,\mid \tilde{X}]$ denotes the conditional expectation of $X$ given $\tilde{X}$.
\end{theorem}

We apply this result to $f_{\theta}(x,\tilde{x})=m_{\theta}(x,\tilde{x})$ and only have to check that all conditions M1-M5 are satisfied.
Conditions M1 and M5 are satisfied from the first additional assumption in Theorem \ref{theorem:consistencyofPL}.
Condition M3 is satisfied from the argument in the proof of the consistency of PLE.
Conditions M2 and M4 are satisfied from the last argument in the proof of the consistency of PLE
and from the second additional assumption in Theorem \ref{theorem:consistencyofPL}. 
Thus we confirm the asymptotic normality of PLE, from which we obtain an estimator of the asymptotic variance.

\section{Applications to real data} \label{section:real}

This section provide
two applications of our method to real data.
The first is to the palmerpenguins dataset, which consists of continuous and categorical variables. The second is to an earthquake catalog dataset, which contains observations on the Stiefel manifold.

\subsection{Penguins data}

We apply our method to the palmerpenguins dataset provided by \cite{Horst_et_al_2020}. The data consist of 344 observations with 8 variables on Palmer Archipelago (Antarctica) penguins. We focus on the characteristics of Adelie penguins and ignore the information on islands and years. The variables are then five: (1) bill length, (2) bill depth, (3) flipper length, (4) body mass, and (5) sex. The first four variables are standardized and the last category variable is quantified as 0 (female) or 1 (male). Individuals with missing values are simply removed.

\begin{figure}[htb]
\centering
\includegraphics[width=7cm]{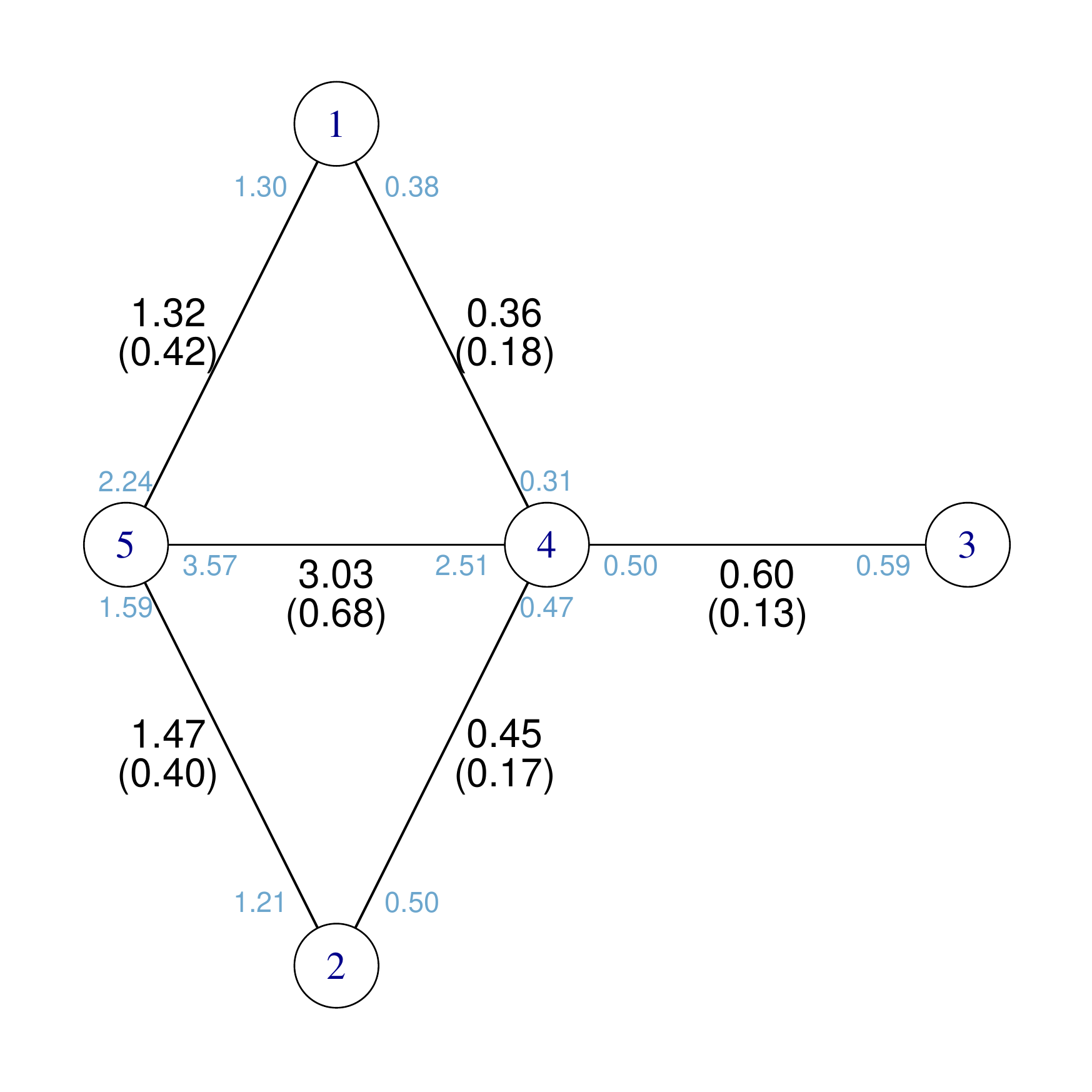}
\caption{Selected interactions for the palmerpenguins dataset. The vertices correspond to (1) bill length, (2) bill depth, (3) flipper length, (4) body mass, and (5) sex. 
The numbers on the edges are the parameter estimates and the standard errors, respectively. The numbers around the vertices indicate the results of regression analysis when each variable is treated as the response variable. The linear regression coefficients are divided by the residual variance for comparison.
}
\label{fig:penguin}
\end{figure}

\begin{table}[htb]
    \caption{
    \label{tab:palmerpenguins-3dim} Estimation of the three-dimensional interaction model for the palmerpenguins dataset.
    The first two rows are the estimates and the standard errors, respectively. The last row indicates the coefficients of linear regression on $x_2$ against $(x_4,x_5,x_4x_5)$, where the coefficients are divided by the residual variance.}
    \centering
    \fbox{%
    {\footnotesize
        \begin{tabular}{l|cccccc|cc}
        & $\theta_{14}$& $\theta_{15}$& $\theta_{24}$& $\theta_{25}$& $\theta_{34}$& $\theta_{45}$& $\theta_{145}$& $\theta_{245}$\\\hline
        estimate& 0.36& 1.30& 1.05& 1.28& 0.58& 2.97& -0.01& -0.78 \\
        (std.\ err.)& (0.25)& (0.38)& (0.33)& (0.39)& (0.12)& (0.63)& (0.32)& (0.39)\\
        \hline
        $x_2$, linear& & & 0.87& 1.25& & & & -0.58
        \end{tabular}
    }}
\end{table}

We first fit the second-order interaction model
$h(x) = (x_ix_j)_{1\leq i<j\leq 5}\in\mathbb{R}^{10}$
on $x=(x_j)\in\mathbb{R}^4\times \{0,1\}$.
We then performed backward stepwise variable selection based on the AIC values with Wald-type approximation. The selected interactions $I$ are shown in Figure~\ref{fig:penguin}.
The flipper length ($x_3$) depends only on the body mass ($x_4$), which is quite reasonable.

CLE of $\theta=(\theta_{ij})_{(i,j)\in I}$ and their standard errors are shown on the edges of the graph in Figure~\ref{fig:penguin}, respectively.
The numbers around the vertices indicate the regression coefficients based on the linear/logistic regression for each variable, where the linear regression coefficients are divided by the residual variance (dispersion parameter) to make the results comparable.
For example, $x_1$ is explained by $x_4$ and $x_5$, since $x_1$ interacts only with $x_4$ and $x_5$. The intercept term was estimated but is not shown.
The estimates of our model are close to the regression coefficients. The results are consistent with Example~\ref{example:mixed}
of Section \ref{subsection:examples}.

Since the obtained graph has two triangles 1--4--5 and 2--4--5, we estimated three-dimensional interactions among the variables by adding the statistics $x_1x_4x_5$ and $x_2x_4x_5$ into the model. Table~\ref{tab:palmerpenguins-3dim} shows the estimates together with the standard errors. There exists a weak negative three-dimensional interaction 2--4--5, which is also observed in the linear regression on $x_2$ against $(x_4,x_5,x_4x_5)$.

\subsection{Earthquake catalog}
\label{subsec: earthquake}
We lastly apply the minimum information dependence modeling to the 158 earthquake data that occurred in Japan during the period from January 1st, 2021 to December 8th, 2021,
by using the catalog provided by the Japan Meteorological Agency (\cite{JMA}).
The catalog consists of origin times (the start time of the earthquake rupture), locations of epicenters (longitudes, latitudes, and depths), magnitudes, strikes (the fault-trace directions), dips (the angles of faults), and rakes (directions of fault motions during ruptures with respect to strikes).
Using  the strike $\xi$, the dip $\delta$, and the rake $\lambda$, we can express the focal mechanism of an earthquake that describes the slip of the fault in the earthquake and the orientation of the fault on which it occurs.
The focal mechanism is described by the normal vector $n$ to one side of the fault blocks (a vector perpendicular to the fault plane) and the slip vector $u$ of the other block against it
(e.g.,\, Section 9 of \cite{Shearer2009}):
\begin{align*}
    u&=(\sin \xi \cos \delta \sin \lambda +\cos\xi\cos\lambda,-\cos\xi\cos\delta\sin\lambda+\sin\xi\cos\lambda,-\sin\delta\sin\lambda )^{\top},\\
    n&=(-\sin\xi\sin\delta,\cos\xi\sin\delta,-\cos\delta)^{\top}.
\end{align*}
However, there are two ambiguities in focal mechanism.
The first is that since the choice of the reference block is arbitrary, we cannot distinguish $(n,u)$ from $(-n,-u)$.
The second is that since the fault plane and the auxiliary plane (the plane perpendicular to the fault plane) cannot be separated from observations,  we cannot distinguish $(n,u)$ from $(u,n)$. We deal with these ambiguities by transforming two vectors into the compressional ($P$) and tensional ($T$) axes:
$P = 2^{-1/2} (n-u)$ and $T = 2^{-1/2} (n+u)$.
This pair $(P,T)$ forms an orthogonal 2-frame in $\mathbb{R}^{3}$ and is an element in the Stiefel manifold $V_{2}(\mathbb{R}^{3}):=\{(u_{1},u_{2}):u_{1}^{\top}u_{2}=0, u_{1}\in\mathbb{S}^{2},u_{2}\in\mathbb{S}^{2}\}$ (e.g.,\ \cite{ArnoldJupp2013,WalshArnoldTownend2009}),
where $\mathbb{S}^{2}$ is the unit sphere in $\mathbb{R}^{3}$.

\begin{figure}[tb]
\centering
\includegraphics[width=12cm]{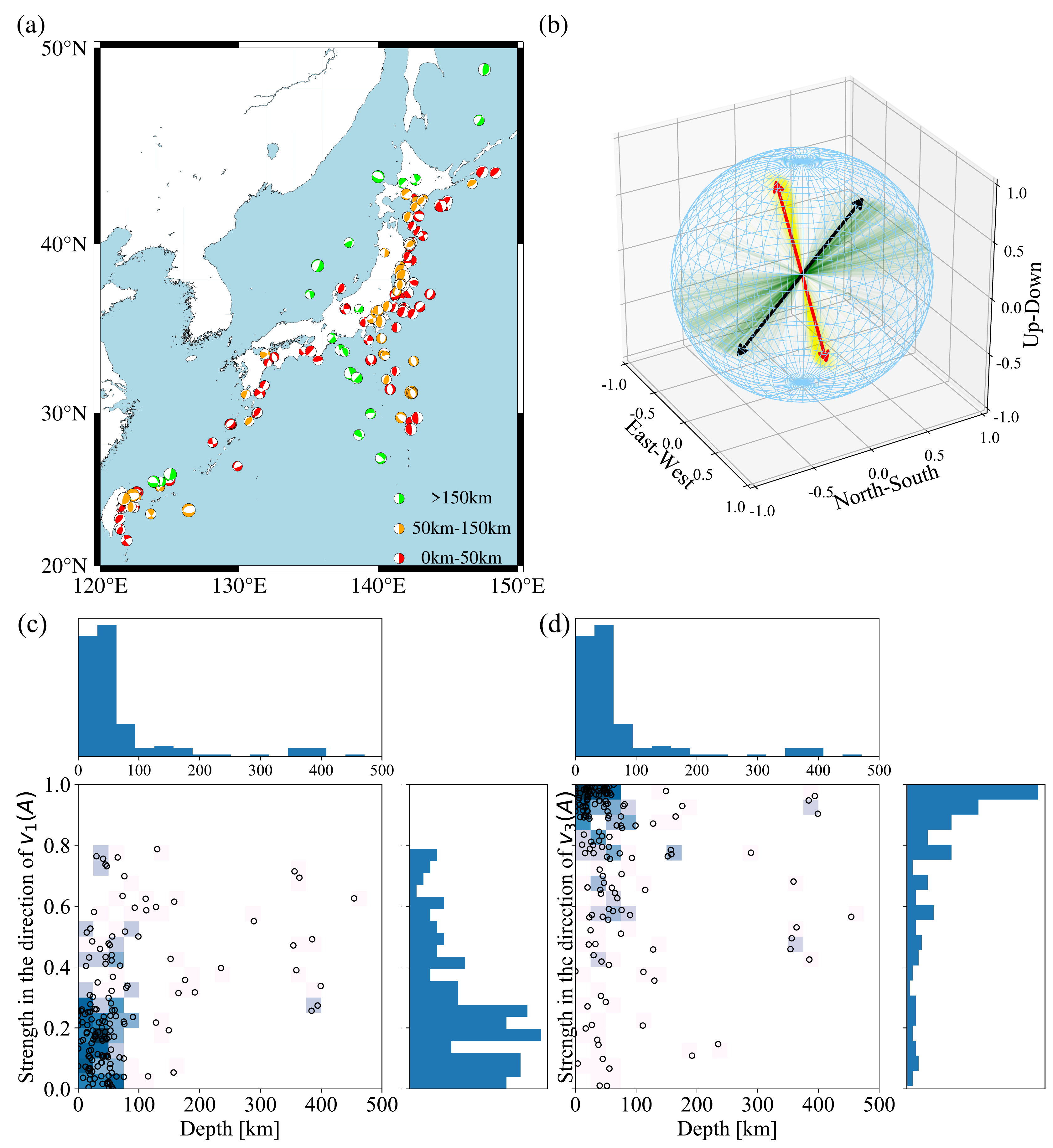} \caption{Application to the earthquake catalog. (a) Map of focal mechanisms. (b) Axes of the first and third eigenvectors of the estimated $A$ in (\ref{eq: min-info for mechanism and depth}). The axis of the first eigenvector is colored red with 1000 samples from the asymptotic distribution colored yellow.
The axis of the third eigenvector is colored black with 1000 samples from the asymptotic distribution colored green. (c,d) The two-dimensional histograms with marginal histograms of the depths and the strengths in the directions of the first (c) and third (d) eigenvectors of the estimated $A$,  where the data themselves are plotted as black circles.}
\label{fig:mechdist}
\end{figure}

Figure~\ref{fig:mechdist} (a) shows the integrated map of focal mechanisms in the ``beach ball'' representation; for  details of this representation, see, for example, Section 9 of \cite{Shearer2009}. 
The beach ball is the stereographic projection of the fault and the slip on the lower half of a sphere surrounding the source location of an earthquake.
The colored quadrants indicate outward motion away from those quadrants, and the white quadrants indicate inward motion towards those quadrants. Then, the $T$-axis is the line through the origin and the center of the colored quadrant, and the $P$-axis is that through the origin and the center of the white quadrant.

\begin{table}[h]
    \caption{
    \label{tab:A} The CLE $\hat{A}$ of $A$. Let $v_{1}(\hat{A})$ be the first eigenvector of $\hat{A}$, let $v_{2}(\hat{A})$ be the second eigenvector of $\hat{A}$, and let $v_{3}(\hat{A})$ be the third eigenvector of $\hat{A}$. The eigenvalues along with each eigenvector are $4.13$,  $0$, and $-1.09$, respectively.}
    \centering
    \fbox{
    \begin{tabular}{c|ccc}
         & $v_{1}(A)$ & $v_{2}(A)$& $v_{3}(A)$  \\\hline
        East-West & -0.27& 0.46& 0.84 \\
        North-South & 0.06& 0.88& -0.46 \\
        Up-Down & 0.96& 0.07& 0.26\\
    \end{tabular}}
\end{table}

In this application, we analyze the dependence between the focal mechanism $(P,T)$ and depth $z$ of earthquakes. To this end, we 
set $\mathcal{X}_{1}=V_{2}(\mathbb{R}^{3})$ and $\mathcal{X}_{2}=\mathbb{R}$,
and use the following canonical statistics:
$h^{(P)}_{ij}(z,P)=zP_{i}P_{j}, 1\le i,j\le 3,
    \quad\text{and}\quad
    h^{(T)}_{ij}(z,T)=zT_{i}T_{j}, 1\le i,j\le 3$,
where $h^{P}_{33}$ and $h^{T}_{33}$ are determined by 
$z$ and the constraints that $P_{1}^{2}+P_{2}^{2}+P_{3}^{2}=T_{1}^{2}+T_{2}^{2}+T_{3}^{2}=1$. 
We denote by $h^{(P)}(z,P)$ and $h^{(T)}(z,T)$ the matrices with $(i,j)$-components $h^{(P)}_{ij}(z,P)$
and $(i,j)$-components $h^{(T)}_{ij}(z,T)$, respectively. 
Then, the minimum information dependence model is described as 
\begin{align}
    \log p(z,P,T;A,B,\nu)=\mathrm{tr}(Ah^{(P)}(z,P))+\mathrm{tr}(Bh^{(T)}(z,T))-c,
    \label{eq: min-info for mechanism and depth}
\end{align}
where 
$A$ and $B$ are symmetric matrices,
and
$c:=a_{1}(z;A,B,\nu)+a_{2}(P,T;A,B,\nu)+\psi(A,B,\nu)$.
Since $\{h^{(P)}_{1,1},h^{(P)}_{1,2},h^{(P)}_{1,3},h^{(P)}_{2,2},h^{(P)}_{2,3}\}$ are linearly independent modulo additive functions,
the number of independent parameters in $A$ is $5$.
The same argument applies to $B$.
Thus, the total number of independent parameters in (\ref{eq: min-info for mechanism and depth}) is 10.

Here we report the inferential result.
We first conducted hypothesis testing with the null hypothesis
$H_{0}^{AB}: A=B=0$ using Wald statistics, which is a null hypothesis implying that the depth and the focal mechanism are independent.
The resulting $p$-value is $0.02$.
We next took a closer look at the dependence between the depth and the $P$-axis. Table \ref{tab:A} shows the value of the estimated $A$.
Figure \ref{fig:mechdist} (b) displays the first and the third eigenvectors $v_{1}$ and $v_{3}$ of the CLE of $A$.
We then projected the data onto the corresponding axes and obtained the joint histograms in Figures~\ref{fig:mechdist} (c) and (d). 
Figure~\ref{fig:mechdist} (c)
shows that the strength $|\langle v_{1}, P\rangle|$ in the direction of $v_{1}$ concentrated around the lower values for shallower depths, but
takes larger values for deeper depths, which might be consistent with the fact that intermediate-depth earthquakes at depths of about $50$km--$300$km, often have normal faults (i.e., the $P$-axes of these are relatively vertical) and most of the shallower earthquakes occur near the Japan trench and have reverse faults or strike-slip faults (i.e., the  $P$-axes of these are relatively horizontal).

\bibliographystyle{imsart-nameyear}
\bibliography{min-info}

\end{document}